\documentclass[reqno, a4paper,11pt]{amsart}

\usepackage[utf8]{inputenc} 
\usepackage{amsthm}
\usepackage{amsmath} 
\usepackage{mathtools} 
 \usepackage{afterpage}

\usepackage[english]{babel}
\usepackage{lmodern}
\makeatletter
\ifcase \@ptsize \relax
  \newcommand{\miniscule}{\@setfontsize\miniscule{4}{5}}%
\or
  \newcommand{\miniscule}{\@setfontsize\miniscule{5}{6}}%
\or
  \newcommand{\miniscule}{\@setfontsize\miniscule{5}{6}}%
\fi
\makeatother

\makeatletter
\ifcase \@ptsize \relax
  \newcommand{\nano}{\@setfontsize\miniscule{3}{4}}%
\or
  \newcommand{\nano}{\@setfontsize\miniscule{4}{5}}%
\or
  \newcommand{\nano}{\@setfontsize\miniscule{4}{5}}%
\fi
\makeatother

\newtheorem{theorem}{Theorem}[section]

 \newtheorem{prop}[theorem]{Proposition}
  \newtheorem{claim}[theorem]{Claim}
\newtheorem{cor}[theorem]{Corollary}
    
  \theoremstyle{definition}
\newtheorem{definition}[theorem]{Definition}
\newtheorem{example}[theorem]{Example}

\newcommand{\balita}{\raisebox{2pt}{\text{ \nano$\bullet$ }}} 
\newcommand{\bala}{\raisebox{1pt}{\text{ \tiny$\bullet$ }}} 
\theoremstyle{remark}
\newtheorem{remark}[theorem]{Remark}

\newtheorem*{acknowledgements}{Acknowledgements}

\makeatletter
\newcommand{\vast}{\bBigg@{4}}
\newcommand{\Vast}{\bBigg@{17.30}}
\makeatother

    \newcommand{\emutres}[3]{e_{\mu_{#1}}e_{\mu_{#2}}e_{\mu_{#3}}}
 \usepackage{textcomp}
\numberwithin{equation}{section}
    
    \makeatletter
    \def\moverlay{\mathpalette\mov@rlay}
    \def\mov@rlay#1#2{\leavevmode\vtop{%
       \baselineskip\z@skip \lineskiplimit-\maxdimen
       \ialign{\hfil$\m@th#1##$\hfil\cr#2\crcr}}}
    \newcommand{\charfusion}[3][\mathord]{
        #1{\ifx#1\mathop\vphantom{#2}\fi
            \mathpalette\mov@rlay{#2\cr#3}
          }
        \ifx#1\mathop\expandafter\displaylimits\fi}
    \makeatother    
    
\newcommand{\cupdot}{\charfusion[\mathbin]{\cup}{\cdot}}

\usepackage{graphicx}
\usepackage{xcolor}
 \definecolor{VerdeFH}{HTML}{009374}

    \usepackage[margin=1.40in]{geometry}
    
    \usepackage[bookmarks=false]{hyperref}
    \hypersetup{
          colorlinks   = true,
         citecolor   =azulf!79!black
    }
     \hypersetup{linkcolor=azulf!79!black}

    \definecolor{azulf}{HTML}{0092D2}
\definecolor{azulc}{HTML}{00AEEF}
    \numberwithin{equation}{section}
    \newcommand{\CD}[1]{\mtr{CD}_{#1}}
    \newcommand{\Zc}{\mathcal{Z}}
    
    \newcommand{\Aut}{\mathrm{Aut}}
    \newcommand{\Out}{\mathrm{Inn}}
    \newcommand{\Inn}{\mathrm{Out}}
    
    \renewcommand{\and}{\,\mbox{and}\,}

    \newcommand{\Tr}{\mathrm{Tr}\hspace{1pt}}
    \newcommand{\TrN}{\mathrm{Tr}_{\hspace{-.9pt}N}\hspace{1pt}}
    \newcommand{\TrV}{\mathrm{Tr}_{\hspace{-.3pt}V}\hspace{1pt}}
    \newcommand{\TrM}{\Tr_{\hspace{-1pt}M_N(\C)}\hspace{1pt}}
    \newcommand{\hmu}{\hat{\mu}}
    \newcommand{\hnu}{\hat{\nu}}
    \newcommand{\Km}{K_{\mu}}
    \newcommand{\Kn}{K_{\nu}}
    \newcommand{\Xm}{X_{\mu}}
    \newcommand{\Xn}{X_{\nu}}
    
    \newcommand{\mtr}[1]{\mathrm{#1}}
    \newcommand{\mtf}[1]{\mathfrak{#1}}

    \newcommand{\A}{\mathcal{A}}
    \newcommand{\dif}[1]{\mathrm{d}#1}

    \newcommand{\re}{\mathbb{R}}

    \newcommand{\diag}{\mtr{diag}}

    \renewcommand{\H}{\mathcal{H}}
    
    \newcommand{\C}{\mathbb{C}}
    \newcommand{\peqsubind}[1]{_{\text{\tiny{#1} } } }
     \newcommand{\peqsupind}[1]{^{\text{\tiny{#1} } } }
    \usepackage{amssymb}
    \usepackage{mathrsfs}
    \usepackage{enumitem}
    \usepackage{accents}
    \newcommand{\udot}[1]{\underaccent{\cdot}{#1}}
    \newcommand{\ga}[1]{\gamma^{{#1}}}
    
    \newcommand{\gam}[1]{\gamma^{\mu_{#1}}}

    \newcommand{\ii}{\mathrm{i}}
    \newcommand{\ee}{\mathrm{e}}
    
    \newcommand{\inv}{^{-1}}
    \newcommand{\mtc}[1]{\mathcal{#1}}
    \newcommand{\Z}{\mathbb{Z}}
    \newcommand{\N}{\mathbb{N}}

    \newcommand{\hp}[1]{^{(#1)}}

    \newcommand{\where}{\mbox{where}\,\,}

\newcommand{\itemB}{\item[\bala]}

\newcommand{\itemW}{\item[\raisebox{2pt}{\scalebox{.45}{$\diamondsuit$}}]}

    \newcommand{\acomm}[1]{\{ #1, \balita \hspace{3pt} \}}

    \newcommand{\comm}[1]{[ #1, \balita \hspace{3pt}]}

 \usepackage{booktabs}
\usepackage{colortbl}
\colorlet{tableheadcolor}{gray!19} 
\newcommand{\headcol}{\rowcolor{tableheadcolor}} %
\colorlet{tablerowcolor}{gray!10} 
\newcommand{\rowcol}{\rowcolor{tablerowcolor}} %
\newcommand{\topline}{\arrayrulecolor{black}\specialrule{0.1em}{\abovetopsep}{0pt}%
            \arrayrulecolor{tableheadcolor}\specialrule{\belowrulesep}{0pt}{0pt}%
            \arrayrulecolor{black}}
\newcommand{\midline}{\arrayrulecolor{tableheadcolor}\specialrule{\aboverulesep}{0pt}{0pt}%
            \arrayrulecolor{black}\specialrule{\lightrulewidth}{0pt}{0pt}%
            \arrayrulecolor{white}\specialrule{\belowrulesep}{0pt}{0pt}%
            \arrayrulecolor{black}}

\newcommand{\bottomlinec}{\arrayrulecolor{tablerowcolor}\specialrule{\aboverulesep}{0pt}{0pt}%
            \arrayrulecolor{black}\specialrule{\heavyrulewidth}{0pt}{\belowbottomsep}}%
\newcommand{\Ht}{Q}
\newcommand{\Lt}{R}
 \usepackage{fancyhdr}
 \fancypagestyle{sin}{%
    \fancyhead{}
    
}

 \RequirePackage{fix-cm}
 \DeclareRobustCommand{\gobblefive}[5]{}

\newcommand{\pair}[2]{\TrN(#1)\cdot \TrN(#2) }

\newcommand{\numerada}{\refstepcounter{equation}\tag{\theequation}}

\newenvironment{salign} 
  {\csname align*\endcsname}
  {\csname endalign*\endcsname} 
\begin{document}

\title[From random NCG 
to bi-tracial multimatrix models]{Computing the 
spectral action for fuzzy geometries: from random noncommutative geometry 
to bi-tracial multimatrix models}

  \author[Computing the 
spectral action for fuzzy geometries]{\small{Carlos I. P\'erez-S\'anchez}}
\address{Faculty of Physics, University of Warsaw\\
 ul. Pasteura 5, 02-093 Warsaw, Poland 
} 
\email{cperez@fuw.edu.pl}


\begin{abstract}  
A fuzzy geometry is a certain type of spectral triple whose Dirac operator crucially turns out to be a finite matrix. This notion was introduced in [J. Barrett, {\em J. Math. Phys.} 56, 082301 (2015)] and accommodates familiar fuzzy spaces like spheres and tori. In the framework of random noncommutative geometry, we use Barrett's characterization of Dirac operators of fuzzy geometries in order to systematically compute the spectral action $S(D)= \mathrm{Tr} f(D)$ for $2n$-dimensional fuzzy geometries. In contrast to the original Chamseddine-Connes spectral action, we take a polynomial $f$ with $f(x)\to \infty$ as $ |x|\to\infty$ in order to obtain a well-defined path integral that can be stated as a random matrix model with action of the type $S(D)=N \cdot \mathrm{tr}\, F+\textstyle\sum_i \mathrm{tr}\,A_i \cdot  \mathrm{tr} \,B_i $, being $F,A_i $ and $B_i $  noncommutative polynomials in $2^{2n-1}$ complex $N\times N$ matrices that parametrize the Dirac operator $D$. For arbitrary signature---thus for any admissible KO-dimension---formulas for 2-dimensional fuzzy geometries are given up to a sextic polynomial,  and up to a quartic polynomial for 4-dimensional ones, with focus on the octo-matrix models  for Lorentzian and Riemannian signatures. The noncommutative polynomials $F,A_i $ and $B_i$ are obtained via chord diagrams and satisfy: independence of $N$; self-adjointness of the main polynomial $F$ (modulo cyclic reordering of each monomial); also up to cyclicity, either self-adjointness or anti-self-adjointness of $A_i $ and $B_i $ simultaneously, for fixed $i$. Collectively, this favors a free probabilistic perspective for the large-$N$ limit we elaborate on.
\end{abstract}

\maketitle
\fontsize{10.0}{13.0}\selectfont  
\setcounter{tocdepth}{2}
\tableofcontents
 
\fontsize{11.4}{14.0}\selectfont  


\section{Introduction}\label{sec:intro} 

In some occasions, the core concept 
of a novel research avenue can be 
traced back to a defiant attitude towards 
a no-go theorem. However uncommon this is,
some prolific theories that 
arose from a slight perturbation of the original
assumptions, aiming at an escape from the no-go,
have shaped the modern landscape of mathematical physics. 
Arguably, the best-known story fitting this description is supersymmetry.

Another illustration is found in   
noncommutative geometry (NCG) applications 
to particle physics:
 In an attempt to unify all fundamental interactions, the proposal of trading \textit{gravitation coupled to matter 
on a usual spacetime manifold} $M$
by \textit{pure gravitation on an extended space}  $M\times F$  is bound to fail---as a well-known symmetry 
argument (amidst other objections) shows---as far as
spacetime is extended by an ordinary manifold $F$.
Rebelling against the no-go result, 
while not giving up a gravitational unification approach,
shows the way out  
of the realm of commutative spaces (manifolds) 
after restating 
the symmetries in an algebraic fashion.
For the precise argument we refer to \cite[Sec. 9.9]{ConnesMarcolli} and 
for the details on the obstruction to \cite{thurston,mather,mather2,epstein}.
\par 
Following the path towards a
noncommutative description of the 
`internal space' 
$F$ (initially a two-point space), Connes
was able to incorporate 
the Higgs field on a geometrically equal footing with 
gauge fields, simultaneously  
avoiding the Kaluza-Klein tower 
that an augmentation of spacetime by an ordinary space $F$ would cause. 
As a matter of fact, not only the Higgs sector but
the whole classical action of the 
Standard Model of particle physics 
has been geometrically derived \cite{CCM,BarrettSM}
from the \textit{Chamseddine-Connes spectral action}\footnote{To be precise, in this 
article `spectral action' means 
`bosonic spectral action'. The 
derivation of the Standard Model requires 
also a fermionic spectral action $\langle J\tilde\psi, D \tilde\psi \rangle$
where $\tilde\psi $ is a matrix (see \cite{DabrowskiSitarzDAndrea}) of classical 
fermions. See 
 \cite{LizziCorfu} for a physics review
 and \cite{WvSbook} and \cite[Secs.9-18]{ConnesMarcolli} for 
 detailed mathematical exposition.} 
 \cite{Chamseddine:1996zu}. 
The three-decade-old history of the impact 
of Connes' groundbreaking idea 
on the physics beyond the Standard Model is told in \cite{surveySpectral} 
(to whose comprehensive references one could add later works
\cite{BesnardUone,Jordan,Martinetti:2019hhq,Bochniak:2020lab});
see his own review \cite{ConnesHighlights} for the impact of the spectral 
formalism on mathematics.
\par 

On top of the very active quest for the noncommutative internal space
$F$ that corresponds to a chosen field theory,
it is pertinent to point out that such theory is 
classical and that quantum field theory 
tools (for instance, the renormalization group) are adapted to it. 
The proposals on presenting noncommutative
geometries in an inherently quantum setting are diverse: 
A spin network approach led to the concept 
of \textit{gauge networks}, along with a blueprint for spin foams in NCG, 
as a quanta of NCG  \cite{MvS}; therein, from the spectral 
action (for Dirac operators) on gauge networks, 
the Wilson action for Higgs-gauge lattice theories and the
Kogut-Susskind Hamiltonian (for a 3-dimensional lattice) were derived, 
as an interesting result of the interplay
among lattice gauge theory, spin networks and NCG. 
Also, significant progress on the matter of  
fermionic second quantization of the spectral action, 
relating it to the
von Neumann entropy, has been proposed in  \cite{EntropySpectral};
and a bosonic second quantization was
undertaken more recently in \cite{KhalkhaliQuantization}. 
The context of this paper is a different, random geometrical approach 
motivated by the path-integral
quantization of noncommutative geometries
\begin{align} \label{eq:QSA}
 \mathcal Z = \int_{\mathcal M} \ee^{-  \Tr f(D) }\dif{D}\,,
\end{align}
where $S(D)=\Tr f(D)$ is the (bosonic) spectral 
action. The integration is 
over the space $\mathcal M$ of geometries 
encoded by Dirac operators $D$ on a Hilbert space 
that, in commutative geometry, corresponds to 
the square integrable spinors $\H_M$ (well-defining 
this $\mathcal Z$ is a fairly simplified version of
the actual open problem stated in \cite[Ch. 18.4]{ConnesMarcolli}). 
The meaning of this partition function $\mtc Z$ is not clear
for Dirac operators corresponding to an ordinary spacetime $M$.
In order to get a finite-rank Dirac operator
one can, on the one hand, truncate the algebra 
$C^\infty(M)$ and the Hilbert space $\H_M$ 
in order to get a well-defined measure
$\dif{D}$ on the space of geometries
$\mathcal M$, now parametrized by finite, albeit large, 
matrices. On the other hand, 
one does not want to fall in the class of 
\textit{lattice geometries}  \cite{PaschkeLattices} nor \textit{finite geometries}  \cite{KrajewskiDiagr}. 
\par 

Fuzzy geometries are finite-dimensional 
geometries that escape the classification
of finite geometries 
given in \cite{KrajewskiDiagr}, depicted in terms 
of the Krajewski diagrams, and \cite{PaschkeSitarz}. In fact, 
fuzzy geometries retain also a (finite dimensional) model 
of the spinor space that is not present in 
a finite geometry. Moreover, in contradistinction to lattices, fuzzy geometries
are genuinely---and not only in spirit--- noncommutative. In particular, the path-integral quantization of fuzzy geometries differs also from the approach in 
\cite{PaschkeLattices} for lattice geometries.
\par 

Of course, fuzziness is not new \cite{MadoreS2} and 
can be understood as limited
spatial resolution on spaces. The prototype is 
the space spanned by finitely many spherical harmonics 
approximating the algebra of 
functions on the sphere $S^2$. 
This picture is in line with models of quantum gravity,  
since classical spacetime 
is expected to break down at scales below Planck length \cite{DoplicherFredenhagenRoberts}. 
\par 

Although the three components of a \textit{spectral triple} 
have sometimes been evoked in the study of fuzzy 
spaces \cite{DolanHuetOConnor} and their Dirac operators 
on some fuzzy spaces are well-studied
(e.g. the Grosse-Pre\v{s}najder Dirac operator \cite{Grosse:1994ed}),
a novelty in \cite{BarrettMatrix} 
is their systematic spectral triple formulation;
for instance, fuzzy tori, elsewhere addressed (e.g. \cite{DolanOConnor,SteinackerFuzzy}), 
acquire a spectral triple \cite{BarrettGaunt}. 
Spectral triples are data that algebraically
generalize spin manifolds. More precisely, when 
the spectral triple is commutative 
(i.e. the algebraic 
structure that generalizes the algebra of coordinates is commutative, with additional assumptions we omit)
a strong theorem is the ability to construct, out of it, an oriented, smooth manifold, with its metric and spin$^c$ structure.
This has been proven by Connes \cite{Reconstruction} taking 
some elements from 
previous constructs by Rennie-V\'arilly \cite{RennieVarilly}. 
\par

This paper computes the spectral 
action for fuzzy geometries.  
Compared with the smooth case, our methods are simpler. 
For an ordinary manifold $M$ or an almost 
commutative space $M\times F$ (being $F$ a finite geometry \cite[Sec. 8]{WvSbook}), one commonly 
relies on a heat kernel expansion 

\[\Tr ( \ee^{-t D^2 } ) \sim \sum_{n \geq 0} t^{\frac{n-\dim(M)}{2}} a_{n}(D^2) \qquad \qquad( t^+\to 0 )\,,  \]
which allows, for $f$ of the Laplace-Stieltjes transform type 
 $f(x)=\int_{\re^+}  \ee^{-t x^2 } \dif \nu(t)$,  to determine the spectral action $\Tr f(D/\Lambda)$  
in terms of the Seeley-DeWitt coefficients $a_{2n}(D^2)$ \cite{Gilkey}, being $t=\Lambda\inv$ the inverse of the cutoff $\Lambda$; see also \cite{EcksteinIochum}. 
The elements of Gilkey's theory are not used here. Crucially,  
$f$ is instead assumed to be a polynomial
(with $f(x)\to \infty$ for  $ |x|\to\infty$), which enables one to 
directly compute traces of powers of the Dirac operator. This
alteration of the Chamseddine-Connes spectral
action ---in which $f$ is typically a symmetric bump function
around the origin--- comes from a convergence requirement 
for the path-integral \eqref{eq:QSA}, as initiated in \cite{BarrettGlaser} 
(a polynomial spectral action itself is  already considered in \cite{MvS}, though,
for gauge networks arising from embedded quivers in a spin manifold).
The motivation of Barrett-Glaser is  
to access information about fuzzy geometries 
by looking at the statistics of the eigenvalues of $D$
using Markov chain Monte Carlo simulations. 
This and a posterior study  \cite{BarrettDruceGlaser} deliver 
evidence for a phase transition to a 2-dimensional behavior (also of significance in quantum gravity \cite{CarlipDimRed}). \par

Finally, the paper is organized as follows:
the next section, based on \cite{BarrettMatrix}, introduces spectral triples and fuzzy geometries
in a self-contained way. The definition is slightly technical, but
the essence of a fuzzy geometry can be understood
from its matrix algebra, its Hilbert space $\H$ and Barrett's characterization
of Dirac operators (Secs. \ref{sec:GenDirac} and \ref{sec:CharacDirac}). In Section
\ref{sec:computeSA} we compute 
the spectral action in a general setting. A convenient 
graphical description of `trace identities' for gamma matrices 
(due to the Clifford module structure of $\H$, Sec. \ref{sec:gammas}) 
is provided in terms of \textit{chord diagrams}, which later serve as organizational 
tool in the computation of $\Tr (D^{m})$, $m\in \N$. As 
the main results in Sections \ref{sec:d2} and \ref{sec:d4}, 
we derive formulas for the spectral action for 
$2$- and $4$-dimensional fuzzy geometries, respectively.
In the latter case, we elaborate on the Riemannian 
and Lorentzian cases, being these 
the first reported (analytic) derivations for the spectral action 
of $d$-dimensional fuzzy geometries with general Dirac operators
in $d>3$.  Formulas for the spectral action for geometries of signature 
 (0,3), which lead to a tetra-matrix model\footnote{Strictly seen, these lead to an   octo-matrix model, but a simplification is allowed by  
the fact that the product of all 
gamma matrices is a scalar.}, were presented 
in \cite[App. A.6]{BarrettGlaser}. 
Later, Glaser explored the phase transition of the
fuzzy-sphere--like (1,3) case---that is of KO-dimension 2, as satisfied by 
the Grosse-Pre\v{s}najder operator---together with that of (1,1) and (2,0) geometries of KO-dimensions 0 and 6, respectively. For 
the (1,3) geometry, the spectral action used in the numerical simulations of
\cite{GlaserScaling} 
was obtained inside MCMCv4, a computer code aimed 
at simulating random fuzzy geometries; 
the formula  (in C++ language) for the spectral action can be found in the file \texttt{Dirac.cpp} 
of  \cite{GlaserCode}.
Our solely analytic approach to spectral action computations 
yields, out of a single general proof, a formula for any admissible KO-dimension,
as it will become apparent in Proposition \ref{thm:Haupt}.
\par
In Section \ref{sec:freeP}, we restate our results, 
aiming at free probabilistic tools
towards the large-$N$ limit (being $N$ the matrix size
in Barrett's parametrization of the Dirac operator). In order to define 
noncommutative (NC) distributions, one often departs from a
self-adjoint NC polynomial. It turns out that only a
weaker concept (`cyclic self-adjointness') defined here is satisfied by the main NC polynomial $P$ in $\Tr f(D)= N \cdot \TrN  P + \sum_i
\TrN \Phi_i \TrN \Psi_i $; the other NC polynomials 
$\Phi_i$ and $\Psi_i$ (for fixed $i$) either satisfy this very condition, or they are both cyclic anti-self-adjoint. The trace $\TrN$
cannot tell apart these conditions from the actual
self-adjointness of a NC polynomial.

The conclusions and the outlook are presented 
in the last two sections. The short Appendix \ref{sec:App} 
contains some useful information about (anti-)hermiticity 
of products of gamma matrices for general 
signature. To ease legibility, 
some steps in the proof of Proposition \ref{thm:Dsix}
(the sextic term in dimension 2) have been placed in
Appendix \ref{sec:appB}, which is also 
intended as a stand-alone example on how to gain 
NC polynomials from chord diagrams. 

\section{Fuzzy geometries as spectral triples} \label{sec:ST} 

The formalism of spectral triples in noncommutative geometry
can be very intricate and its full machinery will not be used here. 
We refer to \cite{WvSbook} for more details 
on the usage of spectral triples in high energy physics.

The essential structure is the \textit{spectral triple}
$(\A,\H,D)$, where  $\A$ is a unital, involutive algebra
of bounded operators on a Hilbert space $\H$. 
The \textit{Dirac operator} $D$ is 
a self-adjoint operator on $\H$ with compact resolvent
and such that $[D,a]$ is bounded for all $a\in \A$.
On $\H$, the algebraic behavior between of the Dirac operator 
and the algebra $\A$---and 
later also among $D$ and some 
additional operators on $\H$---encodes geometrical
 properties. For instance, 
  the geodesic distance $d_g$ between two points $x$ and $y$ of a Riemannian (spin) manifold $(M,g)$,
  can be recovered from $ d_g(x,y)=\sup _{a\in \A}  \{ | a(x)-a(y) | \,:\, a\in \A \mbox{ and } || [D,a] || \leq 1\}$,
  being $\A$ the algebra of functions on $M$ and $D$ the canonical Dirac operator \cite[Sect VI.1]{ConnesNCG}.   \par 
  
Precisely those additional operators lead to the concept of
\textit{real, even} spectral triple, which allows to build
physical models. 
Next definition, taken from \cite{BarrettMatrix},
is given here by completeness, since fuzzy geometries
are a specific type of real (in this paper all of them even) spectral triples.

\begin{definition} \label{def:realST}
A \textit{real, even spectral triple} of KO-dimension $s\in \Z/8\Z$
consists in the following objects and relations:
\begin{enumerate}[label=(\roman*)] \setlength\itemsep{.4em}
\item  an algebra $\A$ with involution $*$
 \item  a Hilbert space $\H$ together with a faithful, $*$-algebra
 representation $\rho:\A \to \mathcal L(\H)$
\item  an anti-linear unitarity (called \textit{real structure}) $J:\H\to \H$, $\langle J v,J w \rangle = \langle w,u\rangle$,
being $\langle \balita\hspace{2pt}, \balita \hspace{2pt}\rangle$ the inner product of $\H$ 
\item  a self-adjoint operator $\gamma:\H\to \H$ commuting with the representation $\rho$
and satisfying $\gamma^2=1$ (called \textit{chirality})
\item  for each $a,b\in \A$, $[\rho(a), J\rho(b)J\inv]=0$
\item a self-adjoint operator $D$ on $\H$ that satisfies
\[[\,[D,\rho(a)]\,, J \rho(b) J\inv]=0,\qquad a,b\in \A\, \]
\item  the relations%
\begin{subequations}
  \begin{align}
 J^2& =\epsilon   \\
 JD & =\epsilon' DJ  \label{eq:except} \\ 
 J\gamma & =\epsilon '' \gamma J
 \end{align}
 \end{subequations}
with the signs $\epsilon,\epsilon',\epsilon''$ determined  by  $s$ according to the following table:

\end{enumerate} \centering
\begin{tabular}{ccccccccc}
  \topline
  \headcol $s$ & 0 & 1 & 2 & 3 & 4& 5 &6 &7  \\
  \midline
  \rowcol $\epsilon$ & $+$ & $+$ & $-$ & $-$ &$-$&$-$& $+ $ &$+ $ \\
 $\epsilon'$ & $+$ & $-$ & $+$ & $+$ &$+$&$-$&$+$&$+$ \\
 \rowcol  $\epsilon''$ & $+$ & + & $-$ & + &$+$& +&$-$& +\\
  \bottomlinec
\end{tabular}
\end{definition}

A \textit{fermion space} of KO-dimension $s$ is 
a collection of objects $(\A,\H,J,\gamma)$ satisfying
axioms (i) through (v) and (vii), except for 
eq. \eqref{eq:except}.

\subsection{Gamma matrices and Clifford modules}\label{sec:Clifford}
Given a \textit{signature} $(p,q)\in \Z^2_{\geq 0}$, a \textit{spinor (vector) space} $V$ is a representation $c$ of 
 the Clifford algebra\footnote{We recall that $\mathcal C\ell(p,q)$ is the tensor algebra
 of $\re^{p+q}$ 
 modulo the relation $2g(v,w)=v\otimes w+w\otimes v$ for each $u,w\in \re^{p+q}$, being $g =\diag(+,\ldots,+,-,\ldots,-) $ 
 the quadratic form with $p$ positive and $q$ negative signs.} $\mathcal C\ell(p,q)$.
 Thus, 
elements of the basis $e^a$ and $e^{\dot a}$ of
 $\re^{p,q}$,
$a=1,\ldots,p$ and $\dot a=1,\ldots, q$, 
become endomorphisms $c(e^a)=\gamma^a, c(e^{\dot a})=\gamma^{\dot a}$
 of $V$. 
  If $d=q+ p$ is even, $V$ is assumed 
 to be irreducible, whereas only 
 the eigenspaces $V^{+}, V^{-} \subset V$ of $\gamma$ are,
 if $q+ p$ is odd.   
 The size of these square matrices (the Dirac \textit{gamma matrices}) is $2^{\lfloor{p+q}\rfloor}$.

It follows from the relations of the Clifford algebra that 
 \[\gamma^{\mu}\gamma^{\nu}= \gamma^{[\mu}\gamma^{\nu]} + \frac12\{\gamma^\mu,\gamma^\nu\}= \gamma^{[\mu}\gamma^{\nu]}+g^{\mu\nu}1_V\,,\]
 which can be used to iteratively compute products of gamma matrices
in terms of $g^{\mu\nu}$ and their anti-symmetrization.
Taking their trace $\TrV$ (contained in the spectral action)
gets rid of the latter, so we 
are left with $\dim V \cdot g^{\mu_1\nu_1}g^{\mu_2\nu_2}\cdots$.
A product of an odd number of gamma matrices is traceless; 
the trace of a product of $2n$ gamma matrices can be expressed as 
a sum of over $(2n-1)!!$ products of $n$ bilinears $g^{\mu\nu}$
that will be represented diagrammatically. 

\subsection{Fuzzy geometries} \label{sec:matrix}
Section \ref{sec:matrix} is based on \cite{BarrettMatrix}.  
A fuzzy geometry can be thought of as a finite-dimensional approximation 
to a smooth geometry.
A simple matrix algebra  $M_N(\C)$ conveys information about
the resolution of a space (an inverse power of $N$, e.g. $\sim 1/\sqrt{N}$
for the fuzzy sphere $\mathbb S^4_N$ \cite{SperlingSteinacker}) where the
noncommutativity effects are no longer negligible. 
To do geometry on a matrix algebra one needs additional information,
which, in the case of fuzzy geometries, is in line with the spectral 
formalism of NCG.  

\begin{definition}[Paraphrased from \cite{BarrettMatrix}]\label{def:fuzzy}
A \textit{fuzzy geometry} of \textit{signature} 
$(p,q)\in \Z_{\geq 0}^2$ is given by 

\begin{itemize}  \setlength\itemsep{.4em}
 \itemB a simple matrix algebra $\A$ with coefficients in $\mathbb K=\re,\C,\mathbb H$; in
 the latter case, $M_{N/2}(\mathbb H)\subset  M_{N}({\C})$, otherwise $\A$ is $M_N(\re)$ or
 $M_N(\C)$ --- in this paper we take always $ \A=M_N(\C)$ 
 
 \itemB a Hermitian $\mtc C\ell (p,q)$-module $V$ with a \textit{chirality} $\gamma$.
That is a linear map $\gamma:V\to V$ satisfying $\gamma^*=\gamma $ and $  \gamma^2=1$ 

 \itemB a Hilbert space $\H = V\otimes M_N({\C})  $ with inner product
 $\langle v\otimes R, w\otimes S \rangle = (v,w) \Tr(R^* S) $ for each $R,S\in M_N(\C)$,
 being $(\hspace{-1pt}\balita\hspace{1pt},\hspace{-2pt}\balita)$ the inner product of $V$
 \itemB a left-$\A$ representation $\rho(a)(v\otimes R) = v\otimes (a R)$ on $\H$, $a\in \A$ and $v\otimes R\in \H$

\itemB three signs $\epsilon, \epsilon',\epsilon''\in \{-1,+1\}$ determined
 through $s:=q-p$ by the following table: 
\[ 
\begin{tabular}{ccccccccc}
  \topline
  \headcol $s\equiv q-p \,\,\mtr{( mod }\, 8)$ & 0 & 1 & 2 & 3 & 4& 5 &6 &7  \\
  \midline
  \rowcol $\epsilon$ & $+$ & $+$ & $-$ & $-$ &$-$&$-$& $+ $ &$+ $ \\
 $\epsilon'$ & $+$ & $-$ & $+$ & $+$ &$+$&$-$&$+$&$+$ \\
 \rowcol  $\epsilon''$ & $+$ & + & $-$ & + &$+$& +&$-$& +\\
  \bottomlinec
\end{tabular}\]

 \itemB a \textit{real structure} $J= C\otimes *$, where $*$ is complex conjugation and $C$ is an 
 anti-unitarity on $V$ satisfying $C^2=\epsilon$ and $C\gamma^\mu=\epsilon'\gamma^\mu C$ for all the gamma matrices $\mu=1,\ldots,p+q$
  
 \itemB a self-adjoint operator $D$ on $\H$
satisfying the \textit{order-one condition}
\[
[\,[ D,\rho(a)]\,, J\rho(b)J\inv]=0 \qquad \mbox{for all }a,b\in \A
\]
 
 \itemB a chirality $\Gamma=\gamma\otimes 1_\A$ for $\H$,
 where $\gamma$ is the chirality of $V$.
These signs impose on the operators the following conditions:
  \begin{align*}
 J^2 =\epsilon \, , \qquad 
 JD =\epsilon' DJ \,,  \qquad
 J\Gamma  =\epsilon '' \Gamma J\,.
 \end{align*}
\end{itemize}
For $s$ odd, $\Gamma$ can be thought of as the identity $1_{\H}$.
The number $d=p+q$ is the \textit{dimension} of the spectral triple and $s=q-p$
is its \textit{KO-dimension}.
\end{definition}
We pick gamma matrices
that satisfy
\begin{subequations}
 \label{eq:ConvGammas}
\begin{align} 
(\gamma^\mu)^2&=+1, \quad \mu=1,\ldots, p,  &&\gamma^\mu  \mbox{ Hermitian}, \\
(\gamma^\mu)^2&=-1, \quad \mu=p+1,\ldots, p+ q, && \gamma^\mu \mbox{ anti-Hermitian}\,,
\end{align}
\end{subequations}
in terms of which the chirality for $V$
is given by $\gamma=(-\ii)^{s(s-1)/2} \gamma^1\cdots \gamma^{p+q}$. 
For mixed signatures it will be convenient to separate spatial from time like indices,
and denote by lowercase Roman 
the former ($a=1,\ldots,p$)
and by dotted indices\footnote{Dotted indices are here unrelated to their usual interpretation
in the theory of spinors. Also for the Lorentzian signature,
the $0,1,2,3$ numeration (without any dots) is used.} ($\dot c=p+1,\ldots, p+q$) the latter.
The gamma matrices $\gamma^a$ are
 Hermitian matrices squaring to $+1$,
 and $\gamma^{\dot c}$'s denote here 
the anti-Hermitian matrices squaring to $-1$.
Greek indices are spacetime indices $\alpha, \beta,\mu,\nu, \ldots \in\Delta_d := \{1,2,\ldots, d\}$.

We let the gamma matrices generate $\Omega:= \langle \gamma^1,\ldots,\gamma^d \rangle_\re$ as algebra;
this splits as $
\Omega = \Omega^+\oplus \Omega ^-$, 
where $\Omega^+$ contains products of even number of gamma
matrices and $\Omega^-$  an odd number of them.

\subsection{General Dirac operator}\label{sec:GenDirac}
Using the spectral triple axioms for fuzzy geometries, 
their Dirac operators can be characterized as
self-adjoint operators of the
form \cite[Sec. 5.1]{BarrettMatrix}
\begin{align} \label{eq:Dcharac}
D (v\otimes R)= \sum _I \, 
\omega^I v \otimes (  K_I R + \epsilon ' R K_I)\,, \quad  v\in V,R\in M_N(\C)\,,\end{align}
with $\{\omega^I \}_I$ a linearly independent set and 
 $I$ an abstract index  to be clarified now. For $r\in\N_{\leq d}$, let $\Lambda_d^r$
 be the set of $r$-tuples of increasingly ordered spacetime indices $\mu_i\in \Delta_d$, i.e. 
$
\Lambda_d^r=\{(\mu_1,\ldots,\mu_{r}) \,| \,\  
\mu_i < \mu_j, \mbox{ if } i<j
\}\,.$ We let $  \Lambda_d^{\phantom -}= \cupdot_r \Lambda_d^r\,, 
$
whose odd part is denoted by $\Lambda_d^-$, 
\begin{align*}
\Lambda_d^-&=\{(\mu_1,\ldots,\mu_{r}) \,| \,\mbox{for some odd $r$, $1\leq r \leq d $ } \,\&\, 
\mu_i<\mu_j \mbox{ if } i<j
\}
\\
&=  \{1,\ldots, d\}
\cup \{ (\mu_1,\mu_2,\mu_3)\, | \,1\leq \mu_1 < \mu_2 < \mu_3 \leq d \}
\\
&\quad \cup  
 \{ (\mu_1,\mu_2,\mu_3,\mu_4,\mu_5)\, | \,1\leq \mu_1<\mu_2 \ldots < \mu_5 \leq d  \} \cup \ldots 
\end{align*}
The most general Dirac operator in dimension $d$
writes in terms of products $\Gamma^I$ of gamma matrices 
that correspond to indices  $I$ in these sets, each bearing a matrix coefficient $k_I$,
\begin{align}
D\hp{p,q}
=\begin{cases}
 \sum_{I\in \Lambda_d^-} \Gamma^I \otimes k_I&
\mbox{for $d=p+q$ even,} \\
\sum_{I\in \Lambda_d^{\phantom{-}}}\Gamma^I \otimes k_I &
\mbox{for $d=p+q$ odd.}
 \end{cases}
\end{align}
We elaborate on each of the tensor-factors, $\Gamma^I$ and $k_I$. 
First, $\Gamma^I$ is the ordered product 
of gamma matrices with all single indices 
appearing in $I$, 
\[\Gamma^I=\gam 1 \cdots  \gam r\,,\]
for $I=(\mu_1,\ldots,\mu_r) \in \Lambda_d^r$. 
This can be thought of as
each gamma matrix $\gam{}$ 
corresponding to a one-form $\dif x^\mu$ 
(in fact, via Clifford
multiplication for canonical spectral triples) and $\Lambda_d$
as the basis elements  of the exterior 
algebra.
The set $\Lambda_d=\cupdot_r\Lambda^r_d$
can thus be seen as an abstract backbone of 
the de Rham algebra $\Omega^{*}_{\mtr{dR}}=\bigoplus_r\Omega^r_{\mtr{dR}}$
and $\Lambda^r_d$ of the $r$-forms $\Omega^r_{\mtr{dR}}$.
There are $\#(\Lambda^r_d)=\binom{d}{r}$ independent 
$r$-tuple products of gamma matrices. 
We now separate the cases according to 
the parity of $s$ (or of $d$).  Second, $k_I(R)=(  K_I R + \epsilon ' R K_I)$ is an
operator on $M_N(\C)\ni R$, which needs a
dimension-dependent characterization. 

\subsection{Characterization of the Dirac operator in even dimensions}\label{sec:CharacDirac}

We constrain the discussion to even (KO-)dimension. 
In Definition \ref{def:fuzzy} the table implies $\epsilon'=1$;
on top of this, the self-adjointness of $D$ implies that
for each $I$, both $\omega^I$ and $R \mapsto (  K_I R +   R K_I^*)$
are either Hermitian or both anti-Hermitian. 
In terms of the matrices $K_I$, this condition 
reads $K_I^*=+K_I$ or
$K_I^*=-K_I$, respectively. In the first
case we write $K_I=H_I$, in the latter 
$K_I=L_I$.
One can thus split the sum in eq. \eqref{eq:Dcharac}
as
\begin{align}
D (v\otimes R)= \sum _  {I}  \omega^I v \otimes (  H_I R +   R H_I)
+
\sum_ {I}   \omega^{I} v \otimes (  L_{I} R -  R L_{I})\,.
\end{align}
Additionally, since $d=2p-s$ is even,  $\gamma \gamma^a + \gamma^a \gamma=0$. Hence
the same anti-commutation relation  $
\gamma \omega +  \omega \gamma=0 \mbox{ holds for each $\omega\in \Omega^-\,.$ }
$
This leads to the splitting
\begin{align}
D = \sum _  {I\in \Lambda_d^-}  \tau^{I} \otimes \{ H_I,  \balita\, \}
+
\sum_ {I\in \Lambda_d^-}   \alpha^I\otimes [ L_I,  \balita\,]\,,
\end{align} where
each $\alpha^I$ and $\tau^I$ is an odd product of gamma matrices,
and $\Lambda_d^-$ is the set of multi-indices of an odd number of
indices $\mu\in \{1,\ldots, d\}$. 
In summary, 
\[
(\tau^I)^* = \tau^I\in  \Omega^-\,, \qquad (\alpha^I)^*=-\alpha^I \in  \Omega^-,\qquad\, H_I^*=H_I\,,
\qquad\, L_I^*=-L_I\,.
\]

We generally treat commutators and anti-commutators 
as (noncommuting) letters $k_I=\{K_I,\balita\hspace{3pt}\}_{\pm}$, 
for each $I\in \Lambda_d$.
The sign $e_I=\pm$ determines the type of the letter 
for $k_I$, being the latter defined by the rule
\begin{align}
\text{if }
e_I=\begin{cases}
        +1 &  \text{then } K_I=H_I \qquad \text{therefore } k_I=  h_I, \\
        -1 & \text{then } K_I=L_I \qquad\, \text{therefore } k_I=  l_I,
        \end{cases}
        \label{eq:HorL}
\end{align}
so 
$k_I(R)=
K_I R + e_I R K_I $, for $R\in M_N(\C)$.
Explicitly, one has
\begin{align} \nonumber
D\hp{p,q}
&=
\sum_{\mu}^{\phantom d}  \gam{} \otimes k_\mu
+\sum_{\mu,\nu,\rho} \ga{\mu}\ga\nu\ga\rho \otimes k_{\mu\nu\rho}+\ldots
+
\sum_{\widehat{\mu\nu\rho}} \Gamma^{\widehat{\mu\nu\rho} } 
\otimes k_{\widehat{\mu\nu\rho} }
+
\sum_{\hat{\mu}} \Gamma^{\hat \mu} \otimes k_{\hat \mu}\,,
    \end{align}
which runs through 
$\sum_{\boldsymbol\mu }  \gam 1 \gam 2 \cdots  \gam {d/2} 
    \otimes k_{\mu_1\mu_2\ldots \mu_{d/2}}$
 if the $4$ divides $d$,
or through \vspace{4pt}
 \[ \sum_{\boldsymbol\mu } \Gamma^{\mu_1 \ldots \mu_ {d/2-1}}
    \otimes k_{\mu_1\mu_2\ldots \mu_{d/2-1}} 
    +    
\sum_{{\boldsymbol\mu}} \Gamma^{\mu_1 \ldots \mu_ {d/2+1}}
    \otimes k_{\mu_1\mu_2\ldots \mu_{d/2-1}} \,  \] 
if $d$ is even but not divisible by $4$. 
Hatted indices are, as usual, those excluded from $\{1,\ldots, d\}$,
\begin{align}
\widehat{\mu \nu\ldots \rho }=(1,2,\ldots \mu-1, \mu+1,\ldots,
\nu-1, \nu+1,\ldots \rho-1, \rho+1,\ldots, d)\,.
\end{align}

In order for a Dirac operator to be self-adjoint, 
$k_I$ is constrained by the parity of $r=r(I)$, being $|I|=2r-1$,
and by the number $u(I)$ of 
spatial gamma matrices in the product $\Gamma^I$. In a mixed signature setting, 
$p,q >0$, an arbitrary $I\in \Lambda_d^-$
has the form $I=(a_1\ldots,a_t,\dot c_1,\ldots \dot c_u)$ for 
$0 \leq t  \leq p 
, \,0  \leq u \leq q $,
and so the corresponding matrix satisfies 
\begin{align} \label{eq:autoadjuncion}
(\Gamma^I)^*=(-1)^{u+ \lfloor (u+t)/2 \rfloor} \Gamma^I = (-1)^{u+r-1} \Gamma^I\,.
\end{align}
The first equality is shown in detail in Appendix \ref{sec:App}.
The second is just due to
$(-1)^{\lfloor (u+t)/2 \rfloor}=(-1)^{\lfloor (2r-1)/2 \rfloor}
=(-1)^{r-1}$. This decides whether $k_I$ 
should be an `$h_I$-operator' or an `$l_I$-operator' (see eq. \ref{eq:HorL}),
which is summarized in Table \ref{tab:HorL}.  

\begin{table}
\begin{tabular}{ccc}
  \topline
  \headcol $u(I)$ & $r(I)$  & $k_I$\\
  \midline
 \rowcol  even & odd & $h_I$ \\
           odd & odd & $l_I$ \\
 \rowcol  even & even & $l_I$\\
          odd & even &$h_I$ \\
  \bottomlinec 
\end{tabular}\vspace{.5cm}
 \caption{For $I=(a_1,\ldots,a_u,\dot c_1,\ldots, \dot c_t)\in \Lambda^{2r-1}_d$
a Hermitian matrix $H_I$ or an anti-Hermitian matrix $L_I$
parametrizes 
 $k_I$ according to the shown operators
 $h_I=\acomm{H_I}$ 
 or 
 $l_I=\comm{L_I}$\label{tab:HorL}}
\end{table}
For indices running where the dimension
bounds allow, one has
\begin{align*}
 D\hp{p,q}
 & = \sum_{a=1}^p \gamma^a\otimes h_a +
 \sum_{\dot c=p+1}^{p+q} \gamma^{\dot c}\otimes l_{\dot c} \\
& +  
 \sum_{a,b,c}  \gamma^a\gamma^b\gamma^c\otimes l_{abc} +
  \sum_{a,b,\dot c}  \gamma^a\gamma^b\gamma^{\dot c}\otimes h_{ab\dot c}
 \\
  & +\sum_{a,\dot b,\dot c}  \gamma^a\gamma^{\dot b} 
  \gamma^{\dot c}\otimes l_{a\dot b\dot c}  
  +\sum_{\dot a,\dot b,\dot c}  \gamma^{\dot a}\gamma^{\dot b}
  \gamma^{\dot c}\otimes l_{a\dot b\dot c}  +\ldots \\
  & + \begin{cases}
    \sum_a \Gamma^{\hat a} \otimes h_{\hat a}+ \sum_{\dot c} \Gamma^{\hat{\udot c}}\otimes l_{\hat{\udot c}} &  \text{if $q$ and $d/2$ have same parity,} \\    
   \sum_a \Gamma^{\hat a} \otimes l_{\hat a}+ \sum_{\dot c} \Gamma^{\hat{\udot c}} \otimes h_{\hat{\udot c}}&  \text{if $q$ and $d/2$ have opposite parity}.
    \end{cases}\,
\end{align*}
The last term is a product of $d-1=p+q-1$ matrices. This expression is
again determined by 
observing that the operator $k_{\hat \mu}$
is self-adjoint if $(-1)^{u+d/2}$ 
equals $+1$ and otherwise anti-Hermitian, 
being $u$ the number of spatial gamma matrices in $\Gamma^{\hat\mu}
=\gamma^1\cdots \widehat{\ga{ \mu}}\cdots \gamma^d$. 
We proceed to give some examples.

\begin{example}[Fuzzy $d=2$ geometries]
The next operators 
appear in \cite{BarrettMatrix}:

\begin{itemize}
 \itemB \textit{Type} (0,2). Then $s=d=2$, so $\epsilon'=1$.
 The gamma matrices are anti-Hermitian
 and satisfy $(\gamma^i)^2=-1$. The Dirac operator is
 \[ D\hp{0,2}=\gamma^{\dot 1} \otimes [L_1, \balita\,] +\gamma^{\dot 2} \otimes [L_2, \balita\,]
\]

  \itemB \textit{Type}   (1,1).
  Then $d=2$, $s=0$, so $\epsilon'=1$.
 \[ D\hp{1,1}=\gamma^1 \otimes \{H, \balita\,\} +\gamma^{\dot 2} \otimes [L, \balita\,] \]

 \itemB \textit{Type} (2,0). Then $d=2$, $s=6$, so $\epsilon'=1$.
 The gamma matrices are Hermitian
 and satisfy $(\gamma^i)^2=+1$. The Dirac operator is
  \[ D\hp{2,0}=\gamma^1 \otimes \{H_1, \balita\,\} +\gamma^2 \otimes \{H_2, \balita\,\} \]

\end{itemize}

\end{example}

\begin{example}[Fuzzy $d=4$ geometries] \label{ex:4d}
For realistic models the most important $4$-fuzzy geometries
have signatures $(0,4)$ and $(1,3)$ corresponding to the
Riemannian and Lorentzian cases. We derive
 the first one in detail in order to arrive at the result
in \cite[Ex. 10]{BarrettMatrix}. The rest 
follows from considering eq. \eqref{eq:autoadjuncion}.
\begin{itemize}
       \itemB \textit{Type}  (0,4), $s=4$,  \textit{Riemannian}.
       Notice that the gamma matrices are all 
       anti-Hermitian and square to $-1$ in this case. Therefore,
       products of three gamma matrices are self-adjoint:
       $(\gamma^{\dot a}\gamma^{\dot b}\gamma^{\dot c})^*=(-)^3
       \gamma^{\dot c}\gamma^{\dot b}\gamma^{\dot a}=\gamma^{\dot a}\gamma^{\dot b}
       \gamma^{\dot c}$.
       The accompanying operators  have the form $  \{  H_{ \dot a \dot b \dot c}, \balita\,\}$ 
       for $  H_{\dot a \dot b \dot c}$ self-adjoint:
     \begin{equation} \label{eq:Riem4d}
     D\hp{0,4}= \sum_{\dot a}\gamma^{\dot a} \otimes [L_{\dot a}, \balita\,]+
     \sum_{\dot a < \dot b < \dot c}\gamma^{\dot a}\gamma^{\dot b}\gamma^{\dot c}\otimes  
     \{  H_{\dot a \dot b \dot c}, \balita\,\}
     \end{equation}

        \itemB \textit{Type}  (1,3), $s=2$, \textit{Lorentzian}.  Call $\gamma^0$
        the only gamma matrix that squares to $+1$, and denote the rest by $\gamma^{\dot c}$, ${\dot c}=1,2,3$. Then
  \begin{align} \nonumber
     D\hp{1,3} & = \gamma^0 \otimes \acomm{H}+
      \gamma^{\dot c}  \otimes   \comm{L_{\dot c}} \\
     &   + \sum_{{\dot a}<{\dot c}}\gamma^0 \gamma^{\dot a}\gamma^{\dot c} \otimes \comm{L_{{\dot a}{\dot c}}}
       +  \gamma^{\dot 1}  \gamma^{\dot 2} \gamma^{\dot 3}\otimes  \{\tilde H , \balita\,\}
       \label{eq:Lor4d}
     \end{align}
     (For the Lorentzian signature,  
     the dotted-index convention is redundant 
     with the usual $0,1,2,3$ spacetime indices;
     then we henceforward drop it).

       \itemB \textit{Type}  (4,0), $s=-4=4$, mod $8$. 
       The opposite case to `Riemannian': now all gamma matrices are Hermitian,
       square to $+1$, and triple products $\Gamma^{\hat a}$ are 
       skew-Hermitian:
       \begin{equation}
     D\hp{4,0}= \sum_a\gamma^a \otimes \acomm{H_a}+ \sum_{a<b<c} \gamma^a  \gamma^b \gamma^c\otimes  [L_{abc}, \balita\,]
     \end{equation}

   \itemB \textit{Type}  (2,2), $s=0$, \textit{two times}. Choosing the first two gamma matrices
   such that they square to +1, and $\gamma^{\dot 3}$ and $\gamma^{\dot 4}$ 
   to $-1$, one gets
  \begin{align} \nonumber
     D\hp{2,2} & = \sum_{a=1,2}\gamma^a \otimes \acomm{H_a}+
      \gamma^a  \gamma^{\dot 3} \gamma^{\dot 4}\otimes  [\tilde L_{a}, \balita\,]  \\
     &   + \sum_{\dot c=\dot 3,\dot 4}\gamma^{\dot c} \otimes \comm{L_{\dot c}}
       +  \gamma^1  \gamma^2 \gamma^{\dot c}\otimes  \{\tilde H_{\dot c}, \balita\,\} \
     \end{align}
 Here we made the notation lighter, writing $\tilde L_{a}=L_{\hat a}=L_{1\ldots \hat a \ldots 4}$
  for the $L$-matrix
 with all indices but $a$. Similarly,  $\tilde H_{\dot c}=H_{1\ldots \hat{\udot c}\ldots 4}$.

               \itemB \textit{Type}  (3,1), $s=6$.
                 \begin{align} \nonumber
     D\hp{3,1} & =  \sum_{a=1,2,3}\gamma^a \otimes \acomm{H_a} + \gamma^1\gamma^2\gamma^3 \otimes \comm{\tilde L} \\
     &   +  \gamma^{\dot 4} \otimes \comm{L}
       +  \sum_{a<c}\gamma^a  \gamma^c \gamma^{\dot 4}\otimes   \acomm{ H_{ac{\dot 4}}}
     \end{align}
\end{itemize}

\end{example}
 
Fuzzy geometries with odd $s$ allow elements of $\Omega^+$ also to parametrize
Dirac operators, 
\begin{align} \nonumber
D\hp{p,q}
&=
\sum_{\mu} \gam{} \otimes k_\mu
+
\sum_{\mu_1,\mu_2} \gam{1}\gam{2} \otimes k_{\mu_1\mu_2}
+\sum_{\mu,\nu,\rho} \ga{\mu}\ga\nu\ga\rho \otimes k_{\mu\nu\rho}+
\\ &\qquad
\ldots +
\sum_{\widehat{\mu\nu\rho}} \Gamma^{\widehat{\mu\nu\rho} } 
\otimes k_{\widehat{\mu\nu\rho} }
+
\sum_{\widehat{\mu \nu}} \Gamma^{\widehat{\mu\nu}} \otimes k_{\widehat{\mu \nu}}
+
\sum_{\hat{\mu}} \Gamma^{\hat \mu} \otimes k_{\hat \mu}\,.
\end{align}
Examples of $D$ for a $d=3$ geometry are given in \cite{BarrettMatrix} 
and are not treated here.

\subsection{Random fuzzy geometries}
Given a fermion space of fixed signature $(p,q)$,
that is to say, a list $(\A,\H, \balita\,,J,\Gamma)$ satisfying
the listed properties in Definition \ref{def:realST} ignoring those
concerning $D$,
we consider the space $\mathcal{M}\equiv \mathcal{M}(\A,\H,J,\Gamma,p,q)$
of all possible Dirac operators $D$ that make of $(\A,\H,D,J,\Gamma)$ a real even
spectral triple of signature $(p,q)\in \Z_{\geq 0}^2$. 
 \par 
The symmetries of a spectral triple are encoded in $\Aut (\A)$,
$\Out (\A)$ and $\Inn(\A)$, none of which implies
the Dirac operator. This can be compared with the classical situation,
in which fixing
the data $(\A,\H,J,\Gamma)$ can be interpreted as
imposing symmetries on the system and subsequently finding
compatible geometries, encoded in $D\in\mathcal{M}$, typically
via the extremization $\delta S(D_0)=0$ of an action functional $S(D)$
that eventually selects a unique classical solution  $D_0\in\mathcal{M}$.
The random noncommutative setting that appears in
\cite{BarrettGlaser}, on the other hand, considers
`off-shell' geometries. These can be stated as
the following matrix integral
\begin{align}\label{eq:Z_spectral}
\mathcal Z\hp{p,q}=\int_{\mathcal{M}}  \, \ee^{-S(D)} \dif D,\quad S(D) =\Tr f(D)\,,
\end{align}
being $f(x)$ an ordinary\footnote{In contrast to noncommutative polynomials mentioned below.} polynomial
of real coefficients and no constant term. 
We next compute the spectral action $\Tr f(D)$.
 
\section{Computing the spectral action}\label{sec:computeSA}

In the spectral action \eqref{eq:Z_spectral} 
the trace is taken on the Hilbert space $\H$. We do not label it
but, to avoid confusion, we label traces on other spaces:
the trace $\TrV$ is that of the 
spinor space $V$, the trace of operators 
on the matrix space $M_N(\C)$ is 
denoted by $\TrM$, and $\TrN$ stands for the trace on $\C^N$.\par
A homogeneous element spanning the  
Dirac operator $D=\sum_I\omega_I \otimes k_I$ contains a first factor $\omega_I$, consisting 
of products of gamma matrices, and a second factor $k_I$
determined by a matrix that is either 
Hermitian or anti-Hermitian \cite{BarrettMatrix}.
We describe each factor and then give a
general formula to compute the spectral
action.

\subsection{Traces of gamma matrices}\label{sec:gammas}

We now rewrite the quantity
\begin{align}
\langle  {\mu_1}\ldots {\mu_{2n}} \rangle :=\frac{1}{\dim V} 
 \Tr_V(\gamma^{\mu_1}\ldots \gamma^{\mu_{2n}})
\end{align}
in terms of \textit{chord diagrams} of $2n$ points\footnote{In a more involved context, these are called `chord diagrams with one backbone' \cite{GenChordDiags}.}, to wit $n$ (disjoint) pairings 
among  $2n$ cyclically ordered points. These are 
typically placed on a circle in whose interior
the pairings are represented by chords that might cross. 
 One finds
  \begin{align} \label{eq:Chords}
\langle  {\mu_1}\ldots {\mu_{2n}}\rangle = 
\sum_{\substack{ 2n\text{\scriptsize -pt chord} 
\\ \text{\scriptsize diagrams }  \chi  } }
(-1)^{\#\{\text{crossings of chords in }\chi\}} 
\prod_{\substack{i,j =1\\ i\sim_\chi j} }^{2n} g^{\mu_{i}\mu_{j}}
\end{align}
where $\sim_\chi$ means that the point $i$ is joined 
with $j$ in the chord diagram $\chi$.
We denote the total number of crossings of chords by 
$\mtr{cr(}\chi) $. We count only 
simple crossings; for instance, the 
sign of the `pizza-cut' 8-point chord diagram with 
longest chords in the upper left corner of Figure 
\ref{Fig:Eight} is $(-1)^6$.

For a mixed signature, $q, p >0$, any non-vanishing
$\langle  {\mu_1}\ldots {\mu_{2n}} \rangle $
has the form (up to a reordering sign) $\langle a_1 \ldots a_{2r} \dot c_1 \ldots \dot c_{2u} \rangle $
with $r+u=n$. Since $g^{a\dot c}$ vanishes,
any chord diagram $\chi$ in the sum 
of eq. \eqref{eq:Chords}
splits into a pair $(\sigma,\rho)$ of smaller
chord diagrams,
of $2r$ and $2u$ points, whose 
chords do not cross (see Fig. \ref{Fig:Cutting}),
so $\mtr{cr(}\chi)=\mtr{cr(}\sigma)+\mtr{cr(}\rho)$. Therefore
\begin{align} \nonumber
\langle a_1 \ldots a_{2r} \dot c_1 \ldots \dot c_{2u} \rangle
&= 
\sum_{\substack{ 2n\text{\scriptsize -pt chord} 
\\ \text{\scriptsize diagrams }  \chi  } }
(-1)^{\mtr{cr}( \chi)} 
\prod_{\substack{i,j \\ i\sim_\chi j} } g^{a_{i}a_{j}}\times\prod_{\substack{u,v \\ u\sim_\chi v} } g^{\dot c_{u}\dot c_{v}}  \\
&=\sum_{\substack{ (2r,2u)\text{\scriptsize -pt chord} 
\\ \text{\scriptsize diagrams }  (\rho,\sigma) } }
(-1)^{\mtr{cr}(  \sigma)  } 
\prod_{\substack{i,j \\ i\sim_\rho j} } g^{a_{i}a_{j}}
\times \nonumber
(-1)^{\mtr{cr}(  \rho ) } 
\prod_{\substack{u,v \\ u\sim_\sigma v} } g^{\dot c_{u}\dot c_{v}}
\\
& = \langle a_1 \ldots a_{2r} \rangle\langle  \dot c_1 \ldots \dot c_{2u} \rangle
\label{eq:ChordsSplit}\,.
\end{align}
For the metric $g^{\mu\nu}=\diag(+,\ldots,+,-,\ldots,-)$ the two factors are 
\begin{subequations}
 \begin{align}
 \langle a_1 \ldots a_{2r} \rangle & = 
 \sum_{\substack{2r\text{\scriptsize -pt chord} 
\\ \text{\scriptsize diagrams }  \rho } }
(-1)^{\mtr{cr}(  \sigma)  } 
\prod_{\substack{i,j \\ i\sim_\rho j} } \delta^{a_{i}a_{j}} \,,
 \\
 \langle  \dot c_1 \ldots \dot c_{2u} \rangle& = (-1)^u  
 \sum_{\substack{ 2u\text{\scriptsize -pt chord} 
\\ \text{\scriptsize diagrams } \sigma  } }
 (-1)^{\mtr{cr}(  \rho ) } 
\prod_{\substack{w,v \\ w\sim_\sigma v} } \delta^{\dot c_{w}\dot c_{v}} \,.
\end{align}
\label{eq:WithDeltas}
\end{subequations}

\begin{figure}
\includegraphics[width=1\textwidth]{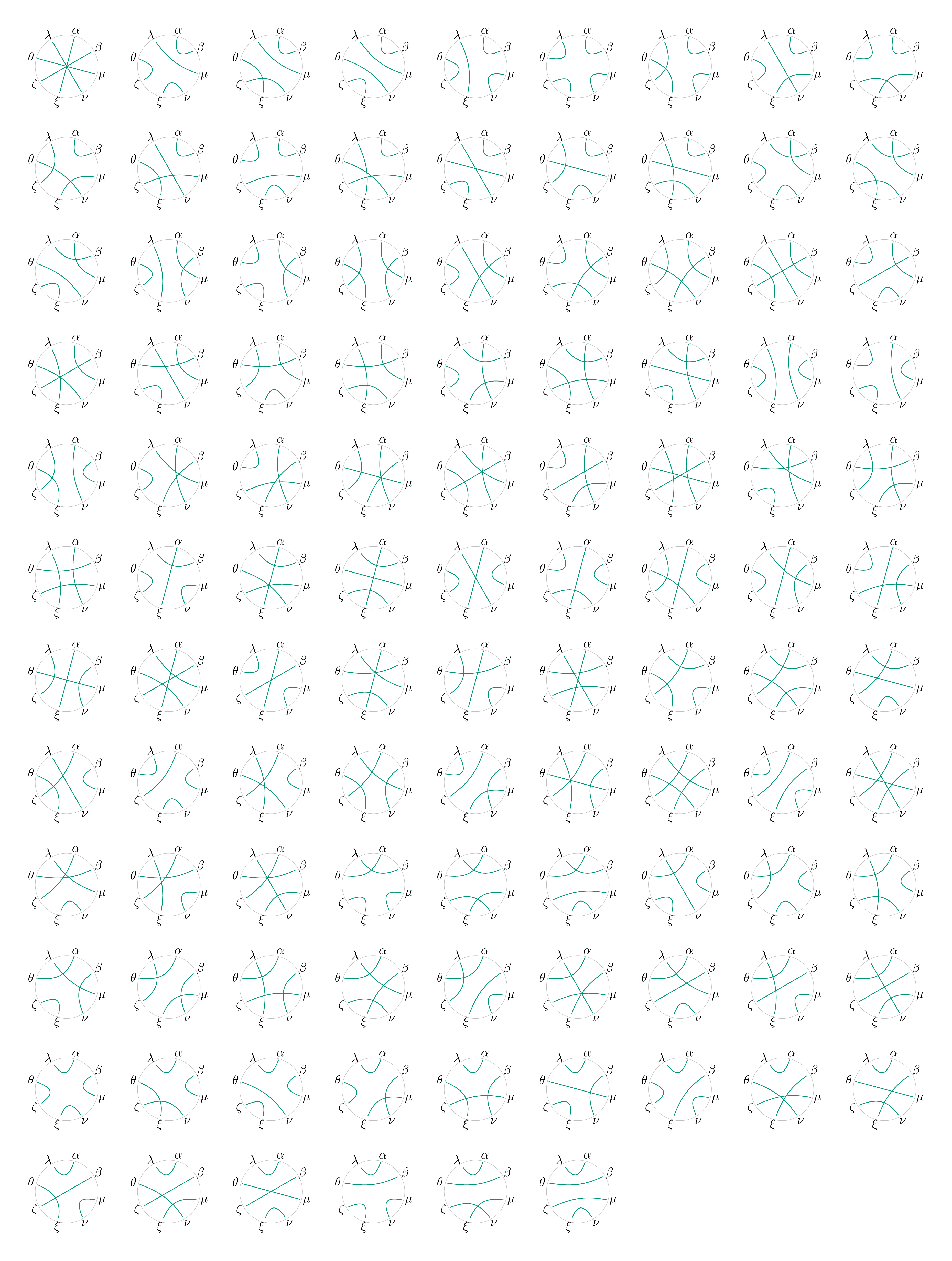}
\caption{\label{Fig:Eight} The $7!!=105$ chord diagrams with eight 
points. These assist to compute $\Tr_V(\gamma^\alpha 
\gamma^\beta\cdots \gamma^\theta\gamma^\lambda)$
in any dimension with diagonal metric of any signature. The sign of a diagram $\chi$
is $(-1)^{\#\{\text{simple crossings of }\chi\}}$.
Thus, the 'pizza-cut' diagram in the upper left
corner that appears as a summand in the normalized trace
$\langle \alpha\beta \mu\nu\xi\zeta \theta\lambda\rangle$ 
evaluates to $(-1)^{1+2+3}  g^{\xi\alpha} 
g^{\zeta\beta}
g^{\theta\mu}
g^{\lambda\nu}   $, where $\langle \cdots \rangle=(1/\dim V)\Tr_V(\cdots)$
} 

\end{figure}

\begin{figure}
\includegraphics[width=10cm]{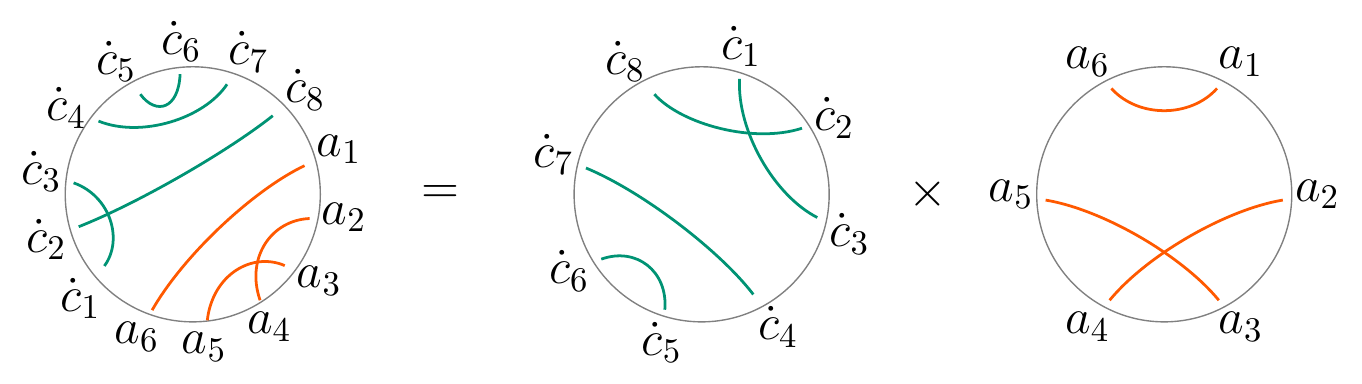}
\caption{Splitting of a chord diagram for indices 
of a mixed signature. 
One of the diagrams 
appearing in the computation of
$\langle a_1\cdots a_6 \dot c_1\cdots \dot c_8 \rangle$
14 points, all of which split into two (of 8 and 6 points). 
The equality of diagrams means 
equality of the product of the bilinears $g^{a_ia_j}$
and $g^{\dot c_l \dot c_k}$ determined by the depicted chords
and the signs for simple crossing
\label{Fig:Cutting} 
}
\end{figure}
It will be convenient to denote by $\mtr{CD}_{2n}$
the \textit{set of $2n$-point chord diagrams} and to associate a tensor $\chi^{\mu_1 \dots \mu_{2n} }$ with $\chi\in \CD{2n} $ 
and an index set 
$\mu_1,\ldots, \mu_{2n}\in \Delta_d$:
\begin{align}
\label{eq:tensorchi}
\chi^{\mu_1 \dots \mu_{2n} }
=
(-1)^{\#\{\text{crossings of chords in }\chi\}} 
\prod_{\substack{i,j =1\\ i\sim_\chi j} }^{2n} g^{\mu_{i}\mu_{j}} \,\,.
\end{align}
This $\chi$-tensor is a
version of the chord diagram $\chi$ whose $i$-th point 
is decorated with the spacetime index $\mu_i$; thus $\chi^{\mu_1 \dots \mu_{2n} }$ depends on the
dimension, although it is not
explicitly so denoted. 
All known identities for traces of gamma 
matrices can be stated in terms of these tensors,
for instance $\Tr_{\C^4}(\gamma^{\mu_1}\gamma^{\mu_2}\gamma^{\mu_3}\gamma^{\mu_4}) =  4 ( g^{\mu_1\mu_2} g^{\mu_3\mu_4}
-g^{\mu_1\mu_3} g^{\mu_2\mu_4}
+g^{\mu_1\mu_4} g^{\mu_2\mu_3}
)$ in four dimensions: If $\theta,\xi,\zeta$ denote the
three $4$-point chord diagrams,
one can rewrite in terms of their corresponding tensors 
\begin{align} \label{eq:chitensors} 
\raisetag{4.5\normalbaselineskip}
\raisebox{-.43\height}{\includegraphics[height=2.1cm]{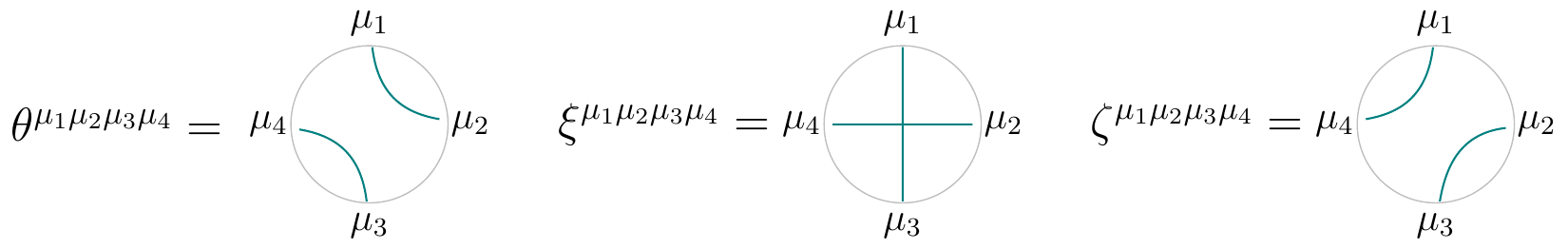} } 
\end{align}
the aforementioned trace identity as   
\[\Tr_V(\gamma^{\mu_1}\gamma^{\mu_2}\gamma^{\mu_3}\gamma^{\mu_4}) = \dim V ( \theta^{\mu_1\mu_2\mu_3\mu_4}
+\xi^{\mu_1\mu_2\mu_3\mu_4}
+\zeta^{\mu_1\mu_2\mu_3\mu_4}
)\,.\] 
For small $n$, this seems to be a 
heavy notation, which however will pay off 
for higher values (the double factorial growth 
$\#\CD{2n} = (2n-1)!!$ notwithstanding).

\subsection{Traces of random matrices}
The aim of this subsection is to compute traces of words of the form
$\TrM(k_{I_1}\cdots k_{I_{2t}})$ 
using the isomorphism $M_N(\C)=\mathbf N \otimes \bar{\mathbf N}$ (being
 $\mathbf N$ the fundamental representation) 
at the level of the operators. By \cite{BarrettGlaser}, 
\begin{align*}
k_I= 
K_I\otimes 1_N + e_I \cdot (1_N\otimes K^T_I)\,,\quad e_I=\pm\,.
\end{align*}
The sign $e_I$ is determined by Table \ref{tab:HorL}.

\begin{prop} \label{thm:RandomPart}
For any $r\in \N$
\begin{align} \label{eq:RandomPart}
\TrM(k_{I_1}\cdots k_{I_{r}})
=\sum_{\Upsilon\in \mathscr P_{r}} \mathrm{sgn}(I_\Upsilon) \cdot \TrN\big( K_{I_{\Upsilon^{\mtr c}}}\big) \cdot \TrN \big[ (K^T)_{I_\Upsilon} \big]\,,
\end{align}
where 
\begin{itemize}
 \itemB $\TrN$ and $\TrM$  are  the traces on $\mtr{End}(\mathbf N)$ and $\mtr{End}(M_N(\C)) $, 
 respectively

 \itemB $\mathscr P_{r}$ is the power 
set $2^{\{1,\ldots,r\}}$ of $\{1,\ldots,r\}$,
and  $\Upsilon^{\mtr c}=\{1,\ldots,r\}\setminus \Upsilon$

\itemB $\mathrm{sgn}(I_\Upsilon)$ is $(-1)^{\#\{\text{commutators appearing in all the $k_{I_j}$ with $j\in\Upsilon$}\}}$, 
that is   \[\mathrm{sgn}(I_\Upsilon)=\big(\prod_{i\in \Upsilon} e_{I_i}\big) \in \{-1,+1\}
           \,
          \]

\itemB and, finally, the cyclic order  $(\ldots \to r \to 1\to 2 \to 3 \to\ldots)$
on the set $\{1,\ldots,r\}$, which
can be read off from the trace in the LHS 
of eq. \eqref{eq:RandomPart},
induces a cyclic order on a given subset $\Xi=\{b_1,\ldots,b_{\xi}\} \in \mathscr P_{r}$.
Respecting this order, define
\vspace{.152cm}
\begin{itemize}
\itemW 
$K_{I_\Xi}=K_{I_{b_1}}K_{I_{b_2}}\cdots K_{I_{b_\xi}}$ and 
\itemW $(K^T)_{I_\Xi}=K^T_{I_{b_1}}K^T_{I_{b_2}}\cdots K^T_{I_{b_\xi}}=
(K_{I_{b_\xi}}\cdots K_{I_{b_2}} K_{I_{b_1}})^T$\,.

\end{itemize}

\end{itemize}
\end{prop}

\begin{proof}  By induction on the number $r-1$ of products,
we prove first that 
\[k_{I_1}\cdots k_{I_r}=
\sum_{\Upsilon\in \mathscr P_{r}} \prod_{i\in \Upsilon} \mtr{sgn}(I_{\Upsilon})   
K_{I_{\Upsilon^{\mtr c}}} \otimes (K^T)_{I_\Upsilon} \,.
\]
The statement holds for $r=2$, by direct computation; we now prove that 
the statement being true for $r$
 implies its veracity for $r+1$. 
 In the first line of the RHS of the 
 next equation we use the assumption and then 
 directly compute:
 \allowdisplaybreaks[3]
\begin{align*}
 (k_{I_1}\cdots k_{I_r})k_{I_{r+1}} & =
\prod_{w=1}^r \big[
 K_{I_{w}}\otimes 1_N + e_{I_{w}} \cdot (1_N\otimes K^T)_{I_{w}}\big]
\\
& \qquad \cdot 
 \big(
K_{I_{r+1}}\otimes 1_N + e_{I_{r+1}} \cdot (1_N\otimes K^T)_{I_{r+1}}\big)
 \,
\\
& = \bigg(
\sum_{\Upsilon\in \mathscr P_{r}} \Big(\prod_{i\in \Upsilon} e_{I_i}  \Big) 
K_{I_{\Upsilon^{\mtr c}}} \otimes (K^T)_{I_\Upsilon} \bigg)
\\&
\qquad \cdot 
 \big(
K_{I_{r+1}}\otimes 1_N + e_{I_{r+1}} \cdot (1_N\otimes K^T)_{I_{r+1}}\big)
\,\\
&=
\sum_{\Upsilon\in \mathscr P_{r}} \Big(\prod_{i\in \Upsilon} e_{I_i}\Big)
K_{I_{\Upsilon^{\mtr c}}}K_{I_{r+1}} \otimes (K^T)_{I_\Upsilon}  
\\
& \qquad +  
\sum_{\Upsilon\in \mathscr P_{r}}\Big( \prod_{i\in \Upsilon} e_{I_i} \Big)e_{I_{r+1}} 
K_{I_{\Upsilon^{\mtr c}}} \otimes (K^T)_{I_\Upsilon} K^T _{I_{r+1}}
  \big)\, \\
&=
\sum_{\Theta\in \mathscr P_{r+1}}\Big( \prod_{i\in \Theta} e_{I_i}\Big)
  \cdot 
K_{I_{\Theta ^{\mtr c}}} \otimes (K^T)_{I_\Theta}\,.
\end{align*}%
 \allowdisplaybreaks[0]%
To the last equality one arrives by considering that  
any set  $\Theta \in \mathscr P_{1+r}$ 
either contains $r+1$ (thus $\Theta =\Upsilon \cup \{r+1\}$  
for some  $\Upsilon   \in \mathscr P_{r}$) or does not ($\Theta =\Upsilon\in \mathscr P_{r} $). 
These two sets are listed in the sum after the third 
equal sign (concretely,
the second term and the first one, respectively). 
Then, it only remains to take traces
\begin{align*}
\TrM(k_{I_1}\cdots k_{I_{r}})  
& =\sum_{\Upsilon\in \mathscr P_{r}} \mathrm{sgn}(I_\Upsilon) \cdot
\Tr_{\mathbf N\otimes \bar{\mathbf N}}  \big(K_{I_{\Upsilon^{\mtr c}}} \otimes (K^T)_{I_\Upsilon} \big) \\ 
& = \sum_{\Upsilon\in \mathscr P_{r}}  \mathrm{sgn}(I_\Upsilon) \cdot
\TrN \big[ K_{I_{\Upsilon^{\mtr c}}} \big] \cdot \TrN \big[ (K^T)_{I_\Upsilon} \big]\,. \qedhere \end{align*} 
\end{proof}

\subsection{The general structure of $\Tr D^m$}

From the analysis of the gamma matrices one infers
that for a polynomial $f(x)=\sum f_t x^t$ 
the spectral action selects only the 
even coefficients $\Tr f(D) =\sum f_{2t} \Tr(D^{2t})$.
In order to compute the spectral action of 
any matrix geometry we only need to 
know the traces of the even powers, which
we now proceed to compute.

\begin{prop} \label{thm:EvenPowers}
Given a collection of multi-indices $I_i\in \Lambda_d^-$, let $2n$
denote the total of indices, $2n=2n(I_1,\ldots,I_{2t}):=|I_1|+\ldots + |I_{2t}|$. 
The even powers of the Dirac operator satisfy
\begin{align} \label{eq:EvenPowers}
 \frac1{\dim V}  \Tr(D^{2t}) &=
\sum_{ I_1,\ldots, I_{2t}\in\Lambda_d} \bigg\{ 
\sum_{  \chi \in \CD{2n}  }   \chi^{I_1\ldots I_{2t}}\\ \nonumber
& \hspace{2.2cm} \times
\Big[ \sum_{\Upsilon\in \mathscr P_{2t}} \mathrm{sgn}(I_\Upsilon) \cdot
\TrN (K_{I_{\Upsilon^{\mtr c}}} ) \cdot  \TrN ((K^T)_{I_\Upsilon}) \Big] \bigg\}\,,
\end{align}
in whose terms the spectral 
action $S(D)=   \Tr f(D)=
 \sum_t f_{2t} \Tr(D^{2t})$
can be completely evaluated. \par
\end{prop}

\begin{proof}
 For $t\in \N$,
  \begin{align} \label{eq:D2t}
 \frac{1}{\dim{V}}\Tr(D^{2t})&= \frac{1}{\dim{V}} \Tr\bigg[\big(\sum_{I\in\Lambda_d} 
 \Gamma^I \otimes k_I \big)^{2t}\bigg] \\
 & =  \sum_{I_1,\ldots I_{2t}\in\Lambda} \nonumber
   \frac{1}{\dim{V}}\Tr_V(\Gamma^{I_1}\cdots \Gamma^{I_{2t}}) \TrN(k_{I_1} \cdots k_{I_{2t}} \big) \\
   &= \sum_{I_1,\ldots I_{2t}\in\Lambda}   \langle I_1 \ldots I_{2t} \rangle \TrM(k_{I_1} \cdots k_{I_{2t}} ) \,.\nonumber
 \end{align}
 One uses then eq. \eqref{eq:Chords} and Proposition \ref{thm:RandomPart}
 with the notation of eq. \eqref{eq:tensorchi}.
\end{proof}

Notice that since the indices $\mu_i$ 
of a multi-index $I=(\mu_1\ldots \mu_{|I|})\in \Lambda_d$ are 
pairwise different, the traces of the gamma matrices
greatly simplify. This also ensures that 
there are no contractions 
between indices of the same $k$-operator, say $g^{\mu\nu}  
k_{\mu \nu\ldots }$ ($k$'s 
with repeated indices do not exist).\par

In even dimension $d$, the Dirac operator
is spanned by the number $\kappa(d)$ of independent
odd products of gamma matrices. This equals 
$
\kappa(d)= \# (\Lambda^-_d)  =
\binom{d}{1}+
\binom{d}{3} 
+\ldots + \binom{d}{d-1}$
which can be rearranged (using Pascal's identity) as 
$
\kappa(d)= 
\binom{d-1}{0}+\binom{d-1}{1}
+\ldots + \binom{d-1}{d-2}+\binom{d-1}{d-1}=2^{d-1}$.
The Dirac operator has then as many `matrix coefficients' 
and is therefore parametrized by (what will turn out to be a subspace of)
$ M_{N}(\C)^{\oplus \kappa({d})}$. In this manner,
the `random spectral action' \eqref{eq:Z_spectral}
becomes a $\kappa(d)$-tuple matrix model.

\begin{definition}
Given integers $t,n \in \N$ 
(interpreted as in Prop. \ref{thm:EvenPowers}) and a chord diagram $\chi\in \CD{2n}$,
its \textit{action (functional)} 
$\mathfrak a_n(\chi) $
is a $\C$-valued functional on the matrix space
$ M_{N}(\C)^{\oplus \kappa({d})}  $ 
defined by 
\begin{align}\label{eq:afrak}
\mathfrak a_n(\chi)[\boldsymbol K]& = \sum_{\substack{I_1,\ldots,I_{2t}\in \Lambda_d^- \\[2pt]
2n= \sum_i |I_i|}} 
\chi^{I_1\ldots I_{2t}} \Big[ \sum_{\Upsilon\in \mathscr P_{2t}} \mathrm{sgn}(I_\Upsilon) \cdot
\TrN (K_{I_{\Upsilon^{\mtr c}}} ) \cdot  \TrN ((K^T)_{I_\Upsilon}) \Big] \raisetag{1.2\normalbaselineskip}
\end{align}
for $\boldsymbol K=\{K_{I_i} \in M_N(\C) \mid I_i\in \Lambda_d^-\} \in M_{N}(\C)^{\oplus \kappa({d})}$. 
We often shall omit the dependence 
on the matrices and write only $\mathfrak a_n(\chi)$.
We define the \textit{bi-trace functional} as 
a sum over the non-trivial subsets $\Upsilon$
in eq. \eqref{eq:afrak}
\begin{align}\raisetag{1.2\normalbaselineskip}
\phantom{N\cdot }\mathfrak b_n(\chi)[\boldsymbol K]& = \sum _{\substack{I_1,\ldots,I_{2t}\in \Lambda_d^- \\[2pt]
2n= \sum_i |I_i|}} 
\chi^{I_1\ldots I_{2t}} \Big[ \sum_{\substack{\Upsilon\in \mathscr P_{2t}\\ 
\Upsilon,\Upsilon^ c \neq \emptyset }} \mathrm{sgn}(I_\Upsilon) \cdot
\TrN (K_{I_{\Upsilon^{\mtr c}}} ) \cdot  \TrN ((K^T)_{I_\Upsilon}) \Big] 
\end{align}
and the \textit{single trace functional}
$\mathfrak{s}_n(\chi)$ via $ \mathfrak{a}_n(\chi)=N\cdot \mathfrak s_n(\chi)+ \mathfrak b_n(\chi)$ .
The factor $N$ ensures that  
$\mathfrak{s}_n$ does not depend on $N$ (cf. Sec. \ref{sec:leading}).
\end{definition}
The restriction $2n= \sum_i |I_i|$ allows one 
to exchange the sums over the multi-indices
$I$ and the chord diagrams in eq. \eqref{eq:EvenPowers}. 
Then one can restate Proposition \ref{thm:EvenPowers} as
$(1/\dim V)\Tr(D^{2t})=N\mathcal S_{2t}+\mathcal{B}_{2t}$, where
 \begin{align} \label{eq:Sdost}
\mathcal S_{2t}&= \sum_{n=t}^{t\cdot(d-1)}  
 \sum_{\chi \in \CD{2n}} \mathfrak s_{n}(\chi) \,, \\
 \mathcal{B}_{2t}&= \sum_{n=t}^{t\cdot(d-1)}\sum_{\chi \in \CD{2n}}\mathfrak b_{n}(\chi) \,. \label{eq:Bdost}
 \end{align}
The parameter $n$ lists all the 
numbers of points $2n=2t,2(t+1),\ldots,2t\cdot(d-1)$ that chord diagrams contributing to 
$\Tr(D^{2t})$ can have in dimension $d$.  
\par

In view of  eqs. 
\eqref{eq:Sdost} and \eqref{eq:Bdost}, all boils down to computing the 
 single trace $\mathfrak s_n(\chi)$ and multiple trace
 part $\mathfrak b_n(\chi)$ of chord diagrams.
 We begin with the former.

\subsection{Single trace matrix model in the spectral action---manifest $\mathcal O(N)$}
\label{sec:leading}

For geometries with even KO-dimension $s=q-p$, the spectral
action's \textit{manifest} leading order in $N$  can be found from 
last proposition (see Sec.  
\ref{sec:freeP}). 
Aiming at their large-$N$ limit we state 
the following

\begin{cor}
\label{thm:largeN}
Let $d=q+p$ (thus $s$) be even.
The spectral action \eqref{eq:Z_spectral}
for the fuzzy $(p,q)$-geometry with 
Dirac operator $ D=D\hp{p,q}$
satisfies, for any polynomial $f(x)=\sum_{1 \leq r\leq m} f_r x^r$, the following:
\begin{align}
 \frac1{\dim V}\Tr f(D)
 &=  \frac1{\dim V}\sum_{1 \leq 2t\leq m}  f_{2t}\Tr\big([D\hp{p,q}]^{2t}\big) =N \sum_{1 \leq 2t \leq m} f_{2t} \mathcal S_{2t} + \mathcal{B} \nonumber
\end{align}
where $2n(I_1,\ldots,I_{2t})=\sum_i |I_i|$. 
Here $\mathcal{B}$ stands for products of two traces,
whose coefficients are all independent of $N$,
with  $\mathcal S_{2t}$ given by eq. \eqref{eq:Sdost}
and for $\chi\in \CD{2n}$, 
\begin{align}\label{eq:sgeneral}
\mathfrak s_n(\chi)
& = \sum_{\substack{I_1,\ldots,I_{2t}\in \Lambda_d^- \\[2pt]
2n= \sum_i |I_i|}} 
\chi^{I_1\ldots I_{2t}}
\Big\{ \TrN(K_{I_1}K_{I_2}\cdots K_{I_{2t}} )  
\\[-20pt]
& \hspace{3.5cm}{+(e_{I_1}\cdots e_{I_{2t}})\TrN(K_{I_{2t}}K_{I_{2t-1}}\cdots K_{I_{1}}) }\nonumber
\Big\}\,.
\end{align}
\end{cor}

\begin{proof}
  For any given collection of 
  multi-indices $I_1,\ldots I_{2t}$ 
  one selects in Proposition \ref{thm:EvenPowers} 
  the two sets $\Upsilon=\emptyset $ and $\Upsilon=\{1,\ldots, 2t\}$
  which correspond with the first and 
  second summands between curly brackets. 
 The overall $N$-factor corresponds to 
 $\TrN(1_N)$. 
  Clearly any other subset $\Upsilon$
  has no factor of $N$ since is of the form
  \[
\TrN \bigg(\prod^{\to}_{i\in \{1,\ldots, 2t\}\setminus\Upsilon} K_{I_i} \bigg) \times \TrN \bigg(\prod^{\leftarrow}_{i\in \Upsilon} K_{I_i}\bigg)\,,
\]
where none of the products is empty.
The arrows indicate the order in which the product is performed ($\to$ 
preserves it and 
$\leftarrow$ inverts it, but this irrelevant to the point of this corollary).  Therefore no trace of $1_N$ appears.
One easily arrives then to eq. \eqref{eq:sgeneral} by
excluding from $\mathfrak{a}_n(\chi)$ all 
the non-trivial sets, that is $\Upsilon, \Upsilon^c \neq \emptyset$. 
\end{proof}

\subsection{Formula for $\Tr D^2$ in any dimension and signature}

We evaluate in this section $\Tr(D^2)$
for Dirac operators $D$ of a fuzzy geometry in any
signature $(p,q)$.  

\begin{prop} \label{thm:Dsquare}
The Dirac operator of a fuzzy geometry
of signature $(p,q)$ satisfies for odd $d=p+q$
\begin{align*}
\frac{1}{\dim V}\Tr\big([D\hp{p,q}]^2\big)
= 2
\sum_{I\in\Lambda_d}(-1)^{u(I)+\binom{|I|}{2}} 
\big[N\cdot \TrN(K_I^2)+
  e_I(\TrN K_I)^2\big]\,,
\end{align*}
being $ u(I)$ 
the number of spatial indices in $I$.
If $d$ is even, then the sum is only over $I\in\Lambda_d^-$,
and the  expression reads
\begin{align*}
\frac{1}{\dim V}\Tr\big([D\hp{p,q}]^2\big)
=
2
\sum_{I\in\Lambda_d^-}(-1)^{u(I)+r(I)-1} 
\big[N\cdot \TrN(K_I^2)+
  e_I(\TrN K_I)^2\big]\,,
\end{align*}
with $|I|=2r(I)-1$.
\end{prop}

\begin{proof}

In order to use eq. \eqref{eq:D2t},
notice that $\langle I_1 I_2 \rangle \neq 0$
implies that $I_1,I_2\in \Lambda_d$ have the 
same cardinality. If that were not the 
case (wlog $|I_1| > |I_2| $), since any non-zero term from 
$\langle I_1 I_2 \rangle$ arises from 
a contraction of indices (cf. eq. \eqref{eq:Chords}), a 
different number of indices would imply
that there is a chord connecting two indices  $\mu_i,\mu_j$ of 
$I_1=(\mu_1\ldots, \mu_r)$. 
Since $I_1\in \Lambda_d$, 
those indices are different, so $g^{\mu_i\mu_j}=0$.
Thus only chord diagrams for pairings $g^{\mu\nu}$ of indices
with $\mu\in I_1$ and $\nu\in I_2$ survive. Since
the indices  of $I_1$ and $I_2$ are strictly increasing,
both ordered sets have to agree.
This means 
that we only have to care about evaluating $\langle I I'\rangle $,
with $I'=(\mu_1',\ldots,\mu'_w)$ being a copy of 
$I=(\mu_1,\ldots,\mu_w)$, i.e. $\mu_i'=\mu_i$. Since
this last equality is the only possible index repetition
  \begin{align} \nonumber
\langle  {\mu_1,\ldots,\mu_w\mu_1',\ldots,\mu'_w}\rangle &= 
\sum_{\substack{ 2w\text{\scriptsize -pt chord} 
\\ \text{\scriptsize diagrams }  \chi  } }
(-1)^{\text{cr}(\chi)} 
\prod_{\substack{i,j \\ i\sim j} } g^{\mu_{i}\mu'_{j}} \delta_{ij}= 
(-1)^{\text{cr}(\pi)} \prod_{\mu=1}^w  g^{\mu\mu} 
\nonumber 
\end{align}
 where $\pi$ is the (pizza-cut)
 diagram with longest chords, that is joining 
 antipodal points. 
The number of crossings $\mtr{cr}(\mtr{\pi})$
 is $\binom{w}{2}$.  
 An additional sign $(-1)^{u}$ 
comes from $\prod_{i=1}^w  g^{\mu_i\mu_i} $,
being $u\leq q$ the number of spatial 
indices in $I$, yielding
\begin{align} \label{eq:II}
 \langle I_1 I_2 \rangle = \delta^{I_1}_{I_2} 
 (-1)^{u(I_1)+\binom{w}{2}}   \,.
\end{align}
From eq. \eqref{eq:D2t} with $t=1$ one has
\allowdisplaybreaks[1]
 \begin{align*}
\frac{1}{\dim V}\Tr\big([D\hp{p,q}]^2\big)
&=
 \sum_{I_1   I_{2}\in\Lambda}   \langle I_1 I_2 \rangle \TrM(k_{I_1} k_{I_{2}} \big)\\ 
&  =
  \sum_{I\in\Lambda}
  (-1)^{u(I)+\binom{|I|}{2}}
  \Tr_{\mathbf N\otimes \bar {\mathbf N}}\Big[\big(K_I\otimes 1_N +e_I \otimes K_I^T
  \big)^2\Big]\\
  & =  \sum_{I\in\Lambda}
  (-1)^{u(I)+\binom{|I|}{2}}\Big[
\TrN(K_I) \TrN (1_N)+ \TrN (1_N)\TrN (K_I^T) \\ 
& \hspace{3.5cm} + 2 e_I \TrN(K_I) \TrN(K_I^T)\Big]\,.
 \nonumber
 \end{align*}%
 \allowdisplaybreaks[0]%
 In the second equality we used eq. \eqref{eq:II}. 
 The third one follows from Proposition \ref{thm:RandomPart}.
For $p+q$ even, the sum runs only
over $I\in \Lambda_d^-$. In the sign appearing in eq. \eqref{eq:II},
$\binom{|I|}{2}$  
could then be replaced by 
$r(I)-1 $ with $2r(I)-1=|I| $,
for $\binom{|I|}{2} \equiv r-1$ (mod 2).
\end{proof}

\section{Two-dimensional fuzzy geometries in general signature} \label{sec:d2}

We compute traces of $D^2$, $D^4$, $D^6$ 
for 2-dimensional fuzzy geometries general signatures. 
Concretely, for $d=2$ the spinor space is $V=\C^2$.

\subsection{Quadratic term}\label{sec:d2D2}
For a metric $g=\diag(e_1,e_2)$  
notice that $e_\mu=(-1)^{u(\mu)}$. Therefore, by Proposition \ref{thm:Dsquare} one gets
\begin{align}
\frac{1}{4}
\Tr\big[(D\hp{p,q})^2\big]  \nonumber
& = 
\sum_\mu (-1)^{u(\mu)} N \cdot \Tr_N (K_\mu ^2)
 + \sum_{\mu}[\TrN (K_\mu) ]^2
\\
&=  N \sum_{\mu,\nu}  g^{\mu\nu} 
\Tr_N (K_\mu K_\nu)
 + \sum_{\mu}[\TrN (K_\mu) ]^2 \,,\label{eq:GeneralD2d2}
\end{align}
where $u(\mu) =0$ if $\mu $ is 
time-like and if its spatial,
$u(\mu) =1$. Case by case, 
\begin{align} \nonumber
 \begin{cases}
 \sum\limits_{a=1}^{ 2} \big(-N \cdot \TrN(L_{\dot a}^2)    + [\TrN(L_{\dot a}) ]^2 \big)& \mbox{for } (p,q)=(0,2) \\[9pt]
 N \cdot \TrN(H^2-L^2)    +   [\TrN(H) ]^2 +[\TrN(L) ]^2  & \mbox{for } (p,q)=(1,1) \\[2pt]
  \sum\limits_{ a=1}^{ 2}\big(+ N \cdot \TrN(H_a^2)    +  [\TrN(H_a)]^2  \big)& \mbox{for } (p,q)=( 2,0)  
\end{cases}
\end{align}
reproducing some of the formulae reported in \cite[App. A]{BarrettGlaser}. 

\subsection{Quartic term}\label{sec:d2D4}
In \textit{op. cit.} also the quartic term for $d=2$ was computed.
We recompute for a general $d=2$ geometry of arbitrary signature
 aiming at illustrating the chord diagrams at work.
Since $d=2$, multi-indices $\mu \in \Lambda^-_2$ are just spacetime indices 
$\mu=1,2$. Hence, after Proposition \ref{thm:EvenPowers},
\allowdisplaybreaks[1]
\begin{salign}
 \frac{1}{2} \Tr(D^4)
 &=\sum_{\mu_1,\ldots,\mu_4\in \Lambda_2^-} \bigg(
\raisebox{-.43\height}{\includegraphics[height=1.6cm]{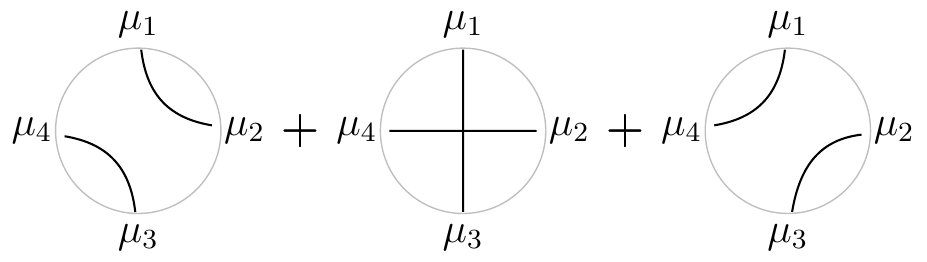} } \bigg) \\
& \times \bigg\{
N\cdot  \TrN( K_{\mu_1}K_{\mu_2}K_{\mu_3}  K_{\mu_4})
\\
&+\sum_{i=1}^4    e_{\mu_i} 
\TrN (K_{\mu_1}\cdots \widehat{ K_{{\mu_i}}} \cdots K_{\mu_4} ) \TrN (K_{\mu_i})
\\ & + \sum_{1\leq i<j\leq 4}^{\phantom .}
e_{\mu_i}e_{\mu_i} 
\big[
\TrN (K_{\mu_v} K_{\mu_w}) \TrN (K_{\mu_i} K_{\mu_j}) \numerada
\big]
\\
& +\sum_{i=1}^4 e \cdot e_{\mu_i} \cdot 
 \TrN (K_{\mu_i})
\TrN (K_{\mu_4}\cdots \widehat{ K_{{\mu_i}}} \cdots K_{\mu_1} )
\\
&
+ e \cdot N\cdot \TrN(K_{\mu_4}K_{\mu_3}K_{\mu_2} K_{\mu_1}) \bigg \}\,,
\end{salign}%
\allowdisplaybreaks[0]%
with $e=e_{\mu_1} e_{\mu_2}e_{\mu_3}e_{\mu_4}$
and $\{i,j,v,w\}=\Delta_4$.
In the first line, the value of the chord diagrams is $g^{\mu_1\mu_2}g^{\mu_3\mu_4}-
g^{\mu_1\mu_3}g^{\mu_2\mu_4}
+g^{\mu_1\mu_4}g^{\mu_2\mu_3}$,
the signs $e_\mu=\pm$ appearing in $g=\diag(e_1,e_2)$ 
being determined by 
$(p,q)$. 
Summing over all indices, one gets 
\begin{align}
  \frac{1}{4} \Tr\big([D\hp{p,q}]^4\big)\label{eq:GeneralDcuatro}
 &= N\Big[ \TrN(K^4_1)+\TrN(K^4_2) \\[3pt]
 &\nonumber +4e_1e_2 \TrN(K_1^2K_2^2) -2e_1e_2 \TrN(K_1K_2K_1K_2)\Big] \\[6pt]
 &+4\Big\{[\TrN(K_1 K_2 )]^2+ \sum_{\mu=1,2} \TrN {K_\mu} \cdot \TrN\big[K_\nu(e_1K_1^2+e_2K^2_2)\big] \Big\} \nonumber \\
  &+3\sum_{\mu=1,2}\big[\TrN(K_\mu^2)\big]^2 
 +2e_1e_2 \TrN (K_1^2)\cdot \TrN (K_2^2)\,. \nonumber
\end{align}
One gets directly the results of \cite[App. A.3, A.4, A.5]{BarrettGlaser}
by setting 
\begin{align}
K_1=\begin{cases}
H_1 & \text{for }(p,q)= (2,0) \, \and \,(1,1)\,, \\ 
L_{\dot 1} & \text{for }(p,q)= (0,2)\,,
\end{cases}
\end{align}
and 
\begin{align}
K_2=\begin{cases}
H_2 & \text{for }(p,q)= (2,0)\,, \\ 
L_{\dot 2} & \text{for }(p,q)= (1,1)\, \and\, (0,2)\,.
\end{cases}
\end{align}
The conventions for which these hold are 
$(\gamma^\nu)^*=e_\nu\gamma^\nu$ (no sum, $\nu=1,2$).
 
 \subsection{Sextic term}\label{sec:d2D6}
We now compute the sixth-order term.
 \allowdisplaybreaks[1]
\begin{prop}\label{thm:Dsix}
Let $g=\diag(e_1,e_2)$ denote the quadratic 
form associated to the signature $(p,q)$ of a 
2-dimensional fuzzy geometry with Dirac operator $D$. Then
\begin{align*}
 \frac{1}{2}
 \Tr(D^6)&= N  \cdot \mathcal S_6[K_1,K_2] + \mathcal{B}_6[K_1,K_2]\,,
\end{align*}
where the single-trace part is given by  
\begin{align*}
 \mathcal S_6[K_1,K_2]&= 
 2 \cdot\TrN\big\{  e_1K_1^6+ 6e_2 K_1^4 K_2^2-6 e_2 K_1^2(K_1K_2)^2+3e_2  (K_1^2K_2)^2     \\
& \hspace{42.5pt} +
e_2K_2^6+ 6e_1 K_2^4K_1^2 - 6 e_1 K_2^2(K_2K_1)^2 + 3e_1 (K_2^2K_1)^2    \big\} 
\end{align*}
and the bi-trace part is
\begin{align*}
\mathcal{B}_6[K_1,K_2]&= 
6 \TrN(K_1) \bigg\{ 2 \TrN(K_1^5)+2 \TrN(K_1 K^4_2) \\
& \hspace{2.2cm}
+    6 e_1 e_2 \TrN(K_1^3K_2^2) -2 e_1 e_2 \TrN(K_1^2K_2K_1K_2) \big] \bigg\}\\
&+\vphantom{\bigg\}}
6 \TrN(K_2) \bigg\{2 \TrN(K_2^5)+2 \TrN(K_2 K^4_1)  \\
& \hspace{2.2cm}
+    6 e_1 e_2 \TrN(K_2^3K_1^2) -2 e_1 e_2 \TrN(K_2^2K_1K_2K_1) \big] \bigg\} \\
& +\vphantom{\bigg\}} 48 \TrN (K_1 K_2  )  \cdot
 \big[ e_1 \TrN ( K_1^3K_2)+ e_2 
\TrN (  K_2 ^3 K_1) \big] \\ &\vphantom{\bigg\}}
+6
 \TrN (K_1^2  ) \cdot \Big\{e_2  \big[8 \TrN (K_1 ^2 K_2 ^2 )-2 
\TrN (K_2 K_1 K_2 K_1 ) \big] \\ & \hspace{2.2cm}+ e_1  \big[5 \TrN (K_1 ^4   ) +
\TrN (K_2  ^4 ) \big] \Big\} \\ &+6 \TrN (K_2^2 ) \cdot 
 \Big\{  e_1   \big[8 \TrN ( K_1 ^2 K_2 ^2)-2 \TrN (K_1 K_2 K_1 K_2 )  \big]\\
  \vphantom{\bigg\}}& \hspace{2.2cm}+e_2   \big[5 \TrN (K_2^4 ) +\TrN (K_1 ^4)\big] \Big\} \\
  \vphantom{\Big\}}&+ 4  \big(5 [\TrN (K_1^3 )]^2+6 e_1 e_2 \TrN (K_1 K_2^2 ) \TrN (K_1^3 )+9[ \TrN K_1^2 K_2 ]^2\\
 \vphantom{\Big\}}&\quad+5 [\TrN (K_2^3 )]^2+6 e_1 e_2 \TrN (K_1^2 K_2 ) \TrN (K_2^3 )+9[ \TrN K_1 K_2^2]^2 \big)\,.
\end{align*}

\end{prop}

\begin{proof}

The part $\langle \mu_1\ldots\mu_6 \rangle$ concerning the chord diagrams 
evaluates to 
\begin{align} \label{eq:InGs}\raisetag{8\normalbaselineskip}
 \hspace{-.26cm}
\raisebox{-.405\height}{
\includegraphics[width=13cm]{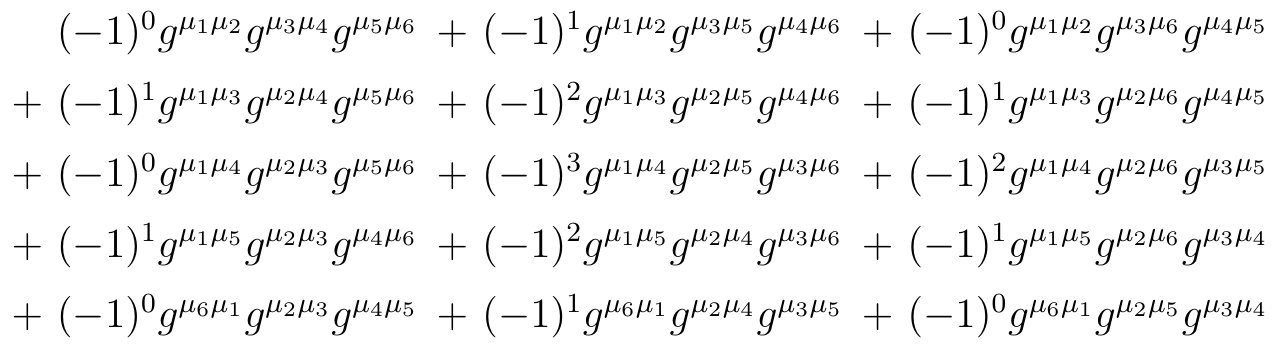}}  
\end{align}
but it is actually useful to depict these terms
as in Figure \ref{fig:Dsix},
\begin{figure}
\[
\langle \mu_1\ldots\mu_6 \rangle =\Vast\{\hspace{-.16cm}
\raisebox{-.49\height}{
\includegraphics[width=10.0cm]{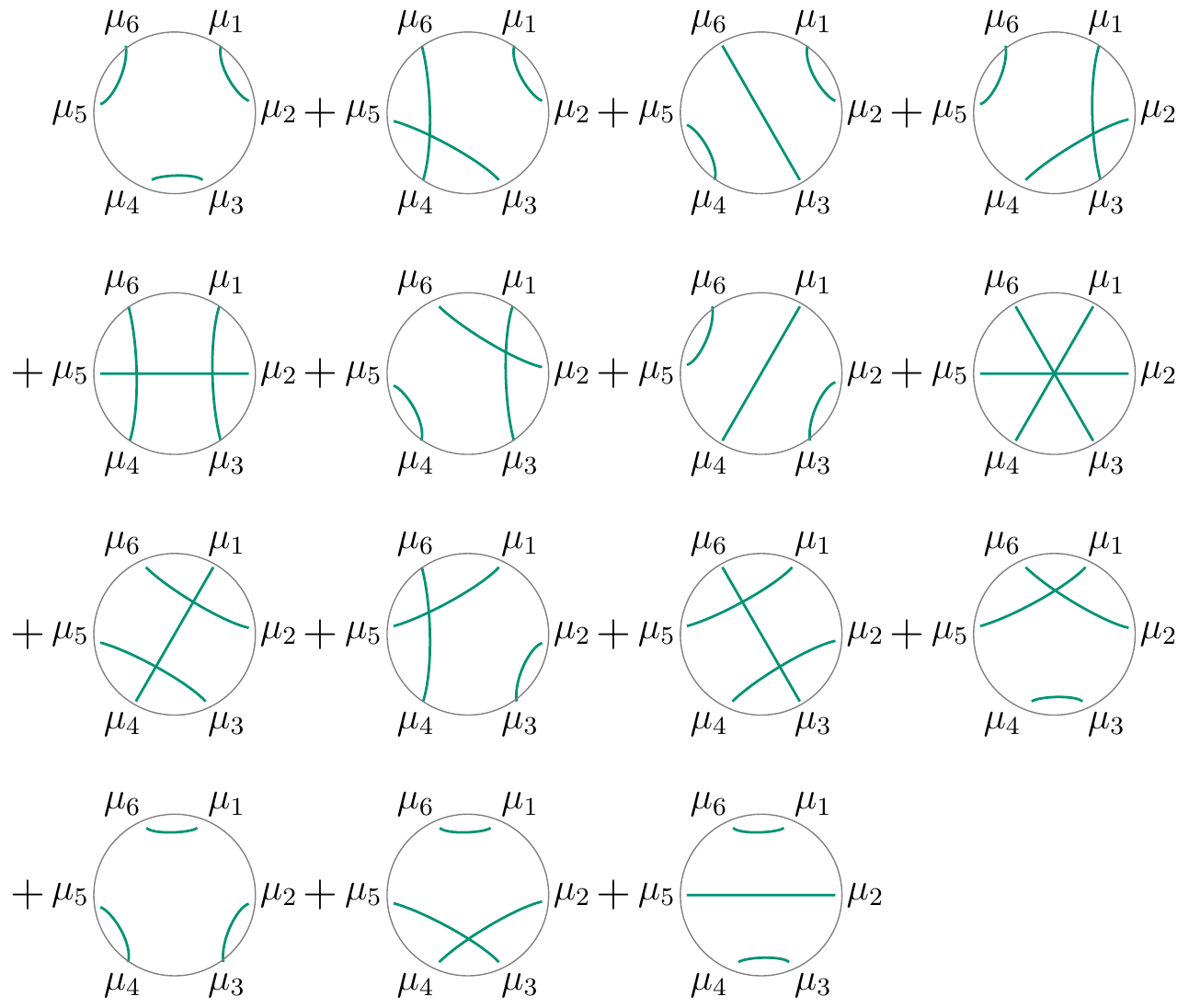}} \Vast\} \]
\caption{On the proof of Proposition \ref{thm:Dsix}\label{fig:Dsix}}
\end{figure}
for then, due to the cyclicity of $\TrN$, one can compute by classes
(modulo $\pi\Z_6/3$-rotations) of 
diagrams. To each class, a Roman number is assigned:
\begin{align} \label{eq:Roman} \raisetag{5.05\normalbaselineskip}
\hspace{-.5cm}
\raisebox{-.483\height}{
\includegraphics[width=\textwidth]{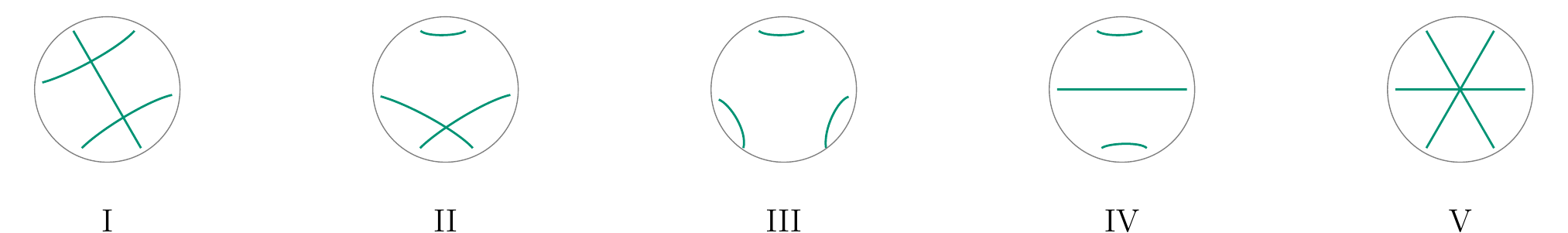}}
\end{align}
One can relabel the $\mu_j$-indices to obtain 
\begin{align*}
 \frac{1}{\dim V}
 \Tr(D^6)&=  \sum_{\chi \in \mtr{CD}_6} \mathfrak{a}( \chi )  
 = 3\mathfrak{a}( \mtr{I} ) + 6\mathfrak{a}( \mtr{II}) +2 \mathfrak{a}( \mtr{III}) +
  3\mathfrak{a}( \mtr{IV} ) +  \mathfrak{a}( \mtr{V} ) \,,
 \end{align*}
 the factors being the multiplicity of each 
 diagram class. The single-trace part $\mathcal S_6$ 
can be computed for each diagram directly (at the 
end of App. \ref{sec:appB} one of these is shown). 
We simplified the notation: $\mtf a_3$ as $\mtf a$,
and similarly we shall write  $\mtf b$ for $\mtf b_3$,
since therein only 6-point diagrams appear (a power $ 2t \geq 2$
of the Dirac operator determines the number of 
points of the chord diagrams only for dimensions $d\leq 2$).

We now compute the bi-trace term. Defining 
\begin{subequations}
\label{eq:OW}
 \begin{align}
 O_{\mu\nu\rho}&=e_{\nu} e_{\rho} \cdot \TrN (K_{\mu} ) \cdot \TrN (K_{\mu} K_{\rho} K_{\nu} K_{\rho} K_{\nu} ) \\
 P_{\mu\nu\rho}&=e_{\nu} e_{\rho} \cdot \TrN (K_{\mu} ) \cdot \TrN (K_{\mu} K_{\rho} K_{\nu}^2 K_{\rho} ) \\
 Q_{\mu\nu\rho}&= e_{\nu} e_{\rho}\cdot  \TrN (K_{\mu} )\cdot  \TrN (K_{\mu} K_{\rho}^2 K_{\nu} ^2 )\\
 R_{\mu\nu\rho}&=   e_{\rho} \cdot \TrN (K_{\mu} K_{\nu} )\cdot  \TrN (K_{\mu} K_{\nu} K_{\rho} ^2 )
\\ 
 S_{\mu\nu\rho}&= e_{\rho} \cdot  \TrN (K_{\mu} K_{\nu} )\cdot  \TrN (K_{\mu} K_{\rho} K_{\nu} K_{\rho} )\\
 T_{\mu\nu\rho}&= e_{\mu} e_{\nu} e_{\rho} \cdot \TrN (K_{\mu} ^2 ) \cdot \TrN (K_{\nu} K_{\rho} K_{\nu} K_{\rho} ) \\
 U_{\mu\nu\rho}&=e_{\mu} e_{\nu} e_{\rho}\cdot  \TrN (K_{\mu} ^2 )\cdot  \TrN (K_{\nu}^2 K_{\rho} ^2 ) \\
 V_{\mu\nu\rho}&= [\TrN (K_{\mu} K_{\nu} K_{\rho} )]^2\\
 W_{\mu\nu\rho}&=   e_{\nu} e_{\rho}  \cdot\TrN (K_{\mu} K_{\nu}^2 ) \cdot \TrN (K_{\mu} K_{\rho}^2  )  \,.
 \end{align}
 \end{subequations}
 we can find by direct computation, that for 
 any of the $6$-point chord diagrams  $\chi$
 there are integers $p_\chi,q_\chi,\ldots, v_\chi,w_\chi$ 
 such that 
 \allowdisplaybreaks[0]
\begin{subequations}
\begin{align}
\mathfrak b (\chi)  =  
 \sum _{\mu,\nu,\rho} 
 & o_\chi  O_{\mu\nu\rho}
 +
 p_\chi  P_{\mu\nu\rho}
 +
 q_\chi  Q_{\mu\nu\rho} \\  \vphantom{ \sum _{\mu,\nu,\rho} }
 +
 &r_\chi  R_{\mu\nu\rho}
 +
 s_\chi  S_{\mu\nu\rho}
 +
 t_\chi  T_{\mu\nu\rho}
 +
 u_\chi  U_{\mu\nu\rho}
 \\
 +& v_\chi  V_{\mu\nu\rho}
 +
 w_\chi  W_{\mu\nu\rho}
 \end{align} 
\end{subequations}

The terms $O,P,Q$ come from the $(1,5)$ partition of $6$, i.e. $\TrN (1\text{ matrix}) \times \TrN (5\text{ matrices} )$;
$R,S,T,U$ terms come from the $(2,4)$ partition 
and $W,V$ from the $(3,3)$ partition of $6$.  
This claim is verified by direct computation;
the proof for $\mathfrak b (\mtr I)$ is presented in Appendix \ref{sec:appB}
and the rest is similarly obtained:
\begin{subequations}
\begin{align} \nonumber
\mathfrak b  (\mathrm{I}) &= + 2 \sum_{\mu,\nu,\rho} 
   \big( 4 O_{\mu\nu\rho}  + 2 P_{\mu\nu\rho}  
        +6 R_{\mu\nu\rho} +  6 S_{\mu\nu\rho}  \phantom{+ 4 O_{\mu\nu\rho}}  \\
  & \hspace{1.323cm} +  2  T_{\mu\nu\rho} +  U_{\mu\nu\rho}
 + 4V_{\mu\nu\rho} + 6 W_{\mu\nu\rho}\big) 
\\
 \nonumber\mathfrak b  (\mathrm{II}) &= -2 \sum_{\mu,\nu,\rho} 
 \big(  2 O_{\mu\nu\rho}  +  2 P_{\mu\nu\rho}  +2  Q_{\mu\nu\rho} 
        +8R_{\mu\nu\rho} + 4 S_{\mu\nu\rho} \\
  &  \hspace{1.323cm} +   T_{\mu\nu\rho} + 2 U_{\mu\nu\rho}
 +4 V_{\mu\nu\rho} +  6W_{\mu\nu\rho}\big) \\
\mathfrak b  (\mathrm{III})&=+ 2 \sum_{\mu,\nu,\rho}  
   \big( 6  Q_{\mu\nu\rho}+ 12 R_{\mu\nu\rho} + 3 U_{\mu\nu\rho}
 + 4V_{\mu\nu\rho} +  6W_{\mu\nu\rho}\big)
\\
\mathfrak b (\mathrm{IV})&=+ 2 \sum_{\mu,\nu,\rho} \nonumber
   \big(   2 P_{\mu\nu\rho}  +4 Q_{\mu\nu\rho} 
        +8R_{\mu\nu\rho} + 4 S_{\mu\nu\rho} \\
  &  \hspace{1.323cm} + 3  U_{\mu\nu\rho}
 + 4V_{\mu\nu\rho} +  6W_{\mu\nu\rho}\big)
\\
\mathfrak b  (\mathrm{V})&= -2 \sum_{\mu,\nu,\rho}  \nonumber
   \big( 6 O_{\mu\nu\rho}  + 6R_{\mu\nu\rho} +  6S_{\mu\nu\rho} \\
  &  \hspace{1.323cm} +   3T_{\mu\nu\rho}  
 + 4V_{\mu\nu\rho} +  6W_{\mu\nu\rho}\big)
\end{align}
\end{subequations}
One then performs the sums explicitly
and arrives to the claim for $\mathcal{B}=\sum_\chi \mathfrak b(\chi)$. 
\end{proof}

\section{Four-dimensional geometries in general signature} \label{sec:d4}
 
We compute now the spectral action for 4-dimensional fuzzy geometries. 

\subsection{The term $\Tr D^2 $}\label{sec:d4D2}

For any four-dimensional 
geometry $p+q=4$ of signature 
$(p,q)$ there are 
eight matrices, $K_1,K_2,K_3,K_4, X_1,X_2,X_3 $ and $X_4 \in M_N(\C)$,  parametrizing the Dirac 
operator \begin{align}
 D\hp{p,q}= 
 \sum_{\mu=1}^4 
 \gamma^\mu\otimes 
 k_{\mu} + 
\Gamma^{\hat \mu} \otimes 
 x_\mu \,.
\end{align}
Here the lower case operators on $M_N(\C)$
are related to said matrices by 
\[
k_\mu=\{K_\mu , \balita \}_{e_\mu}\quad 
\and  \quad
x_\mu=\{X_\mu , \balita \}_{e_{\hat{\mu}}}
\]
where given a sign $\varepsilon=\pm$, the braces $\{A,B\}_\varepsilon=AB+\varepsilon BA$
represent a commutator or an anti-commutator.
As before $\Gamma^{\hat 1}=\gamma^{2}\gamma^{3}\gamma^{4},
\Gamma^{\hat 2}=\gamma^{1}\gamma^{3}\gamma^{4}$, etc,
but in favor of a lighter notation we 
have replaced $K_{\hat{\mu}}$ by $X_\mu$ . The metric here is $g=\diag(e_1,e_2,e_3,e_4)$ and the 
 spinor space is $V=\C^4$. 

\par
The numbers $u(\mu)$ and $u(\hat \nu)$ of spatial subindices 
of each (multi-)index, $\mu$ and $\hat \nu$,
can be written in terms of the signs $e_{  \mu}$ 
and $ e_{\hat \nu}$ 
that define the (anti-)hermiticity conditions
---namely $(\gamma^{ \mu})^{*}= e_{  \mu}\gamma^{  \mu}$ and 
$(\gamma^{\hat \nu})^{*}= e_{\hat \nu}\gamma^{\hat \nu}$. First, 
trivially, $e_{\mu}=(-1)^{u(\mu)}$. On the other hand 
since $u(\nu)+u(\hat \nu)$ is the 
total number $q$ of spatial indices, one 
has, by Appendix \ref{sec:App}, $e_{\hat \nu}=(-1)^{u(\hat \nu)+\lfloor 3/2\rfloor}= (-1)^{q+1+u(\nu)}=e_\nu(-1)^{q+1}$. 
Since the spinor space is four dimensional, 
by Proposition \ref{thm:Dsquare} one has
\begin{salign} 
 \frac{1}{2\cdot 4} \Tr\big[(D\hp{p,q})^2\big]=& 
\sum_{\mu =1}^4(-1)^{u(\mu)+\lfloor 1/2\rfloor} 
\big[N\cdot \TrN(K_\mu^2)+
  e_\mu(\TrN K_\mu)^2\big] 
     \\ +& \sum_{\nu =1}^4
     (-1)^{u(\hat \nu)+\lfloor 3/2\rfloor}
\big[N\cdot \TrN(X_\nu^2)+
  e_{\hat \nu} (\TrN X_\nu)^2\big]\, \label{eq:generalD2d4}\numerada \\
   =&\sum_{\mu=1}^4 
  e_\mu N \cdot\TrN \big[ K_\mu^2+(-1)^{q+1} X_\mu^2 \big]
  +(\TrN K_\mu)^2+(\TrN  X_\mu )^2 \,.
\end{salign}
In Section \ref{sec:RiemLor} 
we specialize eq. \eqref{eq:generalD2d4} to fuzzy Riemannian 
and Lorentzian geometries. Before, it will be   
useful to obtain the quartic term in order to
integrate it with the quadratic one.

\subsection{The term $\Tr D^4 $} \label{sec:d4D4}
  
To access $\Tr(D^4)$ we now detect 
the non-vanishing chord diagrams. 
  
\subsubsection{Non-vanishing chord diagrams}\label{sec:Chordsd4}

In four dimensions chord diagrams of various number of points ($2n=4,6,8,10,12$)
have to be computed  to access $\Tr(D^4)$. Next proposition helps 
to see the only non-trivial diagrams and requires some new notation. With
each multi-index $I_i$ running over eight values $I_i=\mu,\hat{\nu}$ ($\mu,\nu=\Delta_4$),
the $ 8^4$ decorations for the tensor 
$\chi^{I_1I_2I_3I_4}$ fall into the following $\tau$-types:
\begin{align} \label{eq:pelos} 
 \raisebox{-.45\height}{\includegraphics[width=2cm]{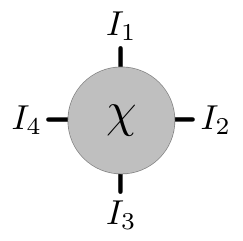}}
 =\chi^{I_1I_2I_3I_4}\in \vast\{ 
& \raisebox{-.45\height}{\includegraphics[width=2cm]{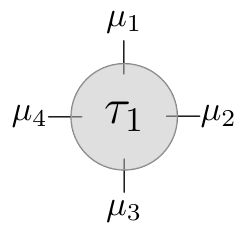}}
,
  \raisebox{-.45\height}{\includegraphics[width=2cm]{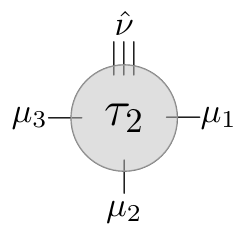}}
,
\raisebox{-.45\height}{\includegraphics[width=2cm]{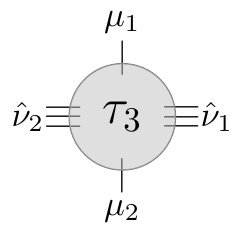}}   \raisetag{-3.2\normalbaselineskip}
\\
 &
\raisebox{-.45\height}{\includegraphics[width=2cm]{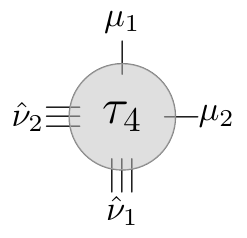}}
,\raisebox{-.45\height}{\includegraphics[width=2cm]{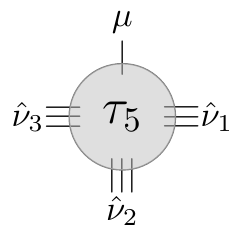}}
,\raisebox{-.45\height}{\includegraphics[width=2cm]{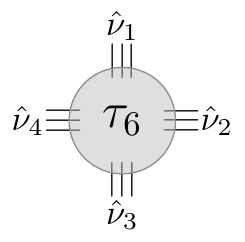}} \nonumber
\vast \} 
\end{align}
The leftmost diagram $\chi$ is of generic type. 
On the other hand, not only do the diagrams in the list indicate the number of 
points (the total number of bars
transversal to the circle), they also state how these are grouped: 
normal indices $\mu_i=1,\ldots, 4$ being a single line 
and multiple $\hat{\nu}_i$ a triple line. 
Although they are in fact ordinary chord diagrams,
they cannot have contractions between 
the grouped lines  due to the strict increasing ordering of their indices.
\par 
If a diagram $\chi$ accepts  
a decoration of the type $\tau_i$ in the LHS of \eqref{eq:pelos}, up to rotation, we symbolically write $\chi\in \tau_i$.
In the $\tau$-types of the RHS, however, 
$I_1$ corresponds strictly to the upper 
index of the respective diagram in the list, $I_2$ to the
rightmost, and so on clockwise. One can sum over the $\tau_i$ classes ---
since we are interested in products of $\chi^{I_1\ldots I_4}$
with traces (which are cyclic) and products of two traces (which
are summed over all the subsets of $\{1,2,3,4\}$, 
see eq. \eqref{eq:afrak} for details)--- and in order to do so, one 
has to include symmetry factors, namely $\{  1,4,2,4,4,1\}$ in that order. 
All the chord diagrams contributing to  $\Tr(D^4)$
in $d=4$ are then covered by 
\begin{align} \label{eq:CDsplit}
\big(\sum_{\chi \in \tau_4}  + 
4 \sum_{\chi\in \tau_2 } +
2 \sum_{\chi \in \tau_3} + 
4 \sum_{\chi \in \tau_4} +
4 \sum_{\chi \in \tau_5} +
 \sum_{\chi \in \tau_6} \big)\,.
\end{align}
A cross check is that the symmetry factors add up to $16$ and, since each
(multi)index in the list \eqref{eq:pelos} can take four values, 
the number of all diagram index decorations is $4^4 \times 16= 8^4$. 
Which of them survives is shown next:
\begin{prop}\label{thm:deltas}
Let $g=\diag (e_1,e_2,e_3,e_4)$ denote the quadratic
form given by the signature $(p,q)$.
  For any $\mu,\nu,\mu_1,\mu_2,\mu_3,\mu_4, \nu_1,\nu_2,\nu_3,\nu_4=1,\ldots, 4$
 the following holds for each one of 
 the diagrams $\chi$ of the type $\tau_i$ ---defined by eq. \eqref{eq:pelos}--- indicated to the right of each equation:
\begin{align} \tag{$\tau_1$}\label{eq:tau1}
\chi^{\mu_1\mu_2\mu_3\mu_4} & = e_{\mu_1}e_{\mu_3}  \delta_{\mu_1}^{\mu_2}  
\delta_{\mu_3}^{\mu_4}
-e_{\mu_1}e_{\mu_2}  \delta_{\mu_1}^{\mu_3}
\delta_{\mu_2}^{\mu_4}+
e_{\mu_1}e_{\mu_3}  \delta_{\mu_1}^{\mu_4}
\delta_{\mu_2}^{\mu_3}\,,\\[4pt]
\chi^{\hat\nu\mu_1\mu_2\mu_3}&= (-1)^{|\sigma|+1} e_{\mu_1}  
 e_{\mu_2}
  e_{\mu_3}\delta_{\nu\mu_1\mu_2\mu_3}\,,  \tag{$\tau_2$}\label{eq:tau2}
\end{align}
where $\sigma=\sigma(\boldsymbol \mu,\nu):=  \binom{\alpha_1 \alpha_2 \alpha_3}{\mu_1 \mu_2 \mu_3} \in \mtr{Sym }\{\alpha_1,\alpha_2,\alpha_3\}$,
with $\hat \nu = (\alpha_1,\alpha_2,\alpha_3)$ ordered as
$\alpha_1 < \alpha_2 < \alpha_3$. Also 
\begin{align} \label{eq:aboveof}
\delta_{\alpha\mu\nu\rho}=
\begin{cases}
1 &\mbox{when $\{\alpha,\mu,\nu,\rho\}=\Delta_4$}
 \\ 0&  \mbox{otherwise}
\end{cases}
\end{align}
(i.e. 
$\delta_{\alpha\mu\nu\rho}$ 
is the Levi-Civita symbol in 
absolute value). Whenever not all the four indices $\mu_1,\mu_2,\nu_1,\nu_2$ agree,
\begin{align}
\chi^{\mu_1\hat\nu_1\mu_2\hat\nu_2}&= -(-1)^{\mu_1+\mu_2}
e_1e_2e_3e_4 
(\delta_{\mu_1}^{\nu_1}\delta_{\mu_2}^{\nu_2}
-
\delta_{\mu_1}^{\nu_2}\delta_{\mu_2}^{\nu_1})
\mp
e_{\mu_1} \Big(\displaystyle\prod_{\alpha\neq \nu_1} e_\alpha \Big)
\delta_{\mu_1}^{\mu_2}\delta_{\nu_1}^{\nu_2}\,, \label{eq:tau3}
\tag{$\tau_3$}
 \\ \tag{$\tau_4$} \label{eq:tau4}
\chi^{\mu_1\mu_2\hat\nu_1\hat\nu_2}&= +
(-1)^{\mu_1+\mu_2}e_1e_2e_3e_4 
(\delta_{\mu_1}^{\nu_1}\delta_{\mu_2}^{\nu_2}
-
\delta_{\mu_1}^{\nu_2}\delta_{\mu_2}^{\nu_1}
)
\mp
e_{\mu_1} \Big(\displaystyle\prod_{\alpha\neq \nu_1} e_\alpha \Big)
\delta_{\mu_1}^{\mu_2}\delta_{\nu_1}^{\nu_2} \,.
\end{align}
(see below for the sign choice). 
Otherwise these two diagrams 
satisfy $\chi^{\mu\hat\mu\mu\hat\mu}=e_1e_2e_3e_4 $
and $\chi^{\mu\mu\hat\mu\hat\mu}=-e_1e_2e_3e_4 $. Moreover, 
letting $\sigma=\sigma(\boldsymbol \nu,\mu)=
 \binom{\theta_1 \theta_2 \theta_3}{\nu_1 \nu_2 \nu_3} \in \mtr{Sym }\{\theta_1,\theta_2,\theta_3\}$,
with $\hat \nu = (\theta_1,\theta_2,\theta_3)$ ordered as
$\theta_1 < \theta_2 < \theta_3$, one has
\begin{align} 
\chi^{\mu\hat\nu_1\hat\nu_2\hat\nu_3}&= (-1)^{|\sigma|+1} e_{\nu_1}  
 e_{\nu_2}
  e_{\nu_3}\delta_{\mu\nu_1\nu_2\nu_3}\,,  \tag{$\tau_5$} \label{eq:tau5}
\end{align}
and, finally, if $\chi \in \tau_6$
\begin{align}
\tag{$\tau_6$}
\chi^{\hat{\nu}_1\hat{\nu}_2\hat{\nu}_3\hat{\nu}_4} & = \pm \big[e_{{\nu}_1}e_{{\nu}_3}  \delta_{{\nu}_1}^{{\nu}_2}  
\delta_{{\nu}_3}^{{\nu}_4}
-e_{{\nu}_1}e_{{\nu}_2}  \delta_{{\nu}_1}^{{\nu}_3}
\delta_{{\nu}_2}^{{\nu}_4}+
e_{{\nu}_1}e_{{\nu}_3}  \delta_{{\nu}_1}^{{\nu}_4}
\delta_{{\nu}_2}^{{\nu}_3}\big]\,. \label{eq:tau6}
\end{align} 
The upper signs in equations \eqref{eq:tau3}, \eqref{eq:tau4} 
and \eqref{eq:tau6} are taken if $\chi$ has minimal crossings. 
\end{prop}

The minimality condition on the crossings, 
assumed for the $\tau_{3,4,6}$ classes, is meant to
shorten the proof. Exactly for those classes, the spacetime 
indices do not necessarily determine a unique diagram
by assuming that it does not vanish. 
This requirement can be left out, and in that case eq. 
\eqref{eq:tau6} should have a global sign $\pm$ 
that depends on the diagram; in the $\tau_{3,4}$
cases \eqref{eq:tau3} and \eqref{eq:tau4} the term $
e_{\mu_1} \big(\prod_{\alpha\neq \nu_1} e_\alpha \big)
\delta_{\mu_1}^{\mu_2}\delta_{\nu_1}^{\nu_2}$
would undergo a diagram-dependent sign change.
However, as we will see, these will be `effectively'
replaced by the minimal-crossing diagram, so the simplified claim suffices.

\begin{proof}
The proof is by direct, even if in cases lengthy,
computation. One first 
seeks the conditions one has to impose 
on the indices for a diagram not to vanish, typically 
in terms of 
Kronecker deltas, and then one computes their coefficients
in terms of the quadratic form $g=\diag (e_1,e_2,e_3,e_4)$. 
We order the proof
by similarity of the statements: 
\begin{itemize}
\itemB \textit{Type} $\tau_1$. The $\tau_1$-type diagram is well-known, 
for it is the only one here without any single multi-index (see the end of Sec. \ref{sec:gammas}).

\itemB \textit{Type}  $\tau_6$.
Notice that at least
two pairings of the $\nu_l$'s are needed 
for the diagram not to vanish: $\nu_i =\nu_j=:\nu$
and $\nu_t=\nu_r=:\mu$ with $\{i,j,t,r\} =\Delta_4$. Therefore, 
the Kronecker deltas are placed precisely as   
for the $\tau_1$ type. The computation 
of their $e$-factors is a matter of counting:
for each chord joining two points labeled 
with, say, $\alpha$ there is an $e_\alpha$ factor. 
There are 6 such chords, labeled 
by $ \{e_\alpha \}_{\alpha\neq \nu}  \cupdot  \{e_\rho \}_{\rho\neq \mu}    $,
for in $\hat \nu$ all the indices  
$\alpha\neq \nu$ appear and similarly for $\hat \mu$. 
Thus, if $\nu \neq \alpha \neq \mu $, 
$e_\alpha$ appears twice, so $e_\alpha^2=1$.
The two remaining chord-labels are those appearing 
either in 
$ \{e_\alpha \}_{\alpha\neq \nu} $ or in 
$ \{e_\rho \}_{\rho\neq \mu}    $. Thus the factor is $e_\mu e_\nu$ 
and we only have to compute the sign:
the $\hat \mu\hat\nu\hat\mu\hat\nu $ configuration with 
minimal crossings has sign $(-1)^5$. For 
$ \hat \mu\hat\mu\hat\nu\hat\nu $
and $ \hat \mu\hat\nu \hat\nu\hat\mu$
the crossings yield a positive sign $(-1)^6$.

\itemB \textit{Type} $\tau_2$. 
Since $\hat \nu \in \Lambda_{d=4}^-$, 
all the three indices $\alpha_i$ in 
$\hat \nu = (\alpha_1,\alpha_2,\alpha_3)$ 
different. For this diagram not to vanish,
the set equality
$\{\alpha_1,\alpha_2,\alpha_3\}=\{\mu_1,\mu_2,\mu_3\}$ 
should hold,
i.e. a permutation $\sigma \in \{\mu_1,\mu_2,\mu_3\}$
with $\alpha_{\sigma(i)}=\mu_{i}$ is needed. This says first, 
that $\nu$ cannot be any of $\mu_i$ (whence 
the $\delta_{\nu\mu_1\mu_2\mu_3}$) and 
second, that each of the three chords yields a factor 
$e_{\mu_i} $ with a sign $ (-1)^{|\sigma|+1}$. The
 extra minus is due to the convention to place the indices,
 e.g.  for $\hat 4   123$, the numbers $123123$ are put cyclicly; 
this permutation $\sigma$ is the identity, 
 which nevertheless looks like the `V diagram'
 in eq. \eqref{eq:Roman}.

\itemB \textit{Type} $\tau_5$. Suppose 
that two indices of a non-vanishing diagram agree. Then 
either $\mu=\nu_i$ or (wlog)
$\nu_1=\nu_2$. 
In the first case notice that 
in $\mu$ and $\hat \nu_i$ the indices $1,2,3,4$ all
appear  listed. This implies 
for the remaining two 
multi-indices have to be 
of the form 
$\hat \nu_l=(***)$  
and $\hat \nu_ m=(1*4)$ or
$\hat \nu_l=(1**)$  
and $\hat \nu_ m=(**4)$  
where $\{i,l,m\}=\{1,2,3\}$ and $*\in \Delta_4$.
\begin{itemize}
 \itemW In the first case, 
 $\hat \nu_l=(***)$  
and $\hat \nu_ m=(1*4)$ 
the numbers $\rho,\rho,2,3$ 
(for some $\rho \in \Delta_4$)
should fill the placeholders $*$. 
Then $\rho$ has to appear in both 
$\hat \nu_l$  
and $\hat \nu_ m$, but no
value of $\rho$ fulfills this
if the increasing ordering
is to be preserved, hence we are only left with next case

\itemW If $\hat \nu_l=(1**)$  
and $\hat \nu_ m=(**4)$, say  
$\hat \nu_l=(1wx)$  
and $\hat \nu_ m=(yz4)$,
then $\{w,x\}$ and 
$\{y,z\}$ are the sets   
$\{2,\rho\}$ and $\{3,\rho\}$ (not necessarily
in this order),
for some $\rho$. Clearly, $\rho$ cannot 
be either $1$ or $4$ since each appears once 
in one multi-index. But $\rho=2,3$ would
force also a repetition of indices in at least
one multi-index, 
which contradicts $\hat \nu_ m, \hat \nu_ l\in \Lambda_4^-$. 
\end{itemize}
This contradiction implies that if
$\mu$ equals some $\nu_i$ then the diagram vanishes.
By a similar analysis one sees that a repetition 
$\nu_i=\nu_j$ implies also that the diagram is zero.
Hence the diagram
is a multiple of $\delta_{\mu\nu_1\nu_2\nu_3}$.
Thereafter it is easy to compute the $e$-coefficients
following a similar argument to the given for 
the type $\tau_2$ diagram to 
arrive at the sign $(-1)^{|\lambda|+1}$
for $\lambda$ a permutation of $\{\nu_1,\nu_2,\nu_3\}$.

\itemB \textit{Type} $\tau_3$.
If no indices coincide then 
one gets two different 
numbers $i,j\in \Delta_4$
appearing exactly once in the list $\mu_1,\hat \nu_1,\mu_2,\hat\nu_2$. 
Since these cannot be matched by a chord,
a non-zero diagram 
requires repetitions. 
\begin{itemize}
 \itemW If $\mu_1=\mu_2$
 then $\nu_1=\nu_2$.
 Since by hypothesis 
 the four cannot agree
 the minimal crossings for this 
 configuration is seen to be one,
 so the sign is $(-1)$. The $e$-factors
 are: $e_{\mu_1}$ 
 for the chord between
 $\mu_1$ and $\mu_2$,
 the product of three 
 $\prod_{\alpha\neq \nu_1 }e_\alpha$,  
 for the three chords between $\hat \nu_1 $ and 
 $\hat \nu_2$.
 This accounts for 
 $
- 
e_{\mu_1} \Big(\displaystyle\prod_{\alpha\neq \nu} e_\alpha \Big)\delta_{\mu_1}^{\mu_2}\delta_{\nu_1}^{\nu_2}$. 

\itemW If $\mu_1=\nu_1$, then 
again $\mu_2=\nu_2$ in order for 
the indices listed in
$\mu_1,\hat \nu_1,\mu_2,\hat\nu_2$
to appear precisely twice. 
Since $\mu_1=\nu_1$ 
implies that $\mu_1$ 
does not appear in $\hat \nu_1$,
there is one chord 
(thus a factor $e_\alpha$)
for each $\alpha\in\Delta_4$.
After straightforward (albeit neither 
brief nor very illuminating)
computation one finds the sign 
$ (-1)^{\mu_1+\mu_2+1} $.
All in all, one gets 
$  (-1)^{\mu_1+\mu_2+1}
e_1e_2e_3e_4 
\delta_{\mu_1}^{\nu_1}\delta_{\mu_2}^{\nu_2}$.

\itemW If $\mu_1=\nu_2$, then 
again $\mu_2=\nu_1$. But this 
is the same as the last point
with $\nu_1 \leftrightarrow \nu_2$.
This accounts for 
$(-1)^{\mu_1+\mu_2}
e_1e_2e_3e_4 \delta_{\mu_1}^{\nu_2}\delta_{\mu_2}^{\nu_1}$.
\end{itemize}

\itemB \textit{Type} $\tau_4$. Mutatis mutandis from 
the type $\tau_3$.  \qedhere

\end{itemize}

\end{proof}

\begin{remark}\label{thm:Nested}
We just used the `minimal' number of crossings 
for diagrams with a more than two-fold index repetition. For 
instance, for the four point diagram 
evaluated in $\chi^{1111}=1$
there might be one crossing or no crossings, but crucially
two diagrams have no crossing so  $\sum_\chi \chi^{1111} \times \mbox{traces} =(1-1+1) \times \mbox{traces} $. This reappears in the computation of 
twelve-point diagrams in a nested fashion, as shown in
Figure \ref{fig:taus}. If we pick 
$\hat 2\hat 4\hat 2\hat 4$ as configuration
of the indices, then imposing 
$\chi^{\hat 2\hat 4\hat 2\hat 4} \neq 0 $
does not determine $\chi\in \CD{6}$:
the lines joining the two $2$-indices 
and the two $4$-indices diagonally are 
mandatory, but for the four $1$-indices and four $3$-indices 
one can choose at the blobs tagged  
with $\phi,\psi$ one of three 
possibilities (shown in the diagrams
of eq. \eqref{eq:chitensors} as $\theta,\xi,\zeta$).
Then there are 9 possible sign values. 
Again, it is essential  
that there are 5 positive and 4 negative
global signs in 
$\{\chi^{\hat 2\hat 4\hat 2\hat 4} (\phi,\psi)\}_{\psi,\phi \in \{\theta,\xi,\zeta\} }$ 
and the sum $\sum_\chi \chi^{\hat 2\hat 4\hat 2\hat 4} (\mbox{traces})$ can be replaced by the diagram with minimal crossings 
(of global positive sign).
\end{remark} 
\begin{figure}
\includegraphics[width=7.2cm]{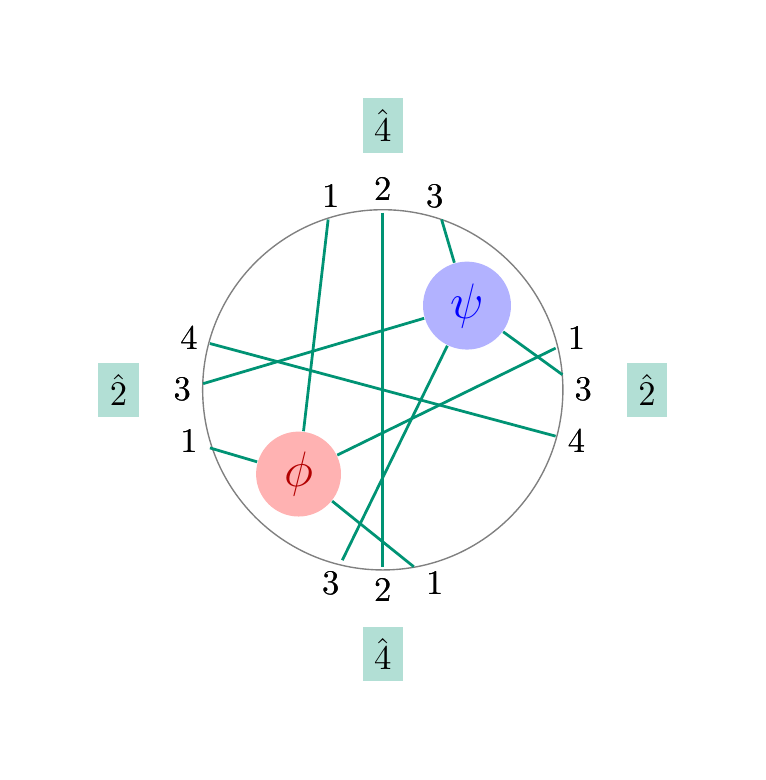}
\vspace{-4ex}
\caption{The diagram $\chi^{\hat 2\hat 4\hat 2\hat 4} (\phi,\psi)$
shows the nested structure of two $\tau_1$-type
diagrams, $\phi$ and $\psi$, in a $\tau_6$-type referred to in Remark \ref{thm:Nested}
\label{fig:taus}}
\end{figure}
As a last piece of preparation, we need to 
determine the signs $e_{I_1}e_{I_2} e_{I_3} e_{I_4}$ 
for each $\tau_i$-type. These 
turn out to be constant and fully
determined by the $\tau_i$-type:
\begin{claim}\label{thm:signos}
 Assuming that for $I_1,I_2,I_3,I_4\in \Lambda_{d=4}^-$
 the tensor 
$ \chi^{I_1I_2I_3I_4} $ does not vanish, then  
$e_{I_1}e_{I_2} e_{I_3} e_{I_4}$ reads
in each case%
\begin{align}
e_{\mu_1}e_{\mu_2} e_{\mu_3} e_{\mu_4} & \equiv + 1  \label{eq:uno} \\
e_{\hat \nu }e_{\mu_1} e_{\mu_2} e_{\mu_3} & \equiv -1  \label{eq:dos}\\
e_{\mu_1}e_{\mu_2} e_{\hat\nu_1} e_{\hat\nu_2} & \equiv +1 \label{eq:tres} \\ 
e_{\mu}e_{\hat\nu_1} e_{\hat\nu_2} e_{\hat\nu_3} & \equiv -1\label{eq:cuatro}\\
e_{\hat \nu_1}e_{\hat \nu_2} e_{\hat \nu_3} e_{\hat \nu_4} & \equiv +1\label{eq:cinco}
\end{align}
\end{claim} 

\begin{proof}
To obtain these relations one needs 
Proposition \ref{thm:deltas}.
The first and last cases are obvious,
since $\chi^{I_1I_2I_3I_4}\neq 0$
requires in each case a repetition 
$e_{\mu_i}^2e_{\mu_j}^2=1$ 
or $e_{\hat \nu_i}^2e_{\hat \nu_j}^2=1$. \par
For the second, $
e_{\hat \nu }e_{\mu_1} e_{\mu_2} e_{\mu_3} =
(-1)^{u(\hat \nu)+\lfloor 3/2\rfloor+ \sum_i u(\mu_i)}$,
by Appendix \ref{sec:App}.
The non-vanishing of $\chi^{\hat \nu  \mu_1 \mu_2 \mu_3}$
implies that $\hat\nu$ is
the multi-index containing $\mu_1,\mu_2,\mu_3$,
so $u(\mu_1)+ u(\mu_2)+u(\mu_3)=u(\hat\nu)$ 
and  eq. \eqref{eq:dos} follows.\par 
For the third identity, if $\chi^{\mu_1\mu_2\hat\nu_1\hat\nu_2}$ (and thus 
$\chi^{\mu_1\hat\nu_1\mu_2\hat\nu_2}$) 
does not vanish, then it
 is either of the form $\chi^{\mu\mu\hat\mu\hat\mu}$, 
 $\chi^{\mu\mu\hat\nu\hat\nu}$ or
   $\chi^{\mu\nu\hat\mu\hat\nu}$ ($\mu\neq \nu$). Only
   for the latter one needs a non-trivial check: 
   $e_{\mu}e_{\nu} e_{\hat\mu} e_{\hat\nu}=e_\mu\cdot (-1)^{1+u(\Delta_4-\{\mu\})} e_\nu \cdot {(-1)}^{1+u(\Delta_4-\{\nu\})}
   =e_\mu e_\nu (-1)^{2u(\Delta_4-\{u,v\})} (-1)^{u(\mu)+u(\nu)}$, 
   and since $e_\mu=(-1)^{u(\mu)}$,
     $e_{\mu}e_{\nu} e_{\hat\mu} e_{\hat\nu}=1$. In either case,
   eq. \eqref{eq:tres} follows.
   \par 
 We are left with the fourth identity. By assumption 
  all the indices $\nu_j\neq \nu_i\neq \mu$ 
  if $i\neq j$. Then by eq. \eqref{eq:autoadjuncion}
 \begin{align*}
 e_{\mu}e_{\hat\nu_1} e_{\hat\nu_2} e_{\hat\nu_3}&=e_\mu
  \cdot  (-1)^{3\times \lfloor 3/2\rfloor  + u(\Delta_4 \setminus \{\nu_1\})+
   u(\Delta_4 \setminus \{\nu_2\})
   + u(\Delta_4 \setminus \{\nu_3\})} \\
   & = -e_\mu(-1)^{3u(\mu) }  (-1)^{2u(\nu_1)} (-1)^{2u(\nu_2)}
    (-1)^{2u(\nu_3)}=-1\end{align*}
From the first to the second line 
we used $\Delta_4-\{\nu_1\}=\{\mu,\nu_2,\nu_3\}$,
and similar relations. 
\end{proof}

\subsubsection{Main claim} 
With help of these two results, we state 
the main one. We recall that the definition of the permutation ${\sigma(\boldsymbol \nu,\mu)}$, appearing next, is given in eq. \eqref{eq:aboveof}.
\begin{prop}\label{thm:Haupt}
For a 4-dimensional fuzzy geometry of signature $(p,q)$, 
the purely quartic spectral action $\frac14\Tr(D^4)=N \mtc S_4 + \mtc B_4$ 
is given by 
\begin{align}  \label{eq:NCPolynomial4}
\mtc S_4  &=  \TrN \Big\{ 
2 \sum_\mu  \Km^4 +  \Xm^4  
\\
&
+
4\sum_{\mu<\nu } e_\mu e_\nu (2 \Km^2\Kn^2 +2 \Xm^2\Xn^2  
-\Km\Kn\Km\Kn -\Xm\Xn\Xm\Xn ) \nonumber  
\\ &  \nonumber  
-\sum_{\alpha,\beta,\mu,\nu} 
\delta_{\alpha\beta\mu\nu}
e_\alpha e_\beta 
 \big[ (K_\mu X_\nu)^2+ 2K_\mu^2 X_\nu^2 \big]+ 2(-1)^q \sum_\mu\big[(K_\mu X_\mu)^2- 2K_\mu^2 X_\mu^2  \big]  \\
&
+8(-1)^{q+1} \sum_{\mu,\boldsymbol \nu} 
(-1)^{|\sigma(\boldsymbol \nu,\mu)|} \delta_{\mu\nu_1\nu_2\nu_3} 
e_\mu( X_\mu K_{\nu_1} K_{\nu_2} K_{\nu_3}
+  K_\mu X_{\nu_1} X_{\nu_2} X_{\nu_3} )
\Big\}\, , \nonumber
\end{align}
 and 
\begin{align} \label{eq:B4Haupt}
 \mathcal{B}_4 &= 
 8
\sum_{\mu,\nu}
 (-1)^{q+1}e_\nu   
\TrN X_\mu 
\cdot \TrN(X_\mu X_\nu^2) 
 +
e_\nu \TrN(\Km) \cdot  \TrN(\Km \Kn^2)
 \\
 &
 +\sum_{\mu,\nu=1}^4  \nonumber
\Big\{ 2\TrN(\Xm^2)\cdot\TrN(\Xn^2) +4 e_\mu e_\nu \big[\TrN(\Xm\Xn)\big]^2  \Big\}
\\  \nonumber
&+\sum_{\mu,\nu=1}^4 
\Big\{ 2\TrN(\Km^2)\cdot\TrN(\Kn^2) +4 e_\mu e_\nu \big[\TrN(\Km\Kn)\big]^2  \Big\}
\\  \nonumber
  &+ 4\sum_{\mu=1}^4
\Big\{ 
2(-1)^{1+q}e_\mu \TrN(\Km)\cdot \TrN(\Km\Xm^2)+2e_{ \mu} \TrN(\Xm) \cdot \TrN(\Xm \Km^2) \\
&\hspace{38pt} +(-1)^{1+q}\TrN(\Xm^2)\cdot \TrN(\Km^2) + 2\big[\TrN(\Km \Xm)\big]^2
\Big\}  \nonumber
\\  \nonumber
 &  \vphantom{\sum_{\mu,\nu=1} } -8 \sum_{\nu,\boldsymbol \mu=1}^4
 (-1)^{|\sigma(\boldsymbol \mu,\nu)|} \delta_{\nu\mu_1\mu_2\mu_3}  \cdot \big \{ 
  -\TrN(X_\nu ) \cdot \TrN(K_{\mu_1}K_{\mu_2}K_{\mu_3}) 
\\
& \hspace{38pt}  \vphantom{\sum_{\mu,\nu=1} }+e_{\mu_2}e_{\mu_3} \big(
\TrN K_{\mu_1}   \cdot \TrN(X_\nu K_{\mu_2}K_{\mu_3}) \nonumber
\big)\\
& \hspace{38pt}  
+\TrN X_{\mu_1}   \cdot \TrN(K_\nu X_{\mu_2}X_{\mu_3}) 
+(-1)^q \TrN K_\nu  \cdot \TrN(X_{\mu_1}X_{\mu_2}X_{\mu_3}) 
\big\} \nonumber \\  \nonumber
&+24 \sum_{\mu\neq \nu=1}^4  
 (-1)^{1+q} e_\nu 
 \TrN(K_\mu)\cdot \TrN\big(K_\mu X^2_{\nu}\big) +
 e_\mu \TrN(X_\nu)\cdot \TrN\big(K_\mu^2 X_{\nu}\big)  
 \\  \nonumber
  & +12 \sum_{\mu\neq \nu}
 \Big\{ 
 2 \big[\TrN(K_\mu X_\nu)\big]^2 + e_\mu e_\nu (-1)^{q+1} \TrN(K_\mu^2 ) \cdot \TrN(X_\nu^2)  
\Big\} \,  , \nonumber
\end{align}
 The eight matrices $K_\mu, X_\mu$ 
 satisfy the following (anti-)hermiticity conditions:
 \begin{align}
 K_{\mu}^*= e_\mu K_\mu  \,\,\and \,\,\, X_\mu^*= e_\mu (-1)^{q+1} X_\mu  \quad \mbox{for any }\mu\in\Delta_4\,,
 \end{align}
where each $e_\mu\in \{+1,-1\}$ is determined by 
$g=\diag (e_1,e_2,e_3,e_4)$.
\end{prop}
\begin{proof}
We first find $\mathcal S_{4}=\sum_{n=2}^{6}
\mathfrak{s}_n(\chi) $ using eqs. \eqref{eq:pelos} and 
\eqref{eq:CDsplit}. 
By direct computation 
\begin{align}  \label{eq:s2}
\sum_{\chi\in\CD{2}} \mathfrak{s}_{2} (\chi)&=
2 \sum_\mu \TrN(\Km^4)+
8\sum_{\mu<\nu } e_\mu e_\nu \TrN(\Km^2\Kn^2) \\
&-4 \sum_{\mu<\nu } e_\mu e_\nu \TrN(\Km\Kn\Km\Kn)\,. \nonumber
\end{align}
In view of \eqref{eq:uno} and \eqref{eq:cinco} and 
the similarity of the
$\tau_1$ and $\tau_6$ type 
diagrams, one gets the same 
result by replacing $K_\mu$ 
by $K_{\hat\mu}=X_{\mu}$, 
namely 
\begin{align}  \label{eq:s6}
\sum_{\chi\in\CD{6}} \mathfrak{s}_{6} (\chi)&=
2  \sum_\mu \TrN(\Xm^4)+
8 \sum_{\mu<\nu } e_\mu e_\nu \TrN(\Xm^2\Xn^2) \\
&-4 \sum_{\mu<\nu } e_\mu e_\nu \TrN(\Xm\Xn\Xm\Xn)\,. \nonumber
\end{align}
Next, using eq. \eqref{eq:dos}, the 6-point diagrams are evaluated:
\begin{align}
\sum_{\chi\in\CD{3}} \mathfrak{s}_{3} (\chi)&= 4 
 \sum_{\nu,\boldsymbol \mu } 
   \TrN \big[ \delta_{\nu\mu_1\mu_2\mu_3} e_{\mu_1} e_{\mu_2} e_{\mu_3} (-1)^{1+|\sigma_{\boldsymbol \mu}|}
  \\ \nonumber &\qquad \cdot (X_\nu K_{\mu_1} K_{\mu_2}K_{\mu_3} - K_{\mu_3} K_{\mu_2}K_{\mu_1}X_\nu  )
\big] \\
&=-8 \sum_{\nu,\boldsymbol \mu } \nonumber
\TrN \big[ \delta_{\nu\mu_1\mu_2\mu_3} e_{\mu_1} e_{\mu_2} e_{\mu_3} (-1)^{|\sigma_{\boldsymbol \mu}|}
X_\nu K_{\mu_1} K_{\mu_2}K_{\mu_3}\big]\\ \nonumber
& = 8(-1)^{1+q} \sum_{\nu,\boldsymbol \mu } \nonumber
\TrN \big[ \delta_{\nu\mu_1\mu_2\mu_3} e_\nu (-1)^{|\sigma_{\boldsymbol \mu}|}
X_\nu K_{\mu_1} K_{\mu_2}K_{\mu_3}\big] \,. \nonumber
\end{align}
Here, again using the duality 
between $\tau_2$ and $\tau_5$ 
evident in Proposition \ref{thm:deltas}
and Claim \ref{thm:signos}, 
the $\mathfrak{s}_{5} $ term can be 
computed by swapping each $K_{\mu}$ 
matrix with the $K_{\hat \mu}$ matrix: 
\begin{align}
 \sum_{\chi\in\CD{5}} \mathfrak{s}_{5} (\chi)=
 \sum_{\chi\in\CD{3}} \mathfrak{s}_{3} (\chi) 
 \bigg |_{K_\nu \leftrightarrow X_\nu \text{ for all $\nu=1,2,3,4$}}
\end{align}
(but there is no the sign swap $e_{\mu} \leftrightarrow e_{\hat \mu}$).\par Finally, we split the sum
$
\sum_{\chi\in\CD{4}} = 2\sum_{\chi \in \tau_3} +4\sum_{\chi \in \tau_4}$
in order to compute the term 
$\mtf s_4$. The calculation simplifies using eq. \eqref{eq:tres} and noticing that 
(for $\mu_1=\mu_2=\nu_1=\nu_2$ being false), one has 
\begin{salign}\numerada &\sum _{\boldsymbol \nu, \boldsymbol \mu} (-1)^{\mu_1+\mu_2}
(\delta_{\mu_1}^{\nu_1}\delta_{\mu_2}^{\nu_2}
-
\delta_{\mu_1}^{\nu_2}\delta_{\mu_2}^{\nu_1})  \big[
\TrN(K_{\mu_1}X_{\nu_1}K_{\mu_2} X_{\nu_2})  \\&\hspace{2.92cm}+ e_{\mu_1}e_{\mu_2} e_{\hat\nu_1} e_{\hat\nu_2}
\TrN(X_{\nu_2}K_{\mu_2}X_{\nu_1}K_{\mu_1} )\big] \\
&=\sum_{\mu \neq \nu} (-1)^{\mu+\nu}  \TrN\big\{
K_{\mu }X_{\mu}K_{\nu} X_{\nu}+ X_{\nu}K_{\nu}X_{\mu}K_{\mu }\\
&\hspace{2.92cm}-K_{\mu }X_{\nu}K_{\nu} X_{\mu}- X_{\mu}K_{\nu}X_{\nu}K_{\mu }
\big\}=0\,,
\end{salign}
using the cyclicity of the trace. Therefore 
both equations \eqref{eq:tau3} and 
\eqref{eq:tau4} the only contribution to 
$\mathfrak s_4$ comes from the 
term $e_{\mu_1} \big(\prod_{\alpha\neq \nu} e_\alpha \big)
\delta_{\mu_1}^{\mu_2}\delta_{\nu_1}^{\nu_2} $ (which require 
$\mu_i\neq \nu_i$) 
and from the terms $\chi^{\mu\hat\mu\mu\hat\mu}$ and $\chi^{\mu\mu\hat\mu\hat\mu}$.
These terms appear, respectively, in the first and second lines of
\begin{align}
\sum_{\chi\in\CD{4}} \mathfrak{s}_{4} (\chi)&= 
-\sum_{\alpha,\beta,\mu,\nu} 
\delta_{\alpha\beta\mu\nu}
e_\alpha e_\beta 
\TrN\big[ (K_\mu X_\nu)^2+ 2K_\mu^2 X_\nu^2 \big] \\ 
&\,\qquad+ 2e_1e_2e_3e_4 \sum_\mu \TrN\big[(K_\mu X_\mu)^2- 2K_\mu^2 X_\mu^2  \big]  \nonumber
\end{align}
Expressing this via the delta $\delta_{\alpha\beta\mu\nu}$ 
is motivated by $e_{\mu} \big(\prod_{\rho\neq \nu} e_\rho\big)=\prod_{\substack{\rho\neq \mu \\ \rho\neq \nu}}e_\rho$.
\par 
We now compute in steps the bi-tracial functional
\vspace{-.13cm}
\begin{align} \raisetag{3.5\normalbaselineskip}\nonumber
\mathcal B_4=
\sum_{I \in (\Lambda_4^-)^{\times 4}}
\bigg\{&
\sum_{\chi\in \CD{2n(I)}} 
\chi^{I_1I_2I_3I_4}
\times \Big[
\sum_{i=1}^4    e_{I_i} 
\TrN (K_{I_1}\cdots \widehat{ K_{{I_i}}} \cdots K_{I_4} ) \cdot \TrN  K_{I_i} 
\\  +& \sum_{   {1\leq i<j\leq 4}^{\phantom .}   } \sum_{v,w\neq i,j}
e_{I_i}e_{I_j} 
\big(
\TrN (K_{I_v} K_{I_w}) \TrN (K_{I_i} K_{I_j})
\big)  \label{eq:B4}
\\ 
 +&\sum_{i=1}^4  \Big( \prod_{j \neq i} e_{I_j} \Big) 
 \TrN (K_{I_i}) \cdot 
\TrN (K_{I_4}\cdots \widehat{ K_{{I_i}}} \cdots K_{I_1} ) \Big]
\bigg\}\,. \nonumber
\end{align}
The contribution to $\mathcal B_4$
arising from the term in the square brackets
in the first, second and third lines are 
referred to as the (1,3), (2,2) and (3,1) 
partitions, respectively. 
For a fixed number $2r$ of points,  
these are denoted by $\sum_\chi \mathfrak b^{\pi}_{r}(\chi)$,
for $\pi \in \{(1,3), (2,2), (3,1)\}$. 
In view of the 
partial duality established in Proposition \ref{thm:deltas},
we obtain the contributions to $\mathcal B_4$  
by similarity; thus we first compute 12-pt and 4-pt 
diagrams together and later 6-pt and 10-pt diagrams.
This duality would be perfect if
both replacements $K_\nu \leftrightarrow K_{\hat \nu}(=X_\nu)
$ and $e_\nu \leftrightarrow e_{\hat{\nu}}$ would
swap the eqs. ($\tau_1  \leftrightarrow \tau_6$) and 
($\tau_2  \leftrightarrow \tau_5$)
in Proposition \ref{thm:deltas}. However, 
$e_\nu \leftrightarrow e_{\hat{\nu}}$ is not 
needed for the swapping to hold.

\par 
We begin with the 12-point diagrams for the (1,3) and (3,1) partitions.
As consequence of Claim \ref{thm:signos}, 
\begin{align}
\sum_{\chi \in \tau_6}\mtf b_6^{(1,3)}(\chi) \nonumber
+\mtf b_6^{(3,1)}(\chi)
&=\hspace{-3pt}\sum_{\nu_1,\ldots,\nu_4=1}^4 \hspace{-3pt}
\chi^{\hat \nu_1 \hat \nu_2\hat \nu_3\hat \nu_4}
\big[ e_{\hat \nu_1} \TrN X_{\nu_1}  \cdot \TrN(X_{\nu_2} \{X_{\nu_3},X_{\nu_4}\})
\\[-9pt] \nonumber 
&\hspace{8cm}+\text{cyclic}\,
\big]
\\[-7.5pt] \nonumber 
&= 4\sum_{\mu,\nu}
e_\mu e_\nu e_{\hat \mu}
\TrN X_\mu 
\cdot \TrN(X_\mu\{X_\nu,X_\nu\})
\\[1.5pt]
&=8
\sum_{\mu,\nu}
 (-1)^{q+1}e_\nu   
\TrN X_\mu 
\cdot \TrN(X_\mu X_\nu^2) \label{eq:61331}
  \end{align} 
after some simplification;
the last equality follows from eq. \eqref{eq:mhm}. 
The (2,2)-partition evaluates similarly 
to 
\begin{align}\nonumber
 \sum_{\chi \in \tau_6}\mtf b_6^{(2,2)}(\chi)&= 
 2\sum_{\mu,\nu} 
 \big\{ 
 e_\mu e_\nu e_{\hmu}^2 \Tr(\Xm^2)\cdot \Tr(\Xn^2)
 +
 2 e_\mu e_\nu e_{\hmu} e_{\hnu}  \Tr(\Xm\Xn)^2
 \big\}
 \\
 &=\sum_{\mu,\nu=1}^4  \label{eq:622}\raisetag{3.2\normalbaselineskip} 
\Big\{ 2\TrN(\Xm^2)\cdot\TrN(\Xn^2) +4 e_\mu e_\nu \big[\TrN(\Xm\Xn)\big]^2  \Big\}\,,
\end{align}
since $e_\mu e_\nu e_{\hmu} e_{\hnu}=(-1)^{2(1+q)}$ by eq. \eqref{eq:mhm}. 
One then computes $
\sum_{\chi \in \tau_6} \mathfrak b_6(\chi) $
by summing eqs. \eqref{eq:622} and \eqref{eq:61331}.\par

The four-point diagrams contain ordinary indices
and their computation is not illuminating. 
Since it moreover 
resembles that for the 12-point diagrams we omit it and 
present the result:
\begin{align}
\sum_{\chi\in \tau_1}\mathfrak{b}_2(\chi)  \label{eq:noconextau1}
&=\sum_{\mu,\nu=1}^4
8 e_\nu \TrN(\Km) \cdot  \TrN(\Km \Kn^2)
\\
&+\sum_{\mu,\nu=1}^4 
\Big\{ 2\TrN(\Km^2)\cdot\TrN(\Kn^2) +4 e_\mu e_\nu \big[\TrN(\Km\Kn)\big]^2  \Big\}\,.
\nonumber
\end{align}

We now present the computation of 
6-point and 10-point diagrams. 
Again, we remark that the terms corresponding to the (1,3)
and (3,1) partitions agree, 
$\sum_{\chi} \mtf b_3\hp{1,3}(\chi)=\sum_\chi\mtf b_3\hp{3,1}(\chi)$.
In order to see this, first we notice that $\sum_{\chi} \mtf b_3\hp{1,3}(\chi)$ equals 
\begin{align*}
4 \sum_{\nu,\boldsymbol \mu}
 (-1)^{1+|\sigma|}\delta_{\nu\mu_1\mu_2\mu_3} & \cdot  \big \{ 
  -\TrN X_\nu   \cdot \TrN(K_{\mu_1}K_{\mu_2}K_{\mu_3}) 
\\[-6pt]
&
+e_{\mu_2}e_{\mu_3}
\TrN K_{\mu_1}   \cdot \TrN(X_\nu K_{\mu_2}K_{\mu_3})
\\[1.5pt]
&
+
e_{\mu_1}e_{\mu_3}
\TrN K_{\mu_2}   \cdot \TrN(X_\nu K_{\mu_1}K_{\mu_3})
\\[-1.5pt]
&
+
e_{\mu_1}e_{\mu_2}
\TrN K_{\mu_3}   \cdot \TrN(X_\nu K_{\mu_1}K_{\mu_2})
\big\}\,,
\end{align*}
due to eq. \eqref{eq:tau2}
and $
e_{\hat \nu }e_{\mu_1} e_{\mu_2} e_{\mu_3} = -1$ (Claim \ref{thm:signos}).
But also departing from $\sum_\chi\mtf b_3\hp{3,1}(\chi)$,  
using 
\eqref{eq:dos}  to convert the triple 
signs to a single one,  (e.g. $e_{\hnu} e_{\mu_1} e_{\mu_3}=-e_{\mu_2}$),
renaming indices (which gets rid of the minus sign via 
the skew-symmetric factor $(-1)^{1+|\sigma|}$) 
one arrives again to the same expression. 
Thus $
\sum_{\chi} \mtf b_3\hp{1,3}(\chi)
+
\sum_{\chi} \mtf b_3\hp{3,1}(\chi)
$ equals
\[
 8 \sum_{\nu,\boldsymbol \mu}
 (-1)^{1+|\sigma|} \cdot \big \{ 
  -\TrN X_\nu   \cdot \TrN(K_{\mu_1}K_{\mu_2}K_{\mu_3}) 
+e_{\mu_2}e_{\mu_3}
\TrN K_{\mu_1}   \cdot \TrN(X_\nu K_{\mu_2}K_{\mu_3}) \big\}\,.
\]
Using the skew-symmetry of $(-1)^{1+|\sigma|}$
and the cyclicity of the trace, one proves
easily that the (2,2)-partition 
$\mathfrak b_3\hp{2,2}$ vanishes, and 
so does in fact $\mathfrak{b}_5\hp{2,2}$.
Thus the only contributions from 10-point 
diagrams are the partitions (1,3) and (3,1)
which can be computed similarly as for 
the 6-point contributions, by a similar
token. Thus 
\begin{align*}
\sum_{\chi\in \CD{10}}\mathfrak{b}_5(\chi)&  =
2\sum_{\chi\in \tau_5}\mathfrak{b}_5\hp{1,3}(\chi)
=-8\sum_{\mu,\nu_1,\nu_2,\nu_3}
\delta_{\mu \nu_1\nu_2\nu_3} (-1)^{|\lambda_{\boldsymbol\nu}|} e_{\nu_1}
e_{\nu_2}e_{\nu_3}
\\& \times 
\big[ e_\mu \TrN K_\mu \cdot \TrN(X_{\nu_1}X_{\nu_2}X_{\nu_3})+
e_{\nu_1} \TrN X_{\nu_1} \cdot \TrN(K_{\mu}X_{\nu_2}X_{\nu_3})\\
& \hphantom{\big[}+
e_{\nu_2} \TrN X_{\nu_2} \cdot \TrN(K_{\mu}X_{\nu_1}X_{\nu_3})
+
e_{\nu_3} \TrN X_{\nu_3} \cdot \TrN(K_{\mu}X_{\nu_1}X_{\nu_2})
\big]
\end{align*}
By performing the sum of the terms in the 
last line one sees that they cancel out due
to the skew-symmetry of $(-1)^{|\lambda_{\boldsymbol \nu}|}$. 
The only contribution come therefore from the two first terms in the square brackets, which are directly 
seen to yield
\begin{align*}\sum_{\chi\in \CD{10}}\mathfrak{b}_5(\chi)
=
&-8\sum_{\mu,\nu_1,\nu_2,\nu_3}
\delta_{\mu \nu_1\nu_2\nu_3} (-1)^{|\lambda_{\boldsymbol\nu}|}
\times 
\big[ (-1)^q \TrN K_\mu \cdot \TrN(X_{\nu_1}X_{\nu_2}X_{\nu_3})
\\& \hspace{140pt}+
e_{\nu_2}e_{\nu_3} \TrN X_{\nu_1} \cdot \TrN(K_{\mu}X_{\nu_2}X_{\nu_3})
\big] \,.
\end{align*}
Concerning the 8-point diagrams,
\begin{align}  
\sum_{\chi\in \CD{8}}  \mathfrak b_4 (\chi )  
& =  \sum_{ \substack{\mu_1,\mu_2 ,\nu_1,\nu_2 \\ \text{not all equal}}} 
\big( 2\chi^{\mu_1\hnu_1 \mu_2 \hnu_2} +
4 \chi^{\mu_1 \mu_2 \hnu_1 \hnu_2} \big) \{\mbox{non-trivial partitions}\} \nonumber \\
& +\sum_\mu (
2\chi^{\mu\hmu \mu \hmu} +
4 \chi^{\mu \mu  \hmu \hmu} )  \{\mbox{non-trivial partitions}\} \,. \label{eq:costam}
\end{align}
The sum over the 8-point chord diagrams is spitted in 
the $\tau_3$ and $\tau_4$ types with their symmetry factors;
in each line these are, respectively, the two summands in parenthesis. 
Here `non-trivial partitions' in curly brackets refers to (1,3), (2,2) and (3,1). 
We call the second line $\boldsymbol \Delta$, for 
which straightforward computation yields
\begin{align}\nonumber
\boldsymbol \Delta&=
 4(-1)^{1+q}\sum_{\mu=1}^4
\Big\{ 
2e_\mu \TrN(\Km)\cdot \TrN(\Km\Xm^2)+2e_{\hat \mu} \TrN(\Xm) \cdot \TrN(\Xm \Km^2) \\
&\hspace{18ex} +\TrN(\Xm^2)\cdot \TrN(\Km^2) + 2(-1)^{1+q}\big[\TrN(\Km \Xm)\big]^2
\Big\}
 \end{align}
by rewriting $e_1e_2e_3e_4=(-1)^{q}$.
We now compute the first line of eq. \eqref{eq:costam}
 considering first
only the $\tau_3$ diagrams (the $\tau_4$-type is addressed later).
The sum $(1,3)+ (3,1)$ of partitions  can be straightforwardly obtained:
\begin{align} \label{eq:triplicar1}
&  2 \sum_{ \substack{\mu_1,\mu_2 ,\nu_1,\nu_2 \\ \text{not all equal}}} 
 \chi^{\mu_1\hnu_1 \mu_2 \hnu_2}  \{(1,3)+(3,1)\mbox{ partitions}\} \\ \nonumber
 &= 2 \sum_{ \substack{\mu_1,\mu_2 ,\nu_1,\nu_2 \\ \text{not all equal}}} \bigg[-(-1)^{\mu_1+\mu_2}
e_1e_2e_3e_4 
(\delta_{\mu_1}^{\nu_1}\delta_{\mu_2}^{\nu_2}
-
\delta_{\mu_1}^{\nu_2}\delta_{\mu_2}^{\nu_1})
-
e_{\mu_1} \Big(\displaystyle\prod_{\alpha\neq \nu_1} e_\alpha \Big)
\delta_{\mu_1}^{\mu_2}\delta_{\nu_1}^{\nu_2} \bigg]   \nonumber
\\ & \times \bigg[
e_{\mu_1} \TrN(K_{\mu_1})
\cdot \TrN\big( X_{\nu_1} \{ K_{\mu_2}, X_{\mu_2}\}\big)
+
e_{\mu_2} \TrN(K_{\mu_2})
\cdot \TrN\big( K_{\mu_1} \{ X_{\nu_1}, X_{\mu_2}\}\big)  \nonumber
\\
& \quad+
e_{\nu_1} \TrN(X_{\nu_1})
\cdot \TrN\big( K_{\mu_1} \{ K_{\mu_2}, X_{\nu_2}\}\big) 
+e_{\nu_2} \TrN(X_{\nu_2})
\cdot \TrN\big( K_{\mu_1} \{ K_{\mu_2}, X_{\nu_1}\}\big) 
\bigg] \nonumber \\[1pt]
& = -8 \sum_{\mu\neq \nu}  \Big(\displaystyle\prod_{\alpha\neq \nu} e_\alpha
\Big) \Big(\TrN(K_\mu)\cdot \TrN\big(K_\mu X_{\nu}^2\big) +
 e_\mu e_{\hnu} \TrN(X_\nu)\cdot \TrN\big(K_\mu^2 X_{\nu}\big) \Big) \nonumber\\[3pt]
 & =8 \sum_{\mu\neq \nu}  
 (-1)^{1+q} e_\nu 
 \TrN(K_\mu)\cdot \TrN\big(K_\mu X^2_{\nu}\big) +
 e_\mu \TrN(X_\nu)\cdot \TrN\big(K_\mu^2 X_{\nu}\big)  \,. \nonumber
 \end{align}
 In the first equality 
we just used the expression for $\chi^{\mu_1\hnu_1 \mu_2 \hnu_2}$. 
In order to obtain the second one, it can be shown 
that the terms proportional to $(\delta_{\mu_1}^{\nu_1}\delta_{\mu_2}^{\nu_2}
-
\delta_{\mu_1}^{\nu_2}\delta_{\mu_2}^{\nu_1})$ cancel out. Using 
eq. \eqref{eq:mhm} one simplifies the signs to obtain the last equality.
The condition of $\mu\neq \nu$ in the sum of the last 
equations 
reflects only the fact that the four indices cannot 
coincide (cf. assumptions in Prop. \ref{thm:deltas}). 
The remaining partition reads

\begin{align}\label{eq:triplicar2}
&  2 \sum_{ \substack{\mu_1,\mu_2 ,\nu_1,\nu_2 \\ \text{not all equal}}} 
 \chi^{\mu_1\hnu_1 \mu_2 \hnu_2} \times \{(2,2)\mbox{ partition}\} \\
 & = -4  \nonumber
 \sum_{\mu\neq \nu}
 \Big(\displaystyle\prod_{\alpha\neq \nu} e_\alpha
\Big) 
 \Big[ 2e_{\hnu}  \TrN(K_\mu X_\nu)^2 + e_\mu \TrN(K_\mu^2 ) \cdot \TrN(X_\nu^2)  
\Big]\\
&= +4 \sum_{\mu\neq \nu}
 \Big\{ 
 2  \big[\TrN(K_\mu X_\nu)\big]^2 + e_\mu e_\nu (-1)^{q+1} \TrN(K_\mu^2 ) \cdot \TrN(X_\nu^2)  
\Big\} \,.\nonumber
\end{align} 
Using a similar approach (which would be redundant here), 
one can similarly show that the contribution of the 
$\tau_4$-diagrams is precisely twice that of $\tau_3$, obtaining in total
$3\times $[eqs.\eqref{eq:triplicar1} + \eqref{eq:triplicar2}] for the 
8-point diagrams.
The claim follows from 
\begin{align}
\mathcal S_4=  
\sum_{r=2}^6 
\sum_{\chi\in\CD{2r}}
\mathfrak s_r(\chi) \qquad \and \qquad  
\mathcal B_4= 
\sum_{\pi} 
\sum_{r=2}^6 
\sum_{\chi\in\CD{2r}}
\mathfrak b_r^{\pi}(\chi) 
\end{align}
where $\pi$ runs 
over the non-trivial partitions $\pi \in \{(1,3)\,, (2,2)\,, (3,1)\}$.
\end{proof}

\subsection{Riemannian and Lorentzian geometries} \label{sec:RiemLor}
Before writing down the action 
functionals for Riemannian and Lorentzian geometries,
it will be helpful to restate eqs. \eqref{eq:NCPolynomial4} and \eqref{eq:B4Haupt} 
via 
\begin{salign}
-\sum_{\alpha,\beta,\mu,\nu} 
\delta_{\alpha\beta\mu\nu}
e_\alpha e_\beta 
\TrN\big[ (K_\mu X_\nu)^2+ 2K_\mu^2 X_\nu^2 \big] \\
=(-1)^{1+q}
\sum_{\mu\neq \nu}
2 e_\mu e_\nu 
\TrN\big[ (K_\mu X_\nu)^2+ 2K_\mu^2 X_\nu^2 \big] 
\end{salign} 
and by writing out (`cycl.' next means equality after cyclic reordering) 
\begin{align} \nonumber
8(-1)^{q+1} \sum_{\mu,\boldsymbol \nu} &
(-1)^{|\sigma(\boldsymbol \nu,\mu)|} \delta_{\mu\nu_1\nu_2\nu_3} 
e_\mu( X_\mu K_{\nu_1} K_{\nu_2} K_{\nu_3}
+  K_\mu X_{\nu_1} X_{\nu_2} X_{\nu_3} )
\\ \label{eq:useful1}
\stackrel{\mtr{cycl.}}{\equiv} 
 &- 8 (-1)^q \Big[ e_1 X_1\big(K_2[K_3,K_4] +K_3[K_4,K_2]+K_4[K_2,K_3] \big)
\\ & \hspace{37pt} \vphantom{\bigg)} \nonumber 
+
e_2 X_2\big(K_1[K_3,K_4] +K_3[K_4,K_1]+K_4[K_1,K_3] \big) 
\\
 & \hspace{37pt} \vphantom{\bigg)} \nonumber 
+
e_3 X_3\big(K_1[K_2,K_4] +K_2[K_4,K_1]+K_4[K_1,K_2]  \big) \\  \nonumber
 & \hspace{37pt}\vphantom{\bigg)}
+
e_4 X_4\big(K_1[K_2,K_3] +K_2[K_3,K_1]+K_3[K_1,K_2]  \big) 
\Big] \nonumber \\&  \vphantom{\bigg)}
 -8 (-1)^q \Big[ e_1 K_1\big(X_2[X_3,X_4] +X_3[X_4,X_2]+X_4[X_2,X_3] \big)  \nonumber
\\ & \hspace{37pt} \vphantom{\bigg)} \nonumber 
+
e_2 K_2\big(X_1[X_3,X_4] +X_3[X_4,X_1]+X_4[X_1,X_3] \big) 
\\
 & \hspace{37pt} \vphantom{\bigg)} \nonumber 
+
e_3 K_3\big(X_1[X_2,X_4] +X_2[X_4,X_1]+X_4[X_1,X_2]  \big) \\
 & \hspace{37pt} \vphantom{\bigg)}\vphantom{\bigg)}
+
e_4 K_4\big(X_1[X_2,X_3] +X_2[X_3,X_1]+X_3[X_1,X_2]  \big) 
\Big]  \nonumber \,,
\end{align}
as well as 
\begin{salign}
& \vphantom{\bigg)}
 -8 \sum_{\nu,\boldsymbol \mu } 
 (-1)^{|\sigma(\boldsymbol \mu,\nu)|} \delta_{\nu\mu_1\mu_2\mu_3}  \cdot \big \{ 
  -\TrN X_\nu  \cdot \TrN(K_{\mu_1}K_{\mu_2}K_{\mu_3}) 
\\[-8pt]
& \vphantom{\bigg)} \hspace{8pt}  +e_{\mu_2}e_{\mu_3} \big(
\TrN K_{\mu_1}    \cdot \TrN(X_\nu K_{\mu_2}K_{\mu_3})
+\TrN X_{\mu_1}  \cdot \TrN(K_\nu X_{\mu_2}X_{\mu_3})  \nonumber
\big)\\
& \vphantom{\bigg)} \hspace{8pt}  
+(-1)^q \TrN K_\nu  \cdot \TrN(X_{\mu_1}X_{\mu_2}X_{\mu_3}) 
\big\}\,  \label{eq:useful2} \numerada
\\
= & \vphantom{\bigg)} +24  \TrN X_1  \cdot \TrN\big(K_2[K_3,K_4]\big)   +24
   \TrN X_2 \cdot \TrN \big ( K_1[K_3,K_4] \big) \\
   &  \vphantom{\bigg)} +24\TrN X_3 \cdot \TrN \big (K _1[K_2,K_4] \big)  +24
   \TrN X_4 \cdot \TrN \big (K _1[K_2,K_3] \big)  
   \\
  & \vphantom{\bigg)}  -8 \TrN K_1  \cdot \TrN\big( e_3 e_4 [K_3,K_4] X_2 + e_2 e_4 [K_2,K_4] X_3 
  + e_2 e_3 [K_2,K_3] X_4\big) \\
  & \vphantom{\bigg)}  -8 \TrN K_2  \cdot \TrN\big( e_3 e_4 [K_3,K_4] X_1 + e_1 e_4 [K_4,K_1] X_3 
  + e_1 e_3 [K_3,K_1] X_4\big) 
  \\
  &  \vphantom{\bigg)} -8 \TrN K_3  \cdot \TrN\big( e_2 e_4 [K_2,K_4] X_1 + e_1 e_4 [K_4,K_1] X_2 
  + e_1 e_2 [K_1,K_2] X_4\big) \\
  & \vphantom{\bigg)}
   -8 \TrN K_4  \cdot \TrN\big( e_2 e_3 [K_2,K_3] X_1 + e_1 e_3 [K_1,K_3] X_2 
  + e_1 e_2 [K_1,K_2] X_3\big)
  \\
 & \vphantom{\bigg)}  +
(-1)^{1+q}   \big\{24 \TrN K_1  \cdot \TrN\big(X_2[X_3,X_4]\big) 
  +24
   \TrN K_2 \cdot \TrN \big ( X_1[X_3,X_4] \big) \\
   &  \vphantom{\bigg)}\hspace{44pt} +24 \TrN K_3 \cdot \TrN \big (X_1[X_2,X_4] \big)  +24
   \TrN K_4 \cdot \TrN \big (X_1[X_2,X_3] \big)  \big\}
   \\
  & \vphantom{\bigg)}   - 8   \TrN X_1  \cdot \TrN\big( e_3 e_4 [X_3,X_4] K_2 + e_2 e_4 [X_2,X_4] K_3 
  + e_2 e_3 [X_2,X_3] K_4\big) \\
  &  \vphantom{\bigg)} - 8 \TrN X_2  \cdot \TrN\big( e_3 e_4 [X_3,X_4] K_1 + e_1 e_4 [X_4,X_1] K_3 
  + e_1 e_3 [X_3,X_1] K_4\big) 
  \\
  &   \vphantom{\bigg)}  -8 \TrN X_3  \cdot \TrN\big( e_2 e_4 [X_2,X_4] K_1 + e_1 e_4 [X_4,X_1] K_2 
  + e_1 e_2 [X_1,X_2] K_4\big) \\
  &  \vphantom{\bigg)}-8
  \TrN X_4  \cdot \TrN\big( e_2 e_3 [X_2,X_3] K_1 + e_1 e_3 [X_1,X_3] K_2 
  + e_1 e_2 [X_1,X_2] K_3\big)\,.
\end{salign}
From these expressions
the Riemannian and Lorentzian 
cases can be readily derived. 

\subsubsection{Riemannian fuzzy geometries}
The metric $g=\diag(-1,-1,-1,-1)$ implies $e_\mu=-1$ for each $\mu\in \Delta_4$ 
and $q=4$. 
The Dirac operator $D=D(\mathbf L, \tilde{\mathbf { H}} )$  (Ex. \ref{ex:4d})
is parametrized by four anti-Hermitian matrices 
$K_\mu=L_\mu$ (where $[L_\mu,\balita\hspace{4pt}]$ corresponds to the 
derivatives $\partial_\mu$ in the smooth case) and 
four Hermitian matrices $X_\mu= {\tilde{H}}_{\mu}$  (corresponding to the spin connection $\omega_\mu$ 
in the smooth spin geometry case 
represented by $[\tilde{H}_\mu,\balita\,]$ here).
In Example \ref{ex:4d} above
these have been called ${\tilde{H}}_{1}=H_{234},\ldots,{\tilde{H}}_{4}=H_{123}$.
The bi-tracial octo-matrix model has the 
following quadratic part
\begin{equation} \label{eq:B2orazS2Riemann}
\frac{1}{8}
\Tr \big( [D^{\mtr{Riemann}}]^2 \big)
=
N \sum_{\mu=1}^4 \Tr [ 
{\tilde{H}}^2_{\mu}
-L_{\mu}^2  ]
+  (\TrN
{\tilde{H}}_{\mu} )^2
+ (\TrN
L_\mu)^2\,,
\end{equation}
which directly follows from eq. \eqref{eq:generalD2d4}.
The quartic part is
more complicated:
\[
\frac{1}{4}
\Tr \big( [D^{\mtr{Riemann}}]^4 \big) =N \mtc S_4^{\mtr{Riemann}} +  \mtc B_4^{\mtr{Riemann}}
\]
having single-trace action
\begin{align} \nonumber 
\mtc S_4^{\mtr{Riemann}}&=  \TrN \Big\{ 
2 \sum_\mu  (L_{\mu}^4  +  {\tilde{H}}_{\mu}^4)  
\\ & \nonumber +
4\sum_{\mu<\nu }   (2 L_{\mu}^2L_{\nu}^2 +2 {\tilde{H}}_{\mu}^2{\tilde{H}}_{\nu}^2  
-L_{\mu}L_{\nu}L_{\mu}L_{\nu} -{\tilde{H}}_{\mu}{\tilde{H}}_{\nu}{\tilde{H}}_{\mu}{\tilde{H}}_{\nu} ) \nonumber  
\\ &  \nonumber  
-
\sum_{\mu\neq \nu}
\big[2  (L_\mu {\tilde{H}}_\nu)^2+ 4L_\mu^2 {\tilde{H}}_\nu^2 \big]
 +   \sum_\mu\big[ 2(L_\mu {\tilde{H}}_{\mu})^2- 4L_\mu^2 {\tilde{H}}_{\mu}^2  \big]  \\
\nonumber
&+ 8   \Big[  {\tilde{H}}_1\big(L_2[L_3,L_4] +L_3[L_4,L_2]+L_4[L_2,L_3] \big)
\\ & \hspace{11pt} \vphantom{\bigg)} \nonumber 
+
  {\tilde{H}}_2\big(L_1[L_3,L_4] +L_3[L_4,L_1]+L_4[L_1,L_3] \big) 
\\
 & \hspace{11pt} \vphantom{\bigg)} \nonumber 
+
 {\tilde{H}}_3\big(L_1[L_2,L_4] +L_2[L_4,L_1]+L_4[L_1,L_2]  \big) \\  \nonumber
 & \hspace{11pt}\vphantom{\bigg)}
+
  {\tilde{H}}_4\big(L_1[L_2,L_3] +L_2[L_3,L_1]+L_3[L_1,L_2]  \big) 
\Big] \nonumber \\&  \vphantom{\bigg)}
  + 8   \Big[  L_1\big({\tilde{H}}_2[{\tilde{H}}_3,{\tilde{H}}_4] +{\tilde{H}}_3[{\tilde{H}}_4,{\tilde{H}}_2]+{\tilde{H}}_4[{\tilde{H}}_2,{\tilde{H}}_3] \big)  \nonumber
\\ & \hspace{11pt} \vphantom{\bigg)} \nonumber 
+
 L_2\big({\tilde{H}}_1[{\tilde{H}}_3,{\tilde{H}}_4] +{\tilde{H}}_3[{\tilde{H}}_4,{\tilde{H}}_1]+{\tilde{H}}_4[{\tilde{H}}_1,{\tilde{H}}_3] \big) 
\\
 & \hspace{11pt} \vphantom{\bigg)} \nonumber 
+
  L_3\big({\tilde{H}}_1[{\tilde{H}}_2,{\tilde{H}}_4] +{\tilde{H}}_2[{\tilde{H}}_4,{\tilde{H}}_1]+{\tilde{H}}_4[{\tilde{H}}_1,{\tilde{H}}_2]  \big) \\
 & \hspace{11pt} \vphantom{\bigg)}\vphantom{\bigg)}
+
  L_4\big({\tilde{H}}_1[{\tilde{H}}_2,{\tilde{H}}_3] +{\tilde{H}}_2[{\tilde{H}}_3,{\tilde{H}}_1]+{\tilde{H}}_3[{\tilde{H}}_1,{\tilde{H}}_2]  \big) 
\Big] 
\Big\}\, ,  
\label{eq:NCPolynomial4Riem}
\end{align}
and bi-tracial action 
\begin{align} \label{eq:B4Riem}
 \mtc B_4^{\mtr{Riemann}}&= 
 8
\sum_{\mu,\nu}  
\TrN {\tilde{H}}_{\mu} 
\cdot \TrN({\tilde{H}}_{\mu} {\tilde{H}}_\nu^2) 
 - \TrN(L_{\mu}) \cdot  \TrN(L_{\mu} L_{\nu}^2)
 \\
 &
 +\sum_{\mu,\nu=1}^4  \nonumber
\Big\{ 2\TrN({\tilde{H}}_{\mu}^2)\cdot\TrN({\tilde{H}}_{\nu}^2) +4   \big[\TrN({\tilde{H}}_{\mu}{\tilde{H}}_{\nu})\big]^2  \Big\}
\\  \nonumber
&+\sum_{\mu,\nu=1}^4 
\Big\{ 2\TrN(L_{\mu}^2)\cdot\TrN(L_{\nu}^2) +4   \big[\TrN(L_{\mu}L_{\nu})\big]^2  \Big\}
\\  \nonumber
  &+ 4\sum_{\mu=1}^4
\Big\{ 
2 \TrN L_{\mu} \cdot \TrN(L_{\mu}{\tilde{H}}_{\mu}^2)-2  \TrN {\tilde{H}}_{\mu}  \cdot \TrN({\tilde{H}}_{\mu} L_{\mu}^2) \\
&\hspace{38pt} -\TrN({\tilde{H}}_{\mu}^2)\cdot \TrN(L_{\mu}^2) + 2\big[\TrN(L_{\mu} {\tilde{H}}_{\mu})\big]^2
\Big\}  \nonumber
    \\  \nonumber
& \vphantom{\bigg)} +24 \sum_{\mu\neq \nu=1}^4  
 \TrN L_\mu \cdot \TrN\big(L_\mu X^2_{\nu}\big) -\TrN {\tilde{H}}_\nu \cdot \TrN\big(L_\mu^2 {\tilde{H}}_{\nu}\big)  
 \\  \nonumber
  & \vphantom{\bigg)} +12 \sum_{\mu\neq \nu}
 \Big\{ 
 2 \big[\TrN(L_\mu {\tilde{H}}_\nu)\big]^2 - \TrN(L_\mu^2 ) \cdot \TrN({\tilde{H}}_\nu^2)  
\Big\}  \nonumber
\\ 
 &  \vphantom{\bigg)}\nonumber +24  \TrN {\tilde{H}}_1  \cdot \TrN\big(L_2[L_3,L_4]\big)   +24
   \TrN {\tilde{H}}_2 \cdot \TrN \big ( L_1[L_3,L_4] \big) \\
   &  \vphantom{\bigg)} \nonumber +24\TrN {\tilde{H}}_3 \cdot \TrN \big (L_1[L_2,L_4] \big)  +24
   \TrN {\tilde{H}}_4 \cdot \TrN \big (L_1[L_2,L_3] \big)  
   \\
  & \vphantom{\bigg)}  \nonumber -8 \TrN L_1  \cdot \TrN\big(  [L_3,L_4] {\tilde{H}}_2 +  [L_2,L_4] {\tilde{H}}_3 
  +  [L_2,L_3] {\tilde{H}}_4\big) \\
  & \vphantom{\bigg)} \nonumber  -8 \TrN L_2  \cdot \TrN\big(  [L_3,L_4] {\tilde{H}}_1 +  [L_4,L_1] {\tilde{H}}_3 
  +  [L_3,L_1] {\tilde{H}}_4\big) 
  \\
  &  \vphantom{\bigg)}\nonumber  -8 \TrN L_3  \cdot \TrN\big(  [L_2,L_4] {\tilde{H}}_1 +  [L_4,L_1] {\tilde{H}}_2 
  +  [L_1,L_2] {\tilde{H}}_4\big) \\
  &  \vphantom{\bigg)}\nonumber
   -8 \TrN L_4 \cdot \TrN\big(  [L_2,L_3] {\tilde{H}}_1 +  [L_1,L_3] {\tilde{H}}_2 
  +  [L_1,L_2] {\tilde{H}}_3\big)
  \\
 &  \vphantom{\bigg)} \nonumber -    24 \TrN L_1  \cdot \TrN\big({\tilde{H}}_2[{\tilde{H}}_3,{\tilde{H}}_4]\big) 
  -24
   \TrN L_2 \cdot \TrN \big ( {\tilde{H}}_1[{\tilde{H}}_3,{\tilde{H}}_4] \big) \\ \nonumber
   &  \vphantom{\bigg)} -24 \TrN L_3 \cdot \TrN \big ({\tilde{H}}_1[{\tilde{H}}_2,{\tilde{H}}_4] \big)  -24
   \TrN L_4 \cdot \TrN \big ({\tilde{H}}_1[{\tilde{H}}_2,{\tilde{H}}_3] \big)  
   \\
  &  \vphantom{\bigg)} \nonumber  - 8   \TrN {\tilde{H}}_1  \cdot \TrN\big(  [{\tilde{H}}_3,{\tilde{H}}_4] L_2 +  [{\tilde{H}}_2,{\tilde{H}}_4] L_3 
  +  [{\tilde{H}}_2,{\tilde{H}}_3] L_4\big) \\
  & \vphantom{\bigg)}  \nonumber - 8 \TrN {\tilde{H}}_2  \cdot \TrN\big(  [{\tilde{H}}_3,{\tilde{H}}_4] L_1 +  [{\tilde{H}}_4,{\tilde{H}}_1] L_3 
  +  [{\tilde{H}}_3,{\tilde{H}}_1] L_4\big) 
  \\
  &  \vphantom{\bigg)} \nonumber   -8 \TrN {\tilde{H}}_3  \cdot \TrN\big(  [{\tilde{H}}_2,{\tilde{H}}_4] L_1 +  [{\tilde{H}}_4,{\tilde{H}}_1] L_2 
  +  [{\tilde{H}}_1,{\tilde{H}}_2] L_4\big) \\
  & \vphantom{\bigg)} \nonumber -8
  \TrN {\tilde{H}}_4 \cdot \TrN\big(  [{\tilde{H}}_2,{\tilde{H}}_3] L_1 +  [{\tilde{H}}_1,{\tilde{H}}_3] L_2 
  +  [{\tilde{H}}_1,{\tilde{H}}_2] L_3\big) \,  .  
\end{align}

\subsubsection{Lorentzian fuzzy geometries}  \label{sec:Lor}
Here we keep the usual conventions: the index $0$ for 
time and (undotted) Latin spatial indices $a=1,2,3$. 
In the Lorentzian setting   
$g=\diag(+1,-1,-1,-1)$, so $q=3$, $e_0=+1$
and $e_a=-1$ for each spatial $a$. 
A parametrization of the Dirac operator
of the form $D=D(H, L_a, \Ht, \Lt_a)$   
by six anti-Hermitian matrices 
$K_a=L_a$, $X_a=\Lt_a$ and 
two Hermitian matrices $K_0=H$ and $X_0=\Ht$ follows.
As before, we give first the quadratic part 
and then the quartic. The former follows from
eq. \eqref{eq:generalD2d4},
\begin{align}  \nonumber
 \frac{1}{8}\Tr D^2
 &= N \TrN \big\{ H^2+\Ht^2-\textstyle\sum_a (L_a^2+\Lt_a^2)\big\}\\
 & +(\TrN H)^2  +(\TrN \Ht)^2  +\sum_a(\TrN L_a)^2 
 +(\TrN \Lt_a)^2\,. \label{eq:B2orazS2Lorentz}
 \end{align}
Using eqs. \eqref{eq:useful1} and \eqref{eq:useful2}
to rewrite Proposition \ref{thm:Haupt},
one gets

\begin{align} \nonumber
\mtc S_4^{\mtr{Lorentz}}  &=  \TrN \Big\{ 
2 H^4 + 2\Ht^4 +
\sum_a \big(  L_{  a }^4 +  \Lt_{ a}^4  \big)
- \sum_a  \big[ 2 (L_a \Lt_a)^2  + 4 L_a^2 \big]
\\ \nonumber
& +
\sum_{a} \big[ 
-8H^2 L^2_{ a}
-8\Ht^2 \Lt^2_{ a}
+4(H L_{ a})^2
+
4(\Ht \Lt_{ a})^2\big]
\\
& +
\sum_{a<c } \big[ 8 L^2_{ a} L^2_{ c}
+8\Lt^2_{ a} \Lt^2_{ c} 
-4( L_{ a} L_{ c })^2
-4( \Lt_{ a} \Lt_{ c })^2
\big]
\nonumber  
\\ &  \nonumber  
- \sum_{a} \big[
2(H\Lt_a)^2 +
4 H^2 \Lt_a^2
+ 2( L_a \Ht)^2 +
4  L_a^2 \Ht^2
\big]
\\
\nonumber
&
+\sum_{a\neq c}
2(L_a \Lt_c)^2 + 4 L_a^2\Lt_c^2
-2 (H\Ht)^2 + 4   H^2\Ht^2 
\\
\nonumber
&+ 8   \Big[ \Ht\big(L_1[L_2,L_3] +L_2[L_3,L_1]+L_3[L_1,L_2] \big)
\\ & \hphantom{+8\Big[} \vphantom{\bigg)} \nonumber 
- \Lt_1\big(H[L_2,L_3] +L_2[L_3,H]+L_3[H,L_2] \big) 
\\
 & \hphantom{+8\Big[} \vphantom{\bigg)} \nonumber 
- \Lt_2\big(H[L_1,L_3] +L_1[L_3,H]+L_3[H,L_1]  \big) \\  \nonumber
 & \hphantom{+8\Big[}\vphantom{\bigg)}
- \Lt_3\big(H[L_1,L_2] +L_1[L_2,H]+L_2[H,L_1]  \big) 
\Big] \nonumber \\&  \vphantom{\bigg)}
  +8 \Big[  H\big(\Lt_1[\Lt_2,\Lt_3] +\Lt_2[\Lt_3,\Lt_1]+\Lt_3[\Lt_1,\Lt_2] \big)  \nonumber
\\ & \hphantom{+8\Big[} \vphantom{\bigg)} \nonumber 
- L_1\big(\Ht[\Lt_2,\Lt_3] +\Lt_2[\Lt_3,\Ht]+\Lt_3[\Ht,\Lt_2] \big) 
\\
 & \hphantom{+8\Big[} \vphantom{\bigg)} \nonumber 
- L_2\big(\Ht[\Lt_1,\Lt_3] +\Lt_1[\Lt_3,\Ht]+\Lt_3[\Ht,\Lt_1]  \big) \\
 & \hphantom{+8\Big[} \vphantom{\bigg)}\vphantom{\bigg)}
-L_3\big(\Ht[\Lt_1,\Lt_2] +\Lt_1[\Lt_2,\Ht]+\Lt_2[\Ht,\Lt_1]  \big) 
\Big] 
\Big\}\, ,  
\label{eq:NCPolynomial4Lor}
\end{align}
and

\begin{salign}
\mtc B_4^{\mtr{Lorentz}}  &= \vphantom{\sum_a}     \numerada \label{eq:B4Lor}
 8 \TrN \Ht \cdot \TrN  \Big\{\Ht^3
 -\textstyle\sum_a (\Ht \Lt_a^2+3L_a \Ht^2\big) +\Ht H^2 \\\vphantom{\sum_a}   
& \hspace{55pt} +
3 L_1[L_2,L_3] 
+
[\Lt_3,\Lt_2]L_1
+[\Lt_3,\Lt_1]L_2
+[\Lt_2,\Lt_2]L_3
\Big\} \\ & \vphantom{\sum_a}   
+8 \TrN H \cdot \TrN \Big\{
H^3-\textstyle\sum_a (HL_a^2 -3 H \Lt_a^2\big) + H \Ht^2
\\ \vphantom{\sum_a}   &  \hspace{55pt} 
+3 \Lt_1[\Lt_2,\Lt_3]
+[L_3,L_2]\Lt_1 
+[L_3,L_1]\Lt_2
+[L_2,L_1]\Lt_3
)
\Big\} \\ &
 \vphantom{\sum_a}   
 + 8 \sum_a \TrN \Lt_a   \cdot \TrN  \Big\{
 \Lt_aH^2  - \Lt_a \textstyle\sum_c \Lt_c^2 - 
 L_a \Lt_a^2 
 \\ &
 \vphantom{\sum_a}   \hspace{95pt}
+3 H^2 \Lt_a -3\textstyle \sum_{\substack{c (c\neq a)}}
 \Lt_a^2 L_c
\Big\}
     \\
  & \vphantom{\sum_a}   - 8 \TrN \Lt_1  \cdot \TrN\big(  [\Lt_2,\Lt_3] H - [\Lt_3,\Ht] L_2 
  - [\Lt_2,\Ht] L_3
+3
   H[L_3,L_2] \big) 
  \\
  & \vphantom{\sum_a}     -8 \TrN \Lt_2 \cdot \TrN\big(  [\Lt_1,\Lt_3] H - [\Lt_3,\Ht] L_1 
  - [\Ht,\Lt_1] L_3  
   +3 H[L_3,L_1] \big) \\
  & \vphantom{\sum_a}  -8
  \TrN \Lt_3   \cdot \TrN\big( [\Lt_1,\Lt_2] H - [\Ht,\Lt_2] L_1 
  - [\Ht,\Lt_1] L_2
   +3  H[L_2,L_1] \big)   
   \\ & \vphantom{\sum_a}   
 + 8 \sum_a \TrN L_a \cdot \TrN \Big\{
 L_aH^2  - L_a \textstyle\sum_c L_c^2 -\Lt_a L_a^2  
 \\ &
 \vphantom{\sum_a}   \hspace{95pt}
+ 3 L_a \Ht^2 -3\textstyle \sum_{\substack{c (c\neq a)}}
 L_a \Lt_c^2
\Big\} \\   
  & \vphantom{\sum_a}   -8 \TrN L_1  \cdot \TrN\big(   [L_2,L_3] \Ht - [L_3,H] \Lt_2 
  - [L_2,H] \Lt_3   +3  \Ht[\Lt_3,\Lt_2] \big)
  \\
  & \vphantom{\sum_a}   -8 \TrN L_2 \cdot \TrN\big(  [L_1,L_3] \Ht - [L_3,H] \Lt_1 
  - [H,L_1] \Lt_3     +3  \Ht[\Lt_3,\Lt_1] \big)  
    \\
  & \vphantom{\sum_a} 
   -8 \TrN L_3 \cdot \TrN\big(  [L_1,L_2] \Ht - [H,L_2] \Lt_1 
  - [H,L_1] \Lt_2  + 3\Ht[\Lt_2,\Lt_1] \big)  
  \\ 
  & \vphantom{\sum_a}  
 +6\big[\TrN{\Ht^2}\big]^2
+ \sum_a \big\{ 4 \pair{\Lt_a^2}{\Ht^2}
-8 \big[ \TrN(\Ht \Lt_a) \big]^2\big\}\\  
& \vphantom{\sum_a}+\sum_{a,c}
\Big( 2 \pair{\Lt_a^2}{\Lt_c^2} + 4 \big[\Tr(\Lt_a\Lt_c)\big]^2\Big)\\
& \vphantom{\sum_a}
+6\big[\TrN{H^2}\big]^2
+ \sum_a \big\{ 4 \pair{L_a^2}{H^2}
-8 \big[ \TrN(H L_a) \big]^2\big\}\\  
& \vphantom{\sum_a}+\sum_{a,c}
\Big( 2 \pair{L_a^2}{L_c^2} + 4 \big[\Tr(L_aL_c)\big]^2\Big)
\\  
 & \vphantom{\sum_a} +4 \pair{\Ht^2}{H^2} +8 \big[ \TrN(H\Ht)\big]^2 \\
 & \vphantom{\sum_a} +4 \sum_a \pair{\Lt_a^2}{L_a^2}
 +8 \big[ \TrN(L_a \Lt_a)\big]^2
\\ & \vphantom{\sum_a} 
+24\sum_a \big\{ \big[ \TrN(H\Lt_a) \big]^2 +   \big[ \TrN(\Ht L_a) \big]^2 \big\} 
 \\  
  & \vphantom{\sum_a}+  12 \sum_{ a\neq c} 
  2 \big[ \TrN(L_a \Lt_c) \big] ^2  + \pair{L_a^2}{\Lt_c^2}
\\ & \vphantom{\sum_a} -12 
\sum_a 
\big\{ 
\pair{L_a^2}{\Ht^2}
+\pair{\Lt_a^2}{H^2}
\big\} \,.
\end{salign}
\begin{remark}\label{thm:tracelessness}
Each anti-hermitian parametrizing matrix $L$ can be
replaced by a traceless one $L'=L-(\TrN L/N) \cdot 1_N$,  since $L$ appears in $D$ only via anti-commutators; since $\TrN L$ is purely imaginary, $L'$ is also 
anti-hermitian.
\end{remark}

\afterpage{\thispagestyle{plain}}
\section{Large-$N$ limit via free probability?}\label{sec:freeP}

Some evidence for viability of a free probabilistic 
approach to the large-$N$ limit is the 
numerical analysis \cite{BarrettGlaser} of 
 \[
 F=\frac{\sum_\mu  [\TrN H_\mu ]^2}{N\big[ \sum_\mu \TrN(H_\mu^2) \big]}\,.
 \]  
 The observable $F$ does not trivially 
vanish at $N\to \infty$, since the double trace $\TrN\cdot \TrN $ in the numerator 
competes with the denominator's $N  \times \TrN$.
In \cite[Fig. 14]{BarrettGlaser}, for the model 
$\Tr(D^4 - |\lambda| D^2)$
in some low-dimensional geometries of signatures
$(1,0), (1,1), (2,0)$ and $(0,3)$,  $F$  is plotted 
in the region $\lambda \in (1,5)$. For each 
signature, a vanishing $F \approx 0$ is found in the regions $   1 \lessapprox \lambda \lessapprox 2$. 
Preliminarily, $F$ is related to the quotient  $\mtc B_2/\mtc S_2$ 
of the bi-trace
 and single-trace functionals introduced here. This 
 motivates free probabilistic tools.

It is well-known that free probability and
(multi)matrix models are related \cite{nica_speicher,GuionnetJonesShlyakhtenko07,GuionnetSDE,NowakTarnowski}. One could start with  
noncommutative self-adjoint polynomials $P \in \re \langle 
x_1,\ldots,x_\kappa \rangle $  \cite{SpeicherLectures}.
To wit,  
$P(x_1,\ldots,x_\kappa)=P(x_1,\ldots,x_\kappa)^*$, 
if each of the noncommutative variables $x_i$ satisfies
formal self-adjointness $x_i^*=x_i$.
For instance, the next polynomials are self-adjoint: 
\begin{subequations} \label{eq:NCPs}
\begin{align}
 P_2(x_1,\ldots ,x_\kappa) &= x_1^2+ \ldots + x_\kappa^2  
 + \frac{ \lambda}2  \sum_{i\neq j} x_ix_j \,, \\
 P_4(x_1,\ldots ,x_\kappa) &=
  x_1^4+ \ldots + x_\kappa^4 +\frac12 \sum_{i\neq j}\big( \lambda_1  x_ix_jx_ix_j 
 + \lambda_2 x_i^2x_j^2 \big)\,,
\end{align}
\end{subequations}
being $ \lambda_i,\lambda$ real coupling constants.
One can instead evaluate $P$ in square matrices of size, say, $N$ 
and define 
\begin{align}
\dif \nu_N(\mathbb X\hp N)=
 \dif \nu_N(X_1\hp N,\ldots,X_\kappa\hp N) &= C_N  
 \cdot \ee^{-N^2 \TrN [P(X_1\hp N,\ldots,X_\kappa\hp N)] } \\ 
 &\qquad\qquad  \nonumber
 \cdot 
 \dif \Lambda ({X_1\hp N})\cdots \dif \Lambda ({X_\kappa\hp N})
 \end{align}
being $C_N$ a normalization constant and $\Lambda$ the 
Lebesgue measure
\[\dif \Lambda(Y)= \prod_i \dif Y_{ii} \prod_{i< j} 
\Re (\dif Y_{ij})
\Im (\dif Y_{ij}),  \qquad Y\in M_N(\C)\,.\]
The  distributions $\varphi_N$
defined by 
\begin{align}
\varphi_N(X_{j_1}\hp{N},\ldots,X_{j_k}\hp N)=\int \TrN(X_{j_1}\hp{N} \cdots  X_{j_1}\hp N) \dif \nu_N(\mathbb X\hp N)
\end{align}
and their eventual convergence to
distributions $\varphi (X_{j_1},\ldots , X_{j_k})$ for large $N$ 
is of interest in free probability. 
As we have shown in Sections 
\ref{sec:computeSA}, \ref{sec:d2} and \ref{sec:d4},
we have developed a geometrically interesting way to produce noncommutative
polynomials. Although these are not directly self-adjoint,
self-adjointness is not essential in order for one to ponder the possible
convergence of the measures they define.  As far as 
the trace does not detect it, a weaker notion suffices:

\begin{definition}
Given variables 
$z_1,\ldots, z_\kappa$,
each of which satisfies either formal
self-adjointness (i.e. for an involution $*$, $z^*_i= +z_i$ holds, in whose case we let $z_i=: h_i$) or formal 
anti-self-adjointness ($z^*_i= - z_i$; and if so write $z_i=: l_i$), 
a noncommutative (NC) polynomial $
P \in \re \langle 
z_1,\ldots,z_\kappa \rangle $ is said to be \textit{cyclic self-adjoint}
if the following conditions hold:

\begin{itemize}
\itemB 
for each word $w$ (or monomial) of $P$ there exists
a word $w'$ in $P$ such that
\begin{align}
 \qquad[w(z_1 ,\ldots,z_\kappa)]^*
 = +(\sigma \cdot w')(z_{1} ,\ldots,z_{\kappa}  ) 
  \text{ holds for some } \sigma \in \Z/ |w'|\Z \,,
 \end{align}
being 
\begin{itemize}
\itemW 
$|w|$ the length of the word $w$ (or order 
of the monomial $w$)
and 
\itemW $\sigma \cdot w'$ the action of $ \Z/ |w'|\Z$ on the word $w'$ by cyclic permutation 
of its letters.
\end{itemize}
\itemB The map defined by 
$w\mapsto w'$ is a bijection 
in the set of the words of $P$.
\end{itemize}
Similarly, a polynomial $G \in \re \langle z_1,\ldots, z_\kappa\rangle$
is \textit{cyclic anti-self-adjoint} if for each of its words $w$ 
if there exist a  $
   \sigma \in \Z/ |w'|\Z $ for which the condition 
\begin{align}
 [w(z_1 ,\ldots,z_\kappa)]^*
 =- (\sigma \cdot w')(z_{1} ,\ldots,z_{\kappa}  ) 
 \end{align}
holds, and if, additionally, the map that results from this condition, $w\mapsto w'$, is a bijection in the set of words
of $G$. 
\end{definition}

\begin{example} \label{thm:NCpolyn} Consider the formal adjoint $P^*$ of the NC polynomial $P$ given by
$P(h,l_1,l_2,l_3)=l_1  \{ h(l_2 l_3 -l_3 l_2) + l_2 (l_3 h -h l_3)+ l_3 (h l_2-l_2h) \}$. One obtains
\begin{align} \nonumber
[P(h,l_1,l_2,l_3)]^*& = 
(h[l_2,l_3] +l_2[l_3,h]+l_3[h,l_2])^* l_1^* \\
&= \{ (l_2 l_3-l_3l_2)h + (l_3 h-hl_3) l_2
+ (hl_2 -l_2h)l_3 ) \} l_1 \label{eq:Ppolyn}
\\
&=  ( h[l_2,l_3] + l_2[l_3,h] + l_3[h,l_2] ) l_1 \nonumber\,,
\end{align}
where $[\balita\,,\balita\,]$ is the commutator in 
$\re \langle h,l_1,l_2,l_3 \rangle $.
Clearly $P^* \neq P$, but 
up to the cyclic permutation 
$\sigma \in \Z/4 \Z$ defined by  
bringing the letter $l_1$ from the last 
to the first position of each 
word, a bijection of words in $P$ 
is established. Hence $P$ is cyclic self-adjoint. 
On the other hand, take $\Psi (h,l_2,l_3)= (l_2 l_3 -l_3l_2)h$. By 
a similar token, one sees that $\Psi$  is cyclic anti-self-adjoint. 
\end{example}
 
It would not be surprising that the spectral action $\Tr f(D)$ 
for fuzzy geometries (since it has to 
be real) in any dimension and allowed KO-dimension 
leads for any ordinary polynomial $f$ to 
the type of NC polynomials we just introduced. Preliminarily, we 
verify this statement only for the explicit computations
we performed in this article:
\begin{cor} 
For the following cases
\begin{itemize}
 \itemB $d=2$, in arbitrary signature and being $f$ a sextic polynomial; and  
 \itemB Riemannian and Lorentzian signatures ($d=4$) being $f$ quartic polynomial,
\end{itemize}
the spectral action $\Tr f(D) = \dim V (N\cdot \mtc S_f + \mtc B_f)$ 
for fuzzy geometries has the form 
\begin{equation}
 \label{eq:NCPolyninSA}
\mtc S_f= \TrN P \quad \text{and} \quad 
\mtc B_f = \sum_i \TrN  \Phi_i  \cdot \TrN \Psi_i\,,
 \end{equation}
where $P, \Phi_i, \Psi_i\in 
\re \langle x_1,\ldots,z_{\kappa(d)}\rangle$ with $\kappa(d)=2^{d-1}$ are NC polynomials 
such that 
\begin{itemize}
 \itemB $P$ is cyclic self-adjoint  
 \itemB and $\Phi_i$ and $\Psi_i$  are
both either cyclic self-adjoint or both 
cyclic anti-self-adjoint.  
\end{itemize}
\end{cor}

\begin{proof} 
Up to an irrelevant (as to assess cyclic self-adjointness) ambiguity in a global factor for $\Phi_i$ and $\Psi_i$,
all the NC polynomials can be read off 
from eqs. \eqref{eq:GeneralD2d2}, \eqref{eq:GeneralDcuatro} and 
Proposition \ref{thm:Dsix} for the $d=2$ case. 
For $d=4$, the result follows by inspection of each term,
which is immediate since formulae \eqref{eq:B4Lor}, \eqref{eq:B4Riem},
\eqref{eq:B2orazS2Riemann} and \eqref{eq:B2orazS2Lorentz} 
are given in terms of commutators. Then one uses 
that $[h,l]^*=[h,l]$, $[h_1,h_2]^*=-[h_1,h_2]$
and $[l_1,l_2]^*= - [l_1,l_2]$. 
\par 
The only non-obvious part is dealing with expressions like 
\[\mtc P= \Ht \{ L_1[L_2,L_3] +L_2[L_3,L_1]+L_3[L_1,L_2] \}\,,\] 
which appears in $\mtc S_4^{\mtr{Lorentz}}$ according to eq. \eqref{eq:NCPolynomial4Lor}. However, if $P$ is the NC polynomial given in eq. \eqref{eq:Ppolyn}
then $\mtc P$ equals 
$ P(Q,L_1,L_2,L_3) $, hence it is cyclic self-adjoint by Example \ref{thm:NCpolyn}.  Also for the NC polynomial $\Psi$ defined there, 
$-8\TrN  L_1  \cdot \TrN(\Psi(Q,L_2,L_3)) $ appears 
in the expression given by eq. \eqref{eq:B4Lor}
for $\mtc B_4^{\mtr{Lorentz}}$, being both $\Phi(L_1)=L_1$ and 
$\Psi$ cyclic
anti-self-adjoint. 
\end{proof}

\section{Conclusions} \label{sec:Conclusions} 
We computed the spectral action $\Tr f(D) $
for fuzzy geometries of even dimension $d$, whose `quantization' 
was stated as a
$2^{d-1}$-matrix model with action $ \Tr f(D) = \dim V  ( N \cdot \mtc S_f + \mtc B_f )$
being the single trace $\mtc S_f$ and bi-tracial 
parts $\mtc B_f$ of the form
\[
\qquad 
 \mtc S_f
=  \TrN F \,,
\qquad \mtc B_f=  \textstyle
\sum_i  (\TrN\otimes\TrN) \{  \Phi_i \otimes \Psi_i  \}\,.
\]
With the aid of chord diagrams
that encode non-vanishing traces on the spinor space $V$,
we organized the obtention of (finitely many) 
noncommutative polynomials $F$, $\Phi_i$ and $\Psi_i $  
in $2^{d-1}$ Hermitian or anti-Hermitian matrices in $M_N(\C)$. These 
polynomials are defined up to cyclic permutation of their words and
have integer coefficients that are independent of $N$.
We commented on a free probabilistic perspective
towards the large-$N$ limit of the spectral action
and adapted the concept of self-adjoint noncommutative
polynomial to a more relaxed one (cyclic self-adjointness) that 
is satisfied by $F$. Furthermore, for fixed $i$, either 
both $\Phi_i$ and $\Psi_i$ are cyclic self-adjoint 
or both are cyclic anti-self-adjoint.
\par 
On the one hand, we elaborated on 
2-dimensional fuzzy geometries in arbitrary signature $(p,q)$. 
When quantized (or randomized), the corresponding partition function is
\begin{align}
\mathcal Z\hp{p,q}= \int_{\mathcal M^{p,q}}
\ee^{-S(D)}
\dif{K_1}
\dif{K_2} \qquad D=D(K_1,K_2),\quad p+q=2\,.
\end{align} 
The space 
$\mathcal M^{p,q}$ of Dirac operators is
$
{\mathcal M^{p,q}} =
                      (\mathbb{H}_N)^{\times p}  \times  \mathfrak{su}(N)^{\times q} 
$,  but this simple parametrization 
does not generally hold for $d>2$. Here $\mathfrak {su}(N)=\mtr{Lie}(\mtr{SU}(N))$
stands for the (Lie algebra of) traceless $N\times N$ skew-Hermitian matrices
and $   \mathbb{H}_N$ for Hermitian matrices.  Concretely, 
in Section \ref{sec:d2} formulas for
$S(D)=\Tr(D^2 + \lambda_4 D^4 + \lambda_6 D^6 )$ are deduced,
but the present method enables to obtain $S(D)=\Tr f(D)$ 
for polynomial $f$. This first result is an 
extension (by the sextic term) of the spectral action presented by Barrett-Glaser \cite{BarrettGlaser}
up to quartic polynomials in $d=2$.
\par 
On the other hand, the novelties (to the best of our knowledge) are 
the analytic derivations we provided for  
Riemannian and Lorentzian fuzzy geometries---and in fact, 
in arbitrary signature in 4 dimensions---as well as a systematic approach 
that maps random fuzzy geometries to multi-matrix bi-tracial models. For one thing,
this sheds some light on arbitrary-dimensional geometries and, for other thing, 
on extensions to (quantum) models including bosonic fields.
For the quadratic-quartic spectral action $S(D)=\Tr(D^2 + \lambda_4 D^4 )$ 
computed in Section \ref{sec:d4} one could study the 
octo-matrix model 
\begin{align}
\mathcal Z\hp{p,q}= \int_{\mathcal M^{p,q}}
\ee^{-S(D)}
\textstyle\prod^4_{\mu=1} \dif{K_\mu}  \dif{X_\mu} \qquad  D=D(K_\mu,X_\mu),\quad  p+q=4 
\end{align}
being, in particular,
\begin{align}
{\mathcal M^{p,q}} =\begin{cases}
                     \mathbb{H}_N^{\times 4} \times \mathfrak{su}(N)^{\times 4} & p=0,\,q=4 \mbox{ (Riemannian)}, \\
                      \mathbb{H}_N^{\times 2} \times \mathfrak{su}(N)^{\times 6}  & p=1,\, q=3 \mbox{ (Lorentzian)}.
                    \end{cases}
\end{align}  
For rest of the signatures, ${\mathcal M^{p,q}}$
                  can be readily obtained with the aid of
                  the eq. \eqref{eq:mhm} as described for the Lorentzian 
                  and Riemannian cases.
As a closing point, it is pertinent to remark that 
determining whether the Dirac operator of a fuzzy geometry
is a 
truncation of a spin$^c$ geometry is a subtle 
problem addressed in 
\cite{GlaserStern} from the viewpoint of the Heisenberg 
uncertainty principle \cite{ChamseddineConnesMukhanov}.

\section{Outlook}\label{sec:Outlook}
We present a miscellanea of short outlook
topics after elaborating on the next: 

\subsection{An auxiliary model for the $d=1$ case} \label{sec:Aux}
One could reformulate the partition functions of $d=1$-models 
in terms of auxiliary models that do not contain multi-traces. 
We pick for concreteness the 
signature $(p,q)=(1,0)$ 
and the polynomial $f(x)=(x^2+\lambda x^4)/2$ 
for the spectral action $\Tr f (D)$.  
We explain why the ordinary matrix model 
given by
$ 
\Zc\peqsubind{(1,0)}^{\mtr{aux}}=\int_{ \mathbb{H}_N} \ee^{  \eta \Tr H-\alpha \Tr (H^2) +\beta  \Tr (H^3)
-N\lambda  \Tr (H^4)  }  \dif H
$ 
over the Hermitian $N\times N$ matrices $ \mathbb{H}_N$,
allows to restate the quartic-quadratic (1,0)-type  Barrett-Glaser 
model with partition function 
\begin{align} \label{eq:ZBG10}
  \Zc\peqsubind{(1,0)}^{\mtr{BG}}= \int_{\mtc M} \ee^{-\frac12\Tr (D^2+\lambda D^4)}\dif D
\end{align}
as formally equivalent to the functional 
\begin{align} \label{eq:EquivModels}
 \langle \exp \{-(3\lambda \Tr H^2 +4 \lambda \Tr H \cdot  \Tr H^3 + (\Tr H)^2)\}\rangle_{\mtr{aux},0}\,\,,
\end{align}
where the expectation value of an observable $\Phi$ is taken with respect to 
the auxiliary model 
\[
\langle \Phi   \rangle_{\mtr{aux}}
= \frac{1}{\Zc\peqsubind{(1,0)}^{\mtr{aux}}}
\int_{ \mathbb{H}_N}   \Phi(H) \ee^{-\mathscr S(H)} \dif H \,,
\]
being 
$
\mathscr S(H)
=
\alpha \Tr (H^2) + N\lambda  \Tr (H^4) +  \eta \Tr H + \beta  \Tr (H^3)
$.
The zero subindex `aux,0' means evaluation 
in the parameters 
\begin{equation}
\label{eq:param}
\alpha=N, \gamma=\beta =0\,.\end{equation}

Indeed, one can use the explicit form 
of the Dirac operator $D=\{H,  \balita\, \}$ 
to rewrite the integral in terms of 
the matrix $H$. One gets 
\begin{align} \label{eq:CambioVar}
\frac12 \Tr (D^2 +\lambda D^4) &=  N [ \Tr(H^2)+ \lambda \Tr(H^4) ]  \\
&\,\,\,\,+  3 \lambda [\Tr(H^2)]^2 +4\lambda \Tr H \cdot \Tr H^3  + [\Tr (H)]^2   
\nonumber 
\\ \nonumber 
&=   N [ \Tr(H^2)+ \lambda \Tr(H^4) ]+ \mathfrak{b}[H]\,.
\end{align}
The second line of eq. \eqref{eq:CambioVar} contains
the bi-tracial terms; this term will be denoted by $\mathfrak{b}(H)$.
Inserting last equations into \eqref{eq:ZBG10}
\begin{align*}
 \Zc\peqsubind{(1,0)}^{\mtr{BG}} & =  \int _{ \mathbb{H}_N} \ee^{-N \Tr (H^2 + \lambda H^4)} \ee^{-\mathfrak{b}[H]}
 \dif H  
\end{align*}
Since $\mathscr S(H)|_{\alpha=N, \gamma=\beta =0}=   N [ \Tr(H^2)+ \lambda \Tr(H^4) ]$,
one can replace the first exponential by $\ee^{-\mathscr S(H)}$ and 
evaluate the parameters as in eq. \eqref{eq:param}:
\begin{align*}
 \Zc\peqsubind{(1,0)}^{\mtr{BG}} & =  \int _{ \mathbb{H}_N} \Big[\ee^{-  \mathscr S(H)}\Big]_{0}
 \ee^{-\mathfrak{b}[H]}
 \dif H   \,.
\end{align*}
If one knows the partition function $\Zc^{\mtr{aux}}\peqsubind{(1,0)}$,
one can compute the model in question by taking out 
$\ee^{-\mathfrak{b} (H)}$ from the integral and 
accordingly substituting the traces by the appropriate derivatives:
\begin{align}
 \Zc\peqsubind{(1,0)}^{\mtr{BG}}= 
\bigg[ \ee^{-\mathfrak{b}_{\partial}}\int _{ \mathbb{H}_N}  \ee^{-  \mathscr S(H)}
 \dif H \bigg]_{0}\, \qquad
\where
\mathfrak{b}_{\partial}= 3\lambda \partial_\alpha^2+4\lambda \partial_\beta \partial_\gamma
 +\partial_{\alpha}\,.
\end{align}
That is $ \Zc\peqsubind{(1,0)}^{\mtr{BG}}=  [\ee^{-\mathfrak{b}_{\partial} } 
\Zc^{\mtr{aux}}\peqsubind{(1,0)} ]_0 
$, which also proves eq. \eqref{eq:EquivModels}.
This motivates to look for similar methods in order
restate, for $d\geq 2$, the bi-tracial part of the 
models addressed here as single-trace auxiliary multi-matrix models.

 \subsection{Miscellaneous}

 \begin{itemize}

 \itemB \textit{Gauge theory.} The NCG-framework pays off
 in high energy physics precisely for gauge-Higgs theories.
 A natural step would be to come back 
 to this initial motivation and to define 
 \textit{almost commutative fuzzy geometries} (ongoing 
 project) in order to derive from them the
 Yang-Mills--Higgs theory on a fuzzy base.
 \vspace{7pt}
 
 \itemB \textit{Analytic approach}. A  
  non-perturbative approach to matrix models, 
 which led to the solvability of all quartic matrix models \cite{RaimarSolution} (after key progress in \cite{RaimarPanzer}) 
 consists in exploiting the $\mtr U(N)$-Ward-Takahashi
 identities in order to descend the tower of the Schwinger-Dyson (or loop) equations (SDE).
 This was initially formulated for a 
 quartic analogue of Kontsevich's model \cite{Kontsevich:1992ti},
 but the Grosse-Wulkenhaar approach (SDE + Ward Identity \cite{GW12})
 showed also applicability to tensor field theory 
\cite{fullward,SDE}, and seems to be flexible.
 \vspace{7pt}
 
  \itemB \textit{Topological Recursion.} Probably
  the analytic approach would lead to a  
  (or multiple) Topological Recursion (TR), as it appeared 
  in \cite{RaimarSolution}. Alternatively, one could build upon the direct TR-approach \cite{KhalkhaliTR}. Namely, the blobbed \cite{BorotBlobbed}  Topological Recursion \cite{EynardOrantin,EynardTopologicalRecursion,Chekhov:2006rq,
Chekhov:2006vd} has been 
lately applied  \cite{KhalkhaliTR} to
 general multi-trace models that encompass 
 the 1-dimensional version of the models
 derived here. An extension of their 
TR to dimension $d\geq 2$ would be interesting. 
 \vspace{7pt}

 \itemB \textit{Combinatorics.} Finally, chord diagrams are combinatorially interesting by themselves. 
 For instance, together with decorated versions known as 
 Jacobi and Gau\ss$ $ diagrams, they are used in algebraic knot theory \cite[Secs 3.4 and 4]{chmutov_duzhin_mostovoy_2012} in order to describe Vassiliev invariants. 
 Those appearing here are related to 
 the Penner matrix model \cite{Penner:1988cza}. One can still explore their generating function
  \cite{GenChordDiags,Andersen:2016umz} in relation to 
  the matrix model with action 
  \[S(X)=\TrN [X^2/2 - st\cdot(1-tX)\inv]\,, \] for $X\in  \mathbb{H}_N$. The free energies 
  $\mathcal F_g$ of this Andersen-Chekhov-Penner-Reidys-Sułkowski (ACPRS) model $\mathcal Z\peqsupind{ACPRS}=\ee^{\sum_g N^{2-2g} \mathcal F_g}$
  generate numbers that are moreover important in computational biology, 
as they encode topologically non-trivial complexes of interacting RNA molecules. These
numbers are related to the isomorphism classes of chord diagrams 
with certain number cuts in the circle, 
leaving segments (`backbones', cf. \cite{GenChordDiags})
but also a connected diagram. For the ACPRS-model there is also a Topological 
Recursion (\textit{op. cit.}).

\end{itemize}

\begin{acknowledgements} 
I thank Lisa Glaser for comments that led to improvements of the manuscript
(in particular, I owe her Remark \ref{thm:tracelessness}), and Andrzej Sitarz
for hospitality. The author was supported by 
      the TEAM programme of the Foundation
for Polish Science co-financed by the European Union under the
European Regional Development Fund (POIR.04.04.00-00-5C55/17-00).
\end{acknowledgements}

\appendix
\section{Some properties of gamma matrices} \label{sec:App}

In order to deal with $d$-dimensional matrix
geometries we prove some of the properties of the corresponding 
gamma matrices.
    
First, notice that in any signature 
for each multi-index $I=(\mu_1,\ldots,\mu_r) \in \Lambda_d$ one has
\begin{align}
\label{eq:Inversion}
\ga{\mu_r}\cdots\ga{\mu_1}
=(-1)^{r(r-1)/2}
\ga{\mu_1}\cdots\ga{\mu_r}
=(-1)^{\lfloor r/2 \rfloor}
\ga{\mu_1}\cdots\ga{\mu_r}\,.
\end{align}
This can be proven by induction
on the number $r-1$ of products.
For $r=2$, this is just $\{\gam 1 ,\gam 2\} =0$,
which holds since the indices are different.
Suppose that eq. \eqref{eq:Inversion} holds
for an $r\in \N$. Then if $(\mu_1,\ldots,\mu_{r+1})\in \Lambda_d$,
one has
\begin{align*}
\gam {r+1}\gam {r} \cdots \gam{2}
\gam {1}&=(-1)^r(
\gam {r} \cdots \gam{2}
\gam {1})\gam {r+1}\\
& =(-1)^{r+r(r-1)/2}
\gam 1 \cdots \gam {r+1} \\
& =(-1)^{r(r+1)/2}
\gam 1 \cdots \gam {r+1}  
=(-1)^{\lfloor (r+1)/2 \rfloor}
\gam 1 \cdots \gam {r+1}\,.
\end{align*}
Now let us fix a signature $(p,q)$.
We dot the spacial indices $\dot c =1,\ldots,q$
and leave the temporal in usual Roman lowercase, $a=1,\ldots ,p$.
Given a mult-index $I=(a_1\ldots,a_t,\dot c_1,\ldots,\dot c_u)\in \Lambda_d$ 
(so $t\leq p$ and $u\leq q$) it will be useful
to know whether $\Gamma^I$ is Hermitian or anti-Hermitian. We compute its Hermitian conjugate $(\Gamma^I)^*$:  
\begin{align*} \label{eq:}
(\ga{a_1} \ldots  \ga {a_t}\ga {\dot c_1}\ldots \ga{\dot c_u} )^*&=
( \ga{\dot c_u})^*\ldots (\ga{\dot c_1})^*( \ga{a_t})^* \ldots  (\ga {a_1})^*
\\&=
(-1)^{u} \ga{\dot c_u}\ldots \ga{\dot c_1} \ga{a_t} \ldots  \ga {a_1}
\\
&=
(-1)^{u  + \lfloor (u+t)/2 \rfloor} \ga{a_1} \ldots  \ga {a_t}\ga {\dot c_1}\ldots \ga{\dot c_u}\,.
\end{align*}
With the conventions set in eq. \eqref{eq:ConvGammas}, have the following 
\begin{itemize}
 \itemB In Riemannian signature $(0,d)$
 a product of $u$ gamma matrices 
 associated to $I\in \Lambda_d$ 
 is (anti-)Hermitian if $u(u+1)/2$ is even (odd). 
 \itemB In $(d,0)$-signature 
 a product of $u$ gamma matrices 
 associated to $I\in \Lambda_d$
is Hermitian if $t(t-1)/2$ is even,
 and anti-Hermitian if it is odd.
 
\end{itemize}
In the main text, it will be useful
to know that 
\begin{equation}
 e_\mu e_{\hat \mu}
= (-1)^{q+1}\hspace{2cm} d=4,\mbox{ with signature } (p,q)\,. \label{eq:mhm}
\end{equation}
This follows from $
(\Gamma^{\hat \mu})^*
=e_{\hat \mu}
\Gamma^{\hat \mu}$, 
from $e_1e_2e_3e_4=(-1)^q$ 
and from
\[ 
(\Gamma^{\hat \mu})^*=
(\gamma^4)^*
\cdots 
\widehat{ (\gamma^\mu)^*}
\cdots
(\gamma^1)^*
=
-e_1e_2e_3e_4 e_{\mu} \gamma^1
\cdots 
\widehat{\gamma^\mu}
\cdots
\gamma^4 =-e_1e_2e_3e_4 e_{\mu} \Gamma^{\hmu}\,.
\]
%
%
 
 \section{Full computation of one chord diagram} \label{sec:appB}
 
Since $t=2n=6$ are constant in this section,
we drop the subindices $n$ in $\mathfrak{a}_n,\mathfrak{b}_n$
and $\mathfrak{s}_n$.
Exclusively in this appendix, we abbreviate the traces as follows:
\[
|\mu_i \mu_j \ldots  \mu_m| := \TrN(K_{\mu_i} K_{\mu_j}  \cdots K_{\mu_m} )\,
\qquad \mu_i,\mu_j,\ldots,\mu_m \in \{1,2\}\,.
\] 
Then the action functional $\mathfrak{a}(\chi)$ 
of a chord diagram $\chi$ of six points is given by  
\[
\mathfrak{a}(\chi) = 
\sum_{ \mu_1 \mu_2 \ldots  \mu_6} 
(-1)^{\mtr{cr}(\chi)} \Big(
\prod_{w\sim_\chi v} 
g^{\mu_v\mu_w}\Big)
\Big( \sum _{\Upsilon \in \mathscr P_6 }  
\big[
\prod_{i\in\Upsilon }e_{\mu_i}\big] \cdot 
|\mu(\Upsilon^c)| \cdot |\mu(\Upsilon)| 
\Big) \,,
\]
that is,
\begin{align*} 
\sum_{ \mu_1 \mu_2 \ldots  \mu_6}  &
(-1)^{\mtr{cr}(\chi)}
\prod_{w\sim_\chi v} 
g^{\mu_v\mu_w}\bigg[
N\big( \,| \mu_1 \mu_2 \ldots  \mu_6| + e  | \mu_6 \mu_5 \ldots  \mu_1| \,\big)\\
&+ \sum_i e_i \big( | \mu_1 \mu_2 \ldots \widehat{\mu_i}\ldots  \mu_6| + e | \mu_6 \mu_5 \ldots \widehat{\mu_i}\ldots  \mu_1|  \big)\cdot |\mu_i| \\
&+ \sum_{i < j} e_ie_j \big( | \mu_1 \mu_2 \ldots \widehat{\mu_{i\phantom{j}}\!}\ldots \widehat{\mu_j}\ldots   \mu_6| + e  
 | \mu_6 \mu_5 \ldots \widehat{\mu_j}\ldots \widehat{\mu_{i\phantom{j}}\!}\ldots  \mu_1| 
\big) \cdot |\mu_i\mu_j| \\
&+ \sum_{i < j < k} e_ie_je_k \Big( | \mu_1 \mu_2 \ldots \widehat{\mu_{i\phantom{j}}\!}\ldots \widehat{\mu_j}\ldots  \widehat{\mu_k}\ldots  \mu_6| \cdot |\mu_i\mu_j\mu_k|    \Big) 
\bigg]\,,
\end{align*} 
where $e=e(\mu_1,\ldots,\mu_6)=e_{\mu_1}\cdots e_{\mu_6}$. 
We just conveniently listed the terms corresponding to $\Upsilon$ and $\Upsilon^c$
together,
in the first line displaying those with 
$\# \Upsilon=0$ and $ \# \Upsilon=6$ (`trivial partitions');
in the second $\# \Upsilon=1$ or $5$;
on the third line $\# \Upsilon=2$ or $4$;
the fourth line corresponds
to the $\# \Upsilon=3$ cases. 
We also used the fact that  $e_\mu$ is a sign $\pm$, and that
$e \cdot  e_{\mu_i}e_{\mu_j}\cdots e_{\mu_v} $ 
equals the product of the $e_{\mu_r}$'s with  $r\neq i,j,\ldots,v$,
i.e. precisely those not appearing in $ e_{\mu_i}e_{\mu_j}\cdots e_{\mu_v} $.
But since  in the non-vanishing terms $e$ 
implies a repetition of indices by pairs,
$e\equiv 1$ for non-vanishing terms. Then we gain a 
factor $2$ for those terms (i.e. except for traces of three matrices) and $\mathfrak{a} (\chi)$ is therefore given by
\begin{subequations}
\label{eq:abcd}
\begin{align} 
\sum_{ \mu_1 \mu_2 \ldots  \mu_6} &
 \Big( (-1)^{\mtr{cr}(\chi)}
\prod_{w\sim_\chi v} 
g^{\mu_v\mu_v} \Big)\cdot \Big\{
\sum_{ \mu_1, \mu_2, \ldots,  \mu_6} 
2 N |\mu_1 \mu_2 \ldots  \mu_6| \label{eq:a} \\ \label{eq:b}
& + \sum_i 2 e_i  | \mu_1 \mu_2 \ldots \widehat{\mu_i}\ldots  \mu_6| \cdot |\mu_i| \\
&+ \sum_{i < j} 2 e_ie_j  | \mu_1 \mu_2 \ldots \widehat{\mu_{i\phantom{j}}\!}\ldots \widehat{\mu_j}\ldots   \mu_6|  \cdot |\mu_i\mu_j| \label{eq:c} \\
&+ \sum_{i < j < k} e_ie_je_k   | \mu_1 \mu_2 \ldots \widehat{\mu_{i\phantom{j}}\!}\ldots \widehat{\mu_j}\ldots  \widehat{\mu_k}\ldots  \mu_6| \cdot |\mu_i\mu_j\mu_k|    \Big\} \,. \label{eq:d}
\end{align}  
\end{subequations}
 We thus compute 
 the first diagram of
 6-points by giving line by line
 last expression. 
 We perform first the computation for the third line \eqref{eq:c} (since 
this is the longest)  explicitly, which can be expanded as 
\begin{align}2 \nonumber  
\sum_{\boldsymbol\mu} 
\raisebox{-.45\height}{\includegraphics[height=1.5cm]{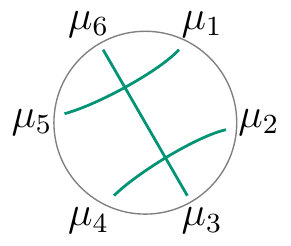}}\times
\big[&  e_{\mu_1}e_{\mu_2}
|\mu_1\mu_2 | \cdot |\mu_3\mu_4\mu_5\mu_6| 
+
e_{\mu_1}e_{\mu_3}
|\mu_1\mu_3 | \cdot | \mu_2  \mu_4\mu_5\mu_6| 
\\[-.41cm]+&   \nonumber 
e_{\mu_1}e_{\mu_4}
|\mu_1\mu_4 | \cdot | \mu_2 \mu_3 \mu_5\mu_6| 
+
e_{\mu_1}e_{\mu_5}
|\mu_1\mu_5| \cdot | \mu_2 \mu_3\mu_4 \mu_6|  \\\nonumber 
+ &  
e_{\mu_1}e_{\mu_6}
|\mu_1\mu_6| \cdot | \mu_2 \mu_3\mu_4\mu_5 | 
+
e_{\mu_2}e_{\mu_3}
|\mu_2\mu_ 3| \cdot |\mu_1 \mu_4\mu_5\mu_6| 
\\\nonumber 
+ &
e_{\mu_2}e_{\mu_4}
|\mu_2\mu_4 | \cdot |\mu_1  \mu_3  \mu_5\mu_6| 
+
e_{\mu_2}e_{\mu_5}
|\mu_2\mu_5 | \cdot |\mu_1  \mu_3\mu_4 \mu_6| 
\\+& \nonumber
e_{\mu_2}e_{\mu_6}
|\mu_2\mu_6 | \cdot |\mu_1  \mu_3\mu_4\mu_5 | 
+e_{\mu_3}e_{\mu_4}
|\mu_3\mu_4 | \cdot |\mu_1\mu_2  \mu_5\mu_6| 
\\+&\nonumber 
e_{\mu_3}e_{\mu_5}
|\mu_3\mu_5 | \cdot |\mu_1\mu_2 \mu_4\mu_6| 
+
e_{\mu_3}e_{\mu_6}
|\mu_3\mu_ 6| \cdot |\mu_1\mu_2 \mu_4\mu_5| 
\\\nonumber 
+&
e_{\mu_4}e_{\mu_5}
|\mu_4\mu_ 5| \cdot |\mu_1\mu_2 \mu_3\mu_6| 
+
e_{\mu_4}e_{\mu_6}
|\mu_4\mu_6 | \cdot |\mu_1\mu_2 \mu_3\mu_5| 
\\
+&
e_{\mu_5}e_{\mu_6}
|\mu_5\mu_6 | \cdot |\mu_1\mu_2 \mu_3\mu_4| \big]\,. \label{eq:explicitsix}
\end{align}
The diagram's meaning is the sign 
and product of $g^{\mu_v\nu_w}$'s 
before the braces in eq. \eqref{eq:abcd}. After contraction 
with the term in square brackets in 
\eqref{eq:explicitsix} one  gets 
\begin{align} \nonumber
2 \sum_{\mu,\nu,\rho} (e_\mu e_\nu e_\rho ) 
\big\{&
e_\rho |\mu\nu| \cdot |\rho\nu\mu\rho| 
+e_\rho |\mu\nu| \cdot |\rho\rho\mu\nu| +
e_\rho |\mu\nu| \cdot |\nu\rho\mu\rho|\\[-10pt] \nonumber
+& 
 e_\mu e_\rho e_\nu |\mu^2| \cdot |\nu\rho\nu\rho|
+ e_\rho |\mu\nu| \cdot |\rho\nu\rho\mu| +
e_\rho |\mu\nu| \cdot |\rho\mu\rho\nu| 
 \\
+&
 e_\mu e_\rho e_\nu |\mu^2| \cdot |\nu\rho\nu\rho|
+ e_\rho |\mu\nu| \cdot |\nu\rho\mu\rho|
+ e_\rho   |\mu\nu| \cdot |\rho\nu\mu\rho|
 \\
+&e_\rho |\mu\nu| \cdot |\rho\nu\rho\mu|
+e_\rho |\mu\nu| \cdot |\nu\rho^2\mu|  + e_\mu e_\rho e_\nu |\mu^2| \cdot |\nu\rho^2\nu| \nonumber
\\
+& 
e_\rho |\mu\nu| \cdot |\nu\mu\rho^2| +e_\rho |\mu\nu| \cdot |\rho\mu\nu\rho| 
+e_\rho |\mu\nu| \cdot |\mu\rho\nu\rho| 
\big\} \nonumber \,,
\end{align}
where the signs $ e_{\mu}e_{\nu}e_{\rho}$ 
are due to $g^{\lambda\sigma}=e_{\lambda}\delta^{\lambda\sigma}$ 
(no sum). Following the notation of eq. \eqref{eq:OW}, using 
the cyclicity of the trace and 
renaming indices, this expression can be written as 
\begin{align}\tag{\ref{eq:c}$'$}
2 \sum_{\mu,\nu,\rho}
(6R_{\mu\nu\rho}+6S_{\mu\nu\rho}+2T_{\mu\nu\rho}+U_{\mu\nu\rho}) \,.
\end{align}
Similarly, for the terms obeying 
$\#\Upsilon $ or $  \#\Upsilon^c = 1$,
i.e. line \eqref{eq:b}, one has
\begin{align*}
2 \sum_{\boldsymbol\mu} 
\raisebox{-.45\height}{\includegraphics[height=1.5cm]{figs/DiagI.pdf}}\times
\big\{&
e_{\mu_1}
|\mu_1| \cdot |\mu_2 \mu_3\mu_4\mu_5\mu_6| 
+
e_{\mu_2}
|\mu_2| \cdot |\mu_1 \mu_3\mu_4\mu_5\mu_6|
\\[-10pt]+&
e_{\mu_3}
|\mu_3| \cdot |\mu_1 \mu_2  \mu_4\mu_5\mu_6|
+
e_{\mu_4}
|\mu_4| \cdot |\mu_1 \mu_2 \mu_3 \mu_5\mu_6| \\[2pt]
 +& e_{\mu_5} |\mu_5| \cdot |\mu_1 \mu_2 \mu_3\mu_4\mu_6|
+e_{\mu_6}
|\mu_6| \cdot |\mu_1 \mu_2 \mu_3\mu_4\mu_5|
\big\}\,,
\end{align*}
which amounts to
\begin{align}\nonumber 
 &2\sum_{\mu,\nu,\rho}e_\nu e_\rho|\mu|\cdot \big\{ 
 |\nu\rho\nu\mu\rho| +|\nu\mu\rho\nu\rho|+ |\nu\rho\mu\nu\rho| 
 + |\nu\mu\rho\nu\rho| + |\mu\nu\rho\nu\rho|  +|\nu\rho\mu\rho\nu| \big\}\, ,
  \end{align}
or, relabeling, to
\begin{align}
2 \sum_{\mu,\nu,\rho} (4O_{\mu\nu\rho} +2P_{\mu\nu\rho}) \,.\tag{\ref{eq:b}$'$}
\end{align}
The terms with $\#\Upsilon = 3$ remain to be computed:
 \begin{align*}
\sum_{\boldsymbol\mu} 
\raisebox{-.45\height}{\includegraphics[height=1.5cm]{figs/DiagI.pdf}}
\times
\Big\{&
(e_{\mu_1}e_{\mu_2}e_{\mu_ 3}
+e_{\mu_4}e_{\mu_5}e_{\mu_ 6})
|\mu_1\mu_2 \mu_3| \cdot |\mu_4\mu_5\mu_6|  
\\[-12pt]
+&(e_{\mu_1}e_{\mu_2} e_{\mu_4}+e_{\mu_3}e_{\mu_5} e_{\mu_6})
|\mu_1\mu_2 \mu_4| \cdot |\mu_3\mu_5\mu_6| 
\\[3pt]
+&(e_{\mu_1}e_{\mu_2} e_{\mu_5}+ \emutres{3}{4}{6})
|\mu_1\mu_2  \mu_5| \cdot |\mu_3\mu_4\mu_6| 
\\[3pt]
+&(e_{\mu_1}e_{\mu_2} e_{\mu_6}+\emutres{3}{4}{5})
|\mu_1\mu_2  \mu_6| \cdot |\mu_3\mu_4\mu_5|
\\[3pt]+ &
(e_{\mu_1}e_{\mu_3}e_{\mu_4}+\emutres{2}{5}{6})
|\mu_1\mu_3\mu_4 | \cdot | \mu_2  \mu_5\mu_6| 
\\[3pt]
+&
(e_{\mu_1}e_{\mu_3}e_{\mu_5}+\emutres{2}{4}{6})
|\mu_1\mu_3\mu_5 | \cdot | \mu_2  \mu_4\mu_6| 
\\[3pt]+&
(e_{\mu_1}e_{\mu_3}e_{\mu_6}+\emutres{2}{4}{5})
|\mu_1\mu_3\mu_6 | \cdot | \mu_2  \mu_4\mu_5|
\\[3pt]  +&
(e_{\mu_1}e_{\mu_4}e_{\mu_5}+\emutres{2}{3}{6})
|\mu_1\mu_4\mu_5 | \cdot | \mu_2  \mu_3\mu_6|    
 \\[3pt]
  +
&(e_{\mu_1}e_{\mu_4}e_{\mu_6}+\emutres{2}{3}{5})
|\mu_1\mu_4\mu_6 | \cdot | \mu_2  \mu_3\mu_5| 
\\[-1pt]
+&(
e_{\mu_1}e_{\mu_5}e_{\mu_6}+\emutres{2}{3}{4})
|\mu_1\mu_5\mu_6 | \cdot | \mu_2  \mu_3\mu_4|  
\Big\} \,.
\end{align*}
 Although 
$\TrN(M_1M_2M_3)=\TrN(M_3M_2M_1)$ is false  for general matrices $M_1,M_2$ and $M_3$
 (e.g. for $M_j=\sigma_j$, the Pauli matrices),
 having at our disposal only two matrices, $K_1$ and $K_2$,
 the relation  $\TrN(K_\mu K_\nu K_\rho)=\TrN(K_\rho K_\nu K_\mu)$
does hold. This fact was used to obtain the last equation. 
Contracting with the diagram, as we already did for 
other partitions, one gets 
\begin{align}
\sum_{\mu,\nu,\rho}8 V_{\mu\nu\rho}+12 W_{\mu\nu\rho} \,.
\tag{\ref{eq:d}$'$}
\end{align}
By collecting 
the terms from the three equations with primed tags, 
the bi-trace term for the $\mtr I$-diagrams
one obtains
\begin{align*}
 \mathfrak b (\mathrm{I}) &= + 2 \sum_{\mu,\nu,\rho} 
   \big( 4 O_{\mu\nu\rho}  + 2 P_{\mu\nu\rho}  
        +6 R_{\mu\nu\rho} +  6 S_{\mu\nu\rho}  \phantom{+ 4 O_{\mu\nu\rho}}  \\
  & \hspace{1.523cm} +  2  T_{\mu\nu\rho} +  U_{\mu\nu\rho}
 + 4V_{\mu\nu\rho} + 6 W_{\mu\nu\rho}\big) \,,
\end{align*}
which is a claim amid the 
proof of Proposition \ref{thm:Dsix}. \par 
Notice that these integer coefficients 
add up to 62, and so will these 
(denoted $p_{\chi},q_{\chi},\ldots, w_{\chi}$ in 
the main text) in absolute value  for a general diagram $\chi$.
There are two missing terms 
to get the needed $2^6=\#\mathscr P_6$ terms.
These are the trivial cases $\Upsilon, \Upsilon^c = \emptyset$,
which can be
readily computed. 

For the $\mtr I$-diagram, 
\begin{align*}
\mathfrak s(\mtr{I})
&= 
2N \cdot\sum_{\boldsymbol\mu} 
\raisebox{-.45\height}{\includegraphics[height=1.5cm]{figs/DiagI.pdf}}
\times |\mu_1\mu_2\mu_3\mu_4\mu_5\mu_6|\\
&=  2 N\cdot\sum_{\mu,\nu,\rho} e_{\mu}e_{\nu}e_{\rho} |\mu\nu\rho\nu\mu\rho| \\[3pt]
&=2N\cdot\TrN \big\{
e_1 K_1^6+2  e_2 (K_1K_2)^2K_1^2 + e_1 (K_2^2K_1)^2 \\
&\hspace{1.77cm} +e_2 K_2^6+2  e_1 (K_2K_1)^2K_2^2 + e_2 (K_1^2K_2)^2
\big\}\,.
\end{align*}
The single-trace action $\mathcal S_6$ in
Proposition \ref{thm:Dsix} is then
obtained by summing
over all 6-point chord diagrams $\sum_\chi \mathfrak s(\chi)$, whose values are found by 
 similar computations. 
 \allowdisplaybreaks[0]

\bibliographystyle{alpha}

\begin{thebibliography}{AFMPS17}

\bibitem[ACPRS13]{GenChordDiags}
Jørgen~E. Andersen, Leonid~O. Chekhov, Robert~C. Penner, Christian~M. Reidys,
  and Piotr Sułkowski.
\newblock {Topological recursion for chord diagrams, RNA complexes, and cells
  in moduli spaces}.
\newblock {\em Nucl. Phys.}, B866:414--443, 2013.
\newblock arXiv:1205.0658.

\bibitem[AFMPS17]{Andersen:2016umz}
Jørgen~Ellegaard Andersen, Hiroyuki Fuji, Masahide Manabe, Robert~C. Penner,
  and Piotr Sułkowski.
\newblock {Partial chord diagrams and matrix models}.
\newblock {\em Trav. Math.}, 25:233, 2017.
\newblock arXiv:1612.05840.

\bibitem[AK19]{KhalkhaliTR}
Shahab Azarfar and Masoud Khalkhali.
\newblock {Random Finite Noncommutative Geometries and Topological Recursion}.
\newblock 2019.
\newblock arXiv:1906.09362.

\bibitem[Bar07]{BarrettSM}
John~W. Barrett.
\newblock {A Lorentzian version of the non-commutative geometry of the standard
  model of particle physics}.
\newblock {\em J. Math. Phys.}, 48:012303, 2007.
\newblock {hep-th/0608221}.

\bibitem[Bar15]{BarrettMatrix}
John~W. Barrett.
\newblock {Matrix geometries and fuzzy spaces as finite spectral triples}.
\newblock {\em J. Math. Phys.}, 56(8):082301, 2015.
\newblock arXiv:1502.05383.

\bibitem[BDG19]{BarrettDruceGlaser}
John~W. Barrett, Paul Druce, and Lisa Glaser.
\newblock {Spectral estimators for finite non-commutative geometries}.
\newblock {\em J. Phys.}, A52(27):275203, 2019.
\newblock arXiv:1902.03590.

\bibitem[Bes19]{BesnardUone}
Fabien Besnard.
\newblock {A $U(1)_{B-L}$-extension of the Standard Model from Noncommutative
  Geometry}.
\newblock 2019.
\newblock arXiv:1911.01100.

\bibitem[BF19]{Jordan}
Latham Boyle and Shane Farnsworth.
\newblock {The standard model, the Pati-Salam model, and ``Jordan geometry''}.
\newblock 2019.
\newblock {arXiv:1910.11888}.

\bibitem[BG16]{BarrettGlaser}
John~W. Barrett and Lisa Glaser.
\newblock {Monte Carlo simulations of random non-commutative geometries}.
\newblock {\em J. Phys.}, A49(24):245001, 2016.
\newblock arXiv:1510.01377.

\bibitem[BG19]{BarrettGaunt}
John~W. Barrett and James Gaunt.
\newblock {Finite spectral triples for the fuzzy torus}.
\newblock 2019.
\newblock {arXiv:1908.06796}.

\bibitem[BS20]{Bochniak:2020lab}
Arkadiusz Bochniak and Andrzej Sitarz.
\newblock {A spectral geometry for the Standard Model without the fermion
  doubling}.
\newblock 2020.
\newblock arXiv:2001.02902.

\bibitem[Bor15]{BorotBlobbed}
Ga{\"e}tan Borot.
\newblock {Blobbed topological recursion}.
\newblock {\em Theor. Math. Phys.}, 185(3):1729--1740, 2015.
\newblock [Teor. Mat. Fiz.185,no.3,423(2015)].

\bibitem[Car19]{CarlipDimRed}
Steven Carlip.
\newblock {Dimension and Dimensional Reduction in Quantum Gravity}.
\newblock {\em Universe}, 5:83, 2019.
\newblock arXiv:1904.04379.

\bibitem[CC97]{Chamseddine:1996zu}
Ali~H. Chamseddine and Alain Connes.
\newblock {The Spectral action principle}.
\newblock {\em Commun. Math. Phys.}, 186:731--750, 1997.
\newblock hep-th/9606001.

\bibitem[CCM07]{CCM}
Ali~H. Chamseddine, Alain Connes, and Matilde Marcolli.
\newblock {Gravity and the standard model with neutrino mixing}.
\newblock {\em Adv. Theor. Math. Phys.}, 11(6):991--1089, 2007.

\bibitem[CCM14]{ChamseddineConnesMukhanov}
Ali~H. Chamseddine, Alain Connes, and Viatcheslav Mukhanov.
\newblock {Geometry and the Quantum: Basics}.
\newblock {\em JHEP}, 12:098, 2014.
\newblock arXiv:1411.0977.

\bibitem[CCvS18]{EntropySpectral}
Ali~H. Chamseddine, Alain Connes, and Walter~D. van Suijlekom.
\newblock {Entropy and the spectral action}.
\newblock 2018.
\newblock arXiv:1809.02944.

\bibitem[CDM12]{chmutov_duzhin_mostovoy_2012}
Sergei Chmutov, Sergei Duzhin, and Jacob Mostovoy.
\newblock {\em Introduction to Vassiliev Knot Invariants}.
\newblock Cambridge University Press, 2012.
\newblock arXiv:1103.5628.

\bibitem[CE06]{Chekhov:2006rq}
Leonid Chekhov and Bertrand Eynard.
\newblock {Matrix eigenvalue model: Feynman graph technique for all genera}.
\newblock {\em JHEP}, 12:026, 2006.
\newblock math-ph/0604014.

\bibitem[CEO06]{Chekhov:2006vd}
Leonid Chekhov, Bertrand Eynard, and Nicolas Orantin.
\newblock {Free energy topological expansion for the 2-matrix model}.
\newblock {\em JHEP}, 12:053, 2006.
\newblock math-ph/0603003.

\bibitem[CM07]{ConnesMarcolli}
Alain Connes and Matilde Marcolli.
\newblock {\em {Noncommutative Geometry, Quantum Fields and Motives}}.
\newblock {American Mathematical Society}, 2007.

\bibitem[Con94]{ConnesNCG}
Alain Connes.
\newblock {\em {Noncommutative geometry}}.
\newblock Academic Press,
\newblock 1994.

\bibitem[Con13]{Reconstruction}
Alain Connes.
\newblock {On the spectral characterization of manifolds}.
\newblock {\em J. Noncommut. Geom.}, 7:1--82, 2013.
\newblock arXiv:0810.2088.

\bibitem[Con19]{ConnesHighlights}
Alain Connes.
\newblock {Noncommutative Geometry, the spectral standpoint}.
\newblock 2019.
\newblock {arXiv:1910.10407}.

\bibitem[CvS19]{surveySpectral}
Ali~H. Chamseddine and Walter~D. van Suijlekom.
\newblock {A survey of spectral models of gravity coupled to matter}.
\newblock 2019.
\newblock arXiv:1904.12392.

\bibitem[DDS18]{DabrowskiSitarzDAndrea}
Ludwik D\c{a}browski, Francesco D'Andrea, and Andrzej Sitarz.
\newblock {The Standard Model in noncommutative geometry: fundamental fermions
  as internal forms}.
\newblock {\em Lett. Math. Phys.}, 108(5):1323--1340, 2018.
\newblock [Erratum: \textit{Lett. Math. Phys.} 109, no.11, 2585 (2019)]
  arXiv:1703.05279.

\bibitem[DFR95]{DoplicherFredenhagenRoberts}
Sergio Doplicher, Klaus Fredenhagen, and John~E. Roberts.
\newblock {The Quantum structure of space-time at the Planck scale and quantum
  fields}.
\newblock {\em Commun. Math. Phys.}, 172:187--220, 1995.
\newblock hep-th/0303037.

\bibitem[DHMO08]{DolanHuetOConnor}
Brian~P. Dolan, Idrish Huet, Sean Murray, and Denjoe O'Connor.
\newblock {A Universal Dirac operator and noncommutative spin bundles over
  fuzzy complex projective spaces}.
\newblock {\em JHEP}, 03:029, 2008.
\newblock arXiv:0711.1347.

\bibitem[DK19]{KhalkhaliQuantization}
Rui Dong and Masoud Khalkhali.
\newblock {Second Quantization and the Spectral Action}.
\newblock 2019.
\newblock arXiv:1903.09624.

\bibitem[DO03]{DolanOConnor}
Brian~P. Dolan and Denjoe O'Connor.
\newblock {A Fuzzy three sphere and fuzzy tori}.
\newblock {\em JHEP}, 10:060, 2003.
\newblock hep-th/0306231.

\bibitem[EI19]{EcksteinIochum}
Michał Eckstein and Bruno Iochum.
\newblock {\em {Spectral Action in Noncommutative Geometry}}.
\newblock SpringerBriefs in Mathematical Physics, 2019.
\newblock arXiv:1902.05306.

\bibitem[EO07]{EynardOrantin}
Bertrand Eynard and Nicolas Orantin.
\newblock {Invariants of algebraic curves and topological expansion}.
\newblock {\em Commun. Num. Theor. Phys.}, 1:347--452, 2007.

\bibitem[Eps70]{epstein}
{David B.A.} Epstein.
\newblock Simplicity of certain groups of homomorphisms.
\newblock {\em Compositio Math.}, 22:165--173, 1970.

\bibitem[Eyn14]{EynardTopologicalRecursion}
Bertrand Eynard.
\newblock {A short overview of the ``Topological recursion''}.
\newblock 2014.
\newblock arXiv:1412.3286.

\bibitem[GHW19]{RaimarSolution}
Harald Grosse, Alexander Hock, and Raimar Wulkenhaar.
\newblock {Solution of all quartic matrix models}.
\newblock 2019.
\newblock arXiv:1906.04600.

\bibitem[Gil95]{Gilkey}
Peter~B. Gilkey.
\newblock {Invariance theory, the heat equation and the Atiyah-Singer index
  theorem}.
\newblock 1995.


\bibitem[GP95]{Grosse:1994ed}
Harald~Grosse and Peter~Pre\v{s}najder.
\newblock {The Dirac operator on the fuzzy sphere}.
\newblock {\em Lett. Math. Phys.}, 33:171--182, 1995.


\bibitem[Gla17]{GlaserScaling}
Lisa Glaser.
\newblock {Scaling behaviour in random non-commutative geometries}.
\newblock {\em J. Phys.}, A50(27):275201, 2017.
\newblock arXiv:1612.00713.


\bibitem[Gla]{GlaserCode}
Lisa Glaser.
\newblock Code for simulating random fuzzy spaces.
\newblock  Accessed 10.January.2020. \href{https://github.com/LisaGlaser/MCMCv4}{\color{black}  \color{blue}\texttt{https://github.com/LisaGlaser/MCMCv4}.}
 
\bibitem[GS19]{GlaserStern}
Lisa Glaser and Abel Stern.
\newblock {Understanding truncated non-commutative geometries through computer
  simulations}.
\newblock 2019.
\newblock arXiv:1909.08054.

\bibitem[GW14]{GW12}
Harald Grosse and Raimar Wulkenhaar.
\newblock {Self-Dual Noncommutative $\phi^4$ -Theory in Four Dimensions is a
  Non-Perturbatively Solvable and Non-Trivial Quantum Field Theory}.
\newblock {\em Commun. Math. Phys.}, 329:1069--1130, 2014.
\newblock arXiv:1205.0465.

 

\bibitem[GJS07]{GuionnetJonesShlyakhtenko07}
Alice Guionnet, Vaughan Jones, and Dimitri Shlyakhtenko.
\newblock Random matrices, free probability, planar algebras and subfactors.
\newblock 2007.
\newblock arXiv:0712.2904.
 
\bibitem[GN14]{GuionnetSDE}
Alice Guionnet and Jonathan Novak.
\newblock Asymptotics of unitary multimatrix models: The schwinger-dyson
  lattice and topological recursion.
\newblock 2014.
\newblock arXiv:1401.2703.


\bibitem[HP03]{PaschkeLattices}
Alexander Holfter and Mario Paschke.
\newblock {Moduli spaces of discrete gravity. 1. A Few points...}
\newblock {\em J. Geom. Phys.}, 47:101, 2003.
\newblock hep-th/0206142.

\bibitem[Kon92]{Kontsevich:1992ti}
M.~Kontsevich.
\newblock {Intersection theory on the moduli space of curves and the matrix
  Airy function}.
\newblock {\em Commun. Math. Phys.}, 147:1--23, 1992.

\bibitem[Kra98]{KrajewskiDiagr}
Thomas Krajewski.
\newblock {Classification of finite spectral triples}.
\newblock {\em J. Geom. Phys.}, 28:1--30, 1998.
\newblock hep-th/9701081.

\bibitem[Liz18]{LizziCorfu}
Fedele Lizzi.
\newblock {Noncommutative Geometry and Particle Physics}.
\newblock {\em PoS}, CORFU2017:133, 2018.
\newblock arXiv:1805.00411.

\bibitem[Mad92]{MadoreS2}
John Madore.
\newblock {The Fuzzy sphere}.
\newblock {\em Class. Quant. Grav.}, 9:69--88, 1992.

\bibitem[MS19]{Martinetti:2019hhq}
Pierre Martinetti and Devashish Singh.
\newblock {Lorentzian fermionic action by twisting euclidean spectral triples}.
\newblock 2019.
\newblock arXiv:1907.02485.

\bibitem[Mat74]{mather}
John Mather.
\newblock Commutators of diffeomorphisms.
\newblock {\em Commentarii Mathematici Helvetici}, 49:512--528, 1974.

\bibitem[Mat75]{mather2}
John Mather.
\newblock Commutators of diffeomorphisms: {II}.
\newblock {\em Commentarii Mathematici Helvetici}, 50:33--40, 1975.

\bibitem[MvS14]{MvS}
Matilde Marcolli and Walter~D. van Suijlekom.
\newblock Gauge networks in noncommutative geometry.
\newblock {\em Journal of Geometry and Physics}, 75:71 -- 91, 2014.
\newblock arXiv:1301.3480.

\bibitem[NS06]{nica_speicher}
Alexandru Nica and Roland Speicher.
\newblock {\em Lectures on the Combinatorics of Free Probability}.
\newblock London Mathematical Society Lecture Note Series. Cambridge University
  Press, 2006.

\bibitem[NT18]{NowakTarnowski}
Maciej~A. Nowak and Wojciech Tarnowski.
\newblock {Probing non-orthogonality of eigenvectors in non-Hermitian matrix
  models: diagrammatic approach}.
\newblock {\em JHEP}, 06:152, 2018.
\newblock arXiv:1801.02526.

\bibitem[Pen88]{Penner:1988cza}
Robert~C. Penner.
\newblock {Perturbative series and the moduli space of Riemann surfaces}.
\newblock {\em J. Diff. Geom.}, 27(1):35--53, 1988.

\bibitem[P{\'e}r18]{fullward}
Carlos~I. P{\'e}r{ez-S\'anchez}.
\newblock {The full Ward-Takahashi Identity for colored tensor models}.
\newblock {\em Commun. Math. Phys.}, 358(2):589--632, 2018.
\newblock arXiv:1608.08134.

\bibitem[PPW17]{SDE}
Romain Pascalie, Carlos~I. P{\'e}r{ez-S\'anchez}, and Raimar Wulkenhaar.
\newblock {Correlation functions of $\mathrm{U}(N)$-tensor models and their
  Schwinger-Dyson equations}.
\newblock {\em (to appear in \textit{Ann. de l’Institut Henri Poincaré D,
  Combinatorics, Physics and their Interactions)}} 
\newblock 2017.
\newblock arXiv:1706.07358.

\bibitem[PS98]{PaschkeSitarz}
Mario Paschke and Andrzej Sitarz.
\newblock {Discrete sprectral triples and their symmetries}.
\newblock {\em J. Math. Phys.}, 39:6191--6205, 1998.

\bibitem[PW18]{RaimarPanzer}
Erik Panzer and Raimar Wulkenhaar.
\newblock {Lambert-W solves the noncommutative $\Phi^4$-model}.
\newblock 2018.
\newblock arXiv:1807.02945.

\bibitem[RV06]{RennieVarilly}
Adam Rennie and Joseph~C. Varilly.
\newblock {Reconstruction of manifolds in noncommutative geometry}.
\newblock 2006.
\newblock math/0610418.

\bibitem[Spe19]{SpeicherLectures}
Roland Speicher.
\newblock Video lectures on noncommutative distributions, {Universit\"at des
  Saarlandes}, 2019.
\newblock
  \href{https://www.math.uni-sb.de/ag/speicher/web\_video/ncdistss19/lec2.html}{\color{black}See
  around time 41:30
  \color{blue}\texttt{https://www.math.uni-sb.de/ag/speicher/web\_video/ncdistss19/lec2.html}}
  Accessed 25.August.2019.

\bibitem[SchSt13]{SteinackerFuzzy}
Paul Schreivogl and Harold Steinacker.
\newblock {Generalized Fuzzy Torus and its Modular Properties}.
\newblock {\em SIGMA}, 9:060, 2013.
\newblock arXiv:1305.7479.

\bibitem[SpSt18]{SperlingSteinacker}
Marcus Sperling and Harold Steinacker.
\newblock {Higher spin gauge theory on fuzzy $S^4_N$}.
\newblock {\em J. Phys.}, A51(7):075201, 2018.
\newblock arXiv:1707.00885.

\bibitem[Thu74]{thurston}
William Thurston.
\newblock Foliations and groups of diffeomorphisms.
\newblock {\em Bull. Amer. Math. Soc.}, 80:304--307, 1974.

\bibitem[vS15]{WvSbook}
Walter~D. van Suijlekom.
\newblock {\em {Noncommutative geometry and particle physics}}.
\newblock Mathematical Physics Studies. Springer, Dordrecht, 2015.

\end{thebibliography}
\end{document}